\documentclass[11pt,fleqn,a4paper]{article} 

\usepackage{amsmath,amssymb,amsthm,amsfonts} 
\usepackage[mathscr]{eucal}
\usepackage{cmap}

\usepackage{hyperref}
\hypersetup{colorlinks, linkcolor=blue, citecolor=blue, urlcolor=blue}

\allowdisplaybreaks

\newcommand{\p}{\partial}
\newcommand{\sgn}{\mathop{\rm sgn}\nolimits}
\newcommand{\const}{{\rm const}}
\newcommand{\lsemioplus}{\mathbin{\mbox{$\lefteqn{\hspace{.77ex}\rule{.4pt}{1.2ex}}{\in}$}}}

\newlength{\mylength}
\newcommand{\solution}{\hspace*{-\mylength}\bullet\quad}

\newtheorem{theorem}{Theorem}
\newtheorem{lemma}[theorem]{Lemma}
\newtheorem{corollary}[theorem]{Corollary}

{\theoremstyle{definition}

\newtheorem{remark}[theorem]{Remark}
}

\newcommand{\todo}[1][\null]{\ensuremath{\clubsuit}}

\newcommand{\noprint}[1]{}

\begin{document}

\par\noindent {\LARGE\bf
Lie reductions and exact solutions\\
of dispersionless Nizhnik equation
\par}

\vspace{5mm}\par\noindent{\large
Oleksandra O.\ Vinnichenko$^\dag$, Vyacheslav M.\ Boyko$^{\dag\ddag}$ and Roman O.\ Popovych$^{\dag\S}$
}

\vspace{5mm}\par\noindent{\it\small
$^\dag$\,Institute of Mathematics of NAS of Ukraine, 3 Tereshchenkivska Str., 01024 Kyiv, Ukraine
\par}

\vspace{2mm}\par\noindent{\it\small
$^\ddag$\,Department of Mathematics, Kyiv Academic University, 36 Vernads'koho Blvd, 03142 Kyiv, Ukraine
\par}

\vspace{2mm}\par\noindent{\it\small
$^\S$\,Mathematical Institute, Silesian University in Opava, Na Rybn\'\i{}\v{c}ku 1, 746 01 Opava, Czech Republic
\par}

\vspace{5mm}\par\noindent{\small
E-mails:
oleksandra.vinnichenko@imath.kiev.ua,
boyko@imath.kiev.ua,
rop@imath.kiev.ua
\par}

\vspace{8mm}\par\noindent\hspace*{10mm}\parbox{140mm}{\small
We exhaustively classify the Lie reductions of the real dispersionless Nizhnik equation
to partial differential equations in two independent variables
and to ordinary differential equations.
Lie and point symmetries of reduced equations are comprehensively studied,
including the analysis of which of them correspond to hidden symmetries of the original equation.
If necessary, associated Lie reductions of a nonlinear Lax representation
of the dispersionless Nizhnik equation are carried out as well.
As a result, we construct wide families of new invariant solutions of this equation
in explicit form in terms of elementary, Lambert and hypergeometric functions
as well as in parametric or implicit form.
We show that Lie reductions to algebraic equations lead to
no new solutions of this equation in addition to the constructed ones.
Multiplicative separation of variables  is used for illustrative construction
of non-invariant solutions.

}\par\vspace{4mm}

\noprint{
Keywords:
dispersionless Nizhnik equation;
Lie reduction;
invariant solutions;
point-symmetry pseudogroup;
Lie invariance algebra;
discrete symmetry;
hidden symmetries;

MISC: 35B06 (Primary) 17B80, 35C05, 35C06, 35A30  (Secondary)
17-XX   Nonspeculative rings and algebras
 17Bxx	 Lie algebras and Lie superlogical {For Lie groups, see 22Exx}
  17B80   Applications of Lie algebras and superlogical to integrable systems
35-XX   Partial differential equations
  35A30   Geometric theory, characteristics, transformations [See also 58J70, 58J72]
  35B06   Symmetries, invariants, etc.
 35Cxx  Representations of solutions
  35C05   Solutions in closed form
  35C06   Self-similar solutions
}

\section{Introduction}

Lie reduction is the main way to use Lie symmetries for finding exact solutions
of partial differential equations \cite{blum2009A,blum1989A,boch1999A,olve1993A,ovsi1982A}.
Since the Lie invariance algebras of models studied in mathematical physics
are usually wide enough,
it is also the most universal way for constructing exact solutions of such models in general
and, especially, of nonlinear ones.
Many papers devoted to this subject were published for several last decades
but correct and comprehensive studies of Lie reductions and the corresponding reduced systems
for specific systems of partial differential equations are rather exceptional,
especially in the case of more than two independent variables,
see, e.g.,~\cite{andr1998A,cham1988a,davi1986a,fush1994a,fush1994b,kont2019a,kova2023b,malt2024a,mart1989a,olve1993A,poch2017a,vane2021a},
the result collections~\cite{CRC_v2,CRC_v1} and references therein
for particular examples.

In the present paper,
we exhaustively classify the Lie reductions of the real dispersionless Nizhnik equation
\begin{gather}\label{eq:dN}
u_{txy}=(u_{xx}u_{xy})_x + (u_{xy}u_{yy})_y
\end{gather}
to partial differential equations in two independent variables
and to ordinary differential equations.
We use the above name for the equation~\eqref{eq:dN}
since it is the dispersionless counterpart of the (real potential symmetric) Nizhnik equation
for the (real) Nizhnik system~\cite[Eq.~(4)]{nizh1980a};
see~\cite[footnote~3]{boyk2024a} for a detail discussion of choosing the name.
The other names are the dispersionless Nizhnik--Novikov--Veselov equation~\cite[Eq.~(63)]{kono2002b}
or even the dispersionless Novikov--Veselov equation, see, e.g., \cite[Eq.~(5)]{pavl2006a} and \cite[Eq.~(1)]{moro2021a}
although the proper Novikov--Veselov counterpart of~\eqref{eq:dN},
which derived in~\cite[Eq.~(30)]{kono2004c} and~\cite[Eq.~(32)]{kono2004b} as a model of nonlinear geometrical optics,
is of different form.

Using the known Lax representation for the Nizhnik system \cite{nizh1980a}
or, equivalently, for the (potential symmetric) Nizhnik equation
and a technique of limit transitions to dispersionless counterparts of
(1+2)-dimensional integrable differential equations and of the corresponding Lax representations,
which was suggested in~\cite[p.~167]{zakh1994a},
one can derive
a nonlinear Lax representation of the dispersionless Nizhnik equation~\eqref{eq:dN},
\begin{gather}\label{eq:dNLaxPair}
v_t=\frac13\left(v_x^3-\frac{u_{xy}^3}{v_x^3}\right)+u_{xx}v_x-\frac{u_{xy}u_{yy}}{v_x},\quad
v_y=-\frac{u_{xy}}{v_x},
\end{gather}
see~\cite{pavl2006a}.
This representation can be converted to the associated linear nonisospectral Lax representation
\begin{gather*}
\chi_t=(p^2+p^{-4}u_{xy}^{\,\,3}+u_{xx}+p^{-2}u_{xy}u_{yy})\chi_x
-(pu_{xxx}-p^{-1}(u_{xy}u_{yy})_x-p^{-3}u_{xy}^{\,\,2}u_{xxy})\chi_p,\\
\chi_y=p^{-2}u_{xy}\chi_x+p^{-1}u_{xxy}\chi_p,
\end{gather*}
where $p$ is a variable spectral parameter, $\chi=\chi(t,x,y,p)$ and  $u=u(t,x,y)$;
see, e.g., \cite[p.~360]{serg2018a} and references therein for the corresponding procedure.

The study of the equation~\eqref{eq:dN} within the framework of classical symmetry analysis
was initiated in~\cite{moro2021a}.
Therein, the representation~$\mathfrak g_{\rm e}$ of the contact invariance algebra of~\eqref{eq:dN}
by vector fields in the evolution form was computed,
and it turned out to coincide with that of the Lie (point) invariance algebra of~\eqref{eq:dN}.
The corresponding point (resp.\ contact) symmetry group was given with mistakes and omissions.
One-dimensional subalgebras of~$\mathfrak g_{\rm e}$ that are appropriate for Lie reduction
were classified with a minor deficiency.
The corresponding Lie reductions of~\eqref{eq:dN} to partial differential equations
with two independent variables and further Lie reductions of these equations were performed.
Wide families of solutions that are polynomial in~$(x,y)$ were constructed
as examples of non-invariant solutions.
Second-order cosymmetries of the equation~\eqref{eq:dN} were found.
Since all of them are conservation-law characteristics of this equation,
the associated conserved currents were computed as well.
At the same time, it was not studied which Lie symmetries of reduced equations
are induced by Lie symmetries of the original equation~\eqref{eq:dN},
and thus a number of presented two-step reductions are in fact needless.
Among obtained Lie-invariant solutions of~\eqref{eq:dN}, there are many equivalent to each other
or those containing typos, which makes them incorrect.
Careful analysis of reduced ordinary differential equations shows
that more of their closed-form solutions can be constructed,
and one should take into account the degeneracy of some of these equations.

In the present paper, we correct, enhance and significantly extend results from~\cite{moro2021a}.
We scrupulously carry out each step of the optimized procedure of comprehensive Lie reduction
for the dispersionless Nizhnik equation~\eqref{eq:dN},
which results in finding wide families of new invariant solutions of~\eqref{eq:dN}
in explicit form in terms of elementary, Lambert and hypergeometric functions
as well as in parametric or implicit form.
An accurate description of this procedure
for the case of a system of partial differential equations with three independent variables,
which is relevant to the equation~\eqref{eq:dN},
is presented for the first time in Section~\ref{sec:LieReductionProcedure}
along with a number of related comments. 

The first step of the Lie reduction procedure for the equation~\eqref{eq:dN}
was in fact implemented in~\cite{boyk2024a} (cf.\ \cite{moro2021a}),
where we in particular computed
the maximal Lie invariance algebras~$\mathfrak g$ and~$\mathfrak g_{\rm L}$
of the equation~\eqref{eq:dN} and its nonlinear Lax representation~\eqref{eq:dNLaxPair}
as well as their point-symmetry pseudogroups~$G$ and~$G_{\rm L}$, respectively,
and performed a preliminary study of these algebras and pseudogroups.
Since the above results are used throughout the present paper,
for convenience we review them in Sections~\ref{sec:LieInvAlgebra} and~\ref{sec:PointSymGroup}.
One- and two-dimensional subalgebras of the algebra~$\mathfrak g$
and one-dimensional subalgebras of the algebra~$\mathfrak g_{\rm L}$
are classified in Section~\ref{sec:dNClassificationOfSubalgs}
up to the $G$- and $G_{\rm L}$-equivalences, respectively.
The constructed optimal lists of one- and two-dimensional subalgebras of~$\mathfrak g$
create a basis for the efficient and exhaustive fulfilment of Lie reductions of the equation~\eqref{eq:dN}
to partial differential equations in two independent variables in Section~\ref{sec:dNLieReductionsOfCodim1}
and to ordinary differential equations in Section~\ref{sec:dNLieReductionsOfCodim2}.
In Section~\ref{sec:dNTrivialSolutions}, we discuss trivial solutions of the equation~\eqref{eq:dN},
which we exclude from the further consideration.
We show in Section~\ref{sec:dNLieReductionsOfCodim3} that
Lie reductions of the equation~\eqref{eq:dN} to algebraic equations
give no new solutions of this equation in comparison with the above reductions to differential equations.
In Section~\ref{sec:dNMultiplicativeSeparationOfVars},
we use multiplicative separation of variables to present
an example of finding non-Lie solutions of the equation~\eqref{eq:dN}
that generalize invariant solutions.
Section~\ref{sec:Conclusion} is devoted to a comprehensive discussion of the obtained results
and their implications.

We compute for the first time point symmetry groups of reduced equations,
including their discrete point symmetries,
and check whether these symmetries are hidden or induced.
Since most of the reduced equations to be considered are quite cumbersome,
various versions of the algebraic method by Hydon \cite{hydo1998a,hydo1998b,hydo2000b}
are much more efficient in the course of the above computation than the direct method.
In addition, some of the reduced equations of the equation~\eqref{eq:dN} are not of maximal rank,
and the study of Lie and general point symmetries of differential equations
that are not of maximal rank is also carried out for the first time in the present paper.
We also make deeper analysis of reduced equations than
in most papers in the field of classical group analysis,
construct more solutions for more reduced equations
and more systematically study hidden symmetries of the original equation.
For integrating some of reduced ordinary differential equations for the equation~\eqref{eq:dN},
we involve the associated Lie reductions of the nonlinear Lax representation~\eqref{eq:dNLaxPair}.

In the course of performing the Lie reduction procedure for the equation~\eqref{eq:dN},
we observe several interesting phenomena.
Thus, not all parameters of a family of inequivalent subalgebras are necessarily inherited
by the corresponding reduced equations.
The limit case for this phenomenon is
when all inequivalent subalgebras from a family even parameterized by arbitrary functions
correspond, under an appropriate choice of ansatzes, to the same reduced equation.
Another display of this phenomenon is the possibility
of mapping a class of reduced equations to its proper subclass, which has a less number of parameters.
Some equivalent (two-dimensional) subalgebras of the algebra~$\mathfrak g$
with a nonzero (one-dimensional) intersection
induce inequivalent (one-dimensional) subalgebras
of the maximal Lie invariance algebra of a reduced partial differential equation
obtained by the Lie reduction with respect to the intersection.
The algebra~$\mathfrak g$ is embedded in the algebra~$\mathfrak g_{\rm L}$
via extending the vector fields from~$\mathfrak g$ to the dependent variable~$v$
of the nonlinear Lax representation~\eqref{eq:dNLaxPair},
and thus any Lie reduction of the equation~\eqref{eq:dN}
has a counterpart among Lie reductions of the system~\eqref{eq:dNLaxPair}
but such a counterpart is in general not unique even up to the $G_{\rm L}$-equivalence.
In contrast to Lie symmetries, simple and obvious discrete point symmetries of the equation~\eqref{eq:dN}
induce, even under the optimal choice of ansatzes,
complicated and nontrivial discrete point symmetries of the corresponding reduced equations.

For readers' convenience,
we marked the constructed solutions of the dispersionless Nizhnik equation~\eqref{eq:dN} by the bullet symbol~$\bullet$\,.

\section{Optimized procedure of Lie reduction}\label{sec:LieReductionProcedure}

Despite many papers devoted to the construction of exact solutions of systems of partial differential equations
using the Lie reduction procedure,
the number of papers with correct, proper and systematic studies of Lie reductions
for particular important systems modeling real-world phenomena is not as large as it could be expected.
Such studies involve cumbersome computations and requires an accurate consideration of many inequivalent cases.
Hence a precondition of successfully performing the above procedure is its optimization.

To be specific, we describe the \emph{optimized procedure of Lie reduction} for the case of three independent variables,
which is relevant to the present paper.
Given a system~$\mathcal L$ of partial differential equations for unknown functions~$u$ in three independent variables,
this procedure consists of the following steps;
see also further comments after the procedure's description.
\begin{enumerate}\itemsep=0ex

\item\label{item:MIA&PointSymGroup}
Compute the maximal Lie invariance (pseudo)algebra~$\mathfrak g$
and the point symmetry (pseu\-do)group~$G$ of~$\mathcal L$.

\item\label{item:ClassificationOf1And2DSubalgs}
Construct complete lists of $G$-inequivalent one- and two-dimensional subalgebras of~$\mathfrak g$
and select those among them that are appropriate for using within the framework of Lie reduction.

\item\label{item:Codim1LRs}
\emph{Lie reductions of codimension one}.
For each of the selected one-dimensional subalgebras of~$\mathfrak g$, say $\mathfrak s_1$,
find an ansatz for the $\mathfrak s_1$-invariant solutions of~$\mathcal L$
such that the corresponding reduced system~$\hat{\mathcal L}_1$ of partial differential equations in two independent variables
is of the simplest or most convenient form.
If the system~$\hat{\mathcal L}_1$ can be completely integrated
or its general solution is expressed in terms of the general solution
of a system that has been well studied within the framework of symmetry analysis,
then the further consideration of the system~$\hat{\mathcal L}_1$
and carrying out the Lie reductions of~$\mathcal L$ with respect to subalgebras of~$\mathfrak g$
containing, up to $G$-equivalence, the subalgebra~$\mathfrak s_1$ are needless.

\item\label{item:HiddenSyms}
Otherwise, compute
the normalizer~${\rm N}_{\mathfrak g}(\mathfrak s_1)$ of~$\mathfrak s_1$ in~$\mathfrak g$,
the stabilizer~${\rm St}_G(\mathfrak s_1)$ of~$\mathfrak s_1$ in~$G$,
the maximal Lie invariance algebra~$\hat{\mathfrak g}_1$ and the point symmetry group~$\hat G_1$ of~$\hat{\mathcal L}_1$
as well as the subalgebra~$\tilde{\mathfrak g}_1$ of~$\hat{\mathfrak g}_1$ and the subgroup~$\tilde G_1$ of~$\hat G_1$
that are induced by elements of~${\rm N}_{\mathfrak g}(\mathfrak s_1)$ and~${\rm St}_G(\mathfrak s_1)$, respectively.
Perform the Lie reduction procedure for the system~$\hat{\mathcal L}_1$
only if $\hat{\mathfrak g}_1\ne\tilde{\mathfrak g}_1$ or at least $\hat G_1\ne\tilde G_1$,
see comments below.

\item\label{item:Codim2LRs}
\emph{Lie reductions of codimension two}.
For each of the two-dimensional subalgebras of~$\mathfrak g$
that have passed the selection in steps~\ref{item:ClassificationOf1And2DSubalgs} and~\ref{item:Codim1LRs},
say $\mathfrak s_2$, find an ansatz for the $\mathfrak s_2$-invariant solutions of~$\mathcal L$
such that the corresponding reduced system~$\hat{\mathcal L}_2$ of ordinary differential equations
is of the simplest or most convenient form.

\item\label{item:Codim2LRsSolutionEquiv}
Compute
the normalizer~${\rm N}_{\mathfrak g}(\mathfrak s_2)$ of~$\mathfrak s_2$ in~$\mathfrak g$,
the stabilizer~${\rm St}_G(\mathfrak s_2)$ of~$\mathfrak s_2$ in~$G$,
the maximal Lie invariance algebra~$\hat{\mathfrak g}_2$ and the point symmetry group~$\hat G_2$ of~$\hat{\mathcal L}_2$
as well as the subalgebra~$\tilde{\mathfrak g}_2$ of~$\hat{\mathfrak g}_2$ and the subgroup~$\tilde G_2$ of~$\hat G_2$
that are induced by elements of~${\rm N}_{\mathfrak g}(\mathfrak s_2)$ and~${\rm St}_G(\mathfrak s_2)$, respectively.

\item\label{item:Codim2LRsIntegration}
Construct, if possible, the general solution of~$\hat{\mathcal L}_2$
or at least some particular solutions of~$\hat{\mathcal L}_2$.
Use transformations from the group~$\tilde G_2$
for gauging integration constants in the constructed solutions.
Substitute the arranged solutions into the ansatz for the $\mathfrak s_2$-invariant solutions,
which gives $G$-inequivalent solutions of the original system~$\mathcal L$.

\item\label{item:LRofCodim3}
\emph{Lie reductions of codimension three}.
Analyze whether there are Lie reductions of~$\mathcal L$ with respect to three-dimensional subalgebras of~$\mathfrak g$
to algebraic equations that lead to new exact solutions of~$\mathcal L$
in comparison with those constructed in the previous steps
using Lie reductions of codimensions one and~two.
If this is the case, then carry out all such $G$-inequivalent Lie reductions.
\end{enumerate}

In steps~\ref{item:MIA&PointSymGroup}, \ref{item:HiddenSyms} and~\ref{item:Codim2LRsSolutionEquiv},
it is convenient to carry out the computation of the corresponding point symmetry groups by
a version of the algebraic method,
the automorphism-based version \cite{hydo1998a,hydo1998b,hydo2000b,hydo2000A}
(see also further examples, e.g., in \cite{kont2019a,vane2021a})
or one of the various modifications of the megaideal-based version \cite{bihl2011b,card2013a,card2021a,malt2024a,opan2020a}
in the case of finite or infinite dimension of the associated maximal Lie invariance algebra, respectively.
The direct method~\cite{bihl2011b,king1998a} may be advantageous for those systems
that belong to classes of systems of differential equations
with known restrictions for point symmetries of their elements, see \cite{bihl2011b,kova2023b,kova2023a}.
Systems that are not of maximal rank, which are not too uncommon among reduced systems of differential equations,%
\footnote{%
As a simple illustrative example, consider
the quadratic porous medium (Boussinesq) equation $u_t=(uu_x)_x$ for the groundwater pressure~$u$,
which describes unsteady flows of groundwater with the presence of a free surface,
and its first- and second-level potential equations $v_t=v_xv_{xx}$ and $w_t=\frac12(w_{xx})^2$.
Each of these equations admits the one-parameter group of shifts with respect to~$t$ with the generator~$\p_t$
as its Lie symmetry group.
The corresponding invariant solutions are just stationary solutions,
and the associated ansatzes $u=\varphi(\omega)$, $v=\psi(\omega)$ and  $w=\zeta(\omega)$ with $\omega=x$
respectively reduce these equations to the ordinary differential equations
$\varphi\varphi_{\omega\omega}=0$, $\psi_\omega\psi_{\omega\omega}=0$ and $(\zeta_{\omega\omega})^2=0$,
which are not of maximal rank.
In particular, the last reduced equation $(\zeta_{\omega\omega})^2=0$
is not of maximal rank on the entire set of its solutions.
}
require a specific study, which complicates the consideration,
see Section~\ref{sec:dNLieReductionsOfCodim2Collection2}.

A subalgebra~$\mathfrak s$ of~$\mathfrak g$ is appropriate for using within the framework of Lie reduction
if and only if satisfies the local transversality condition.
In fixed local coordinates, this condition is equivalent to the equality of the ranks of the matrices
that are respectively constituted by all the components of basis vector fields of~$\mathfrak s$
and by solely those corresponding to the independent variables.
In step~\ref{item:ClassificationOf1And2DSubalgs},
one can classify merely one- and two-dimensional appropriate subalgebras of~$\mathfrak g$
but, in general, this does not lead to a significant simplification
in comparison with the complete classification and the further selection of appropriate subalgebras.
Usually, subalgebras of~$\mathfrak g$ are classified
up to their equivalence generated by the group ${\rm Inn}(\mathfrak g)$ of inner automorphisms of~$\mathfrak g$,
which coincides with the $G_{\rm id}$-equivalence,
where $G_{\rm id}$ is the identity component of~$G$.
At the same time, it is advantageous to use the stronger $G$-equivalence instead of the $G_{\rm id}$-equivalence
since it allows one to reduce the list of subalgebras to be considered.
Moreover, this makes the Lie reduction procedure consistent
with the natural $G$-equivalence on the solution set of the system~$\mathcal L$.

Particular attention in steps~\ref{item:Codim1LRs} and~\ref{item:Codim2LRs}
should be paid to the optimal choice of ansatzes \cite{fush1994a,fush1994b,poch2017a,popo1995b}.
Given a subalgebra~$\mathfrak s$ of~$\mathfrak g$,
an $\mathfrak s$-invariant ansatz is defined up to an arbitrary point transformation of invariant variables.
In other words, there is an infinite family of $\mathfrak s$-invariant ansatzes,
and selecting a proper representative in this family usually leads to an essential simplification
of the corresponding reduced system and its further study.
The simplicity of the form of reduced systems
and its certain similarity to the form of the original system~$\mathcal L$
do not exhaust possible criteria in the course of selecting ansatzes.
Another criterion is to unify the form of reduced systems for a subset of listed families
of $G$-inequivalent subalgebras of~$\mathfrak g$
for embedding them into a nice superclass of differential equations and unifying their study.
After reducing the system~$\mathcal L$ using a preliminary ansatz,
one can improve the form of the obtained reduced system by a point transformation of invariant variables
and then optimize the ansatz by means of combining it with this transformation.
At the same time, such transformations may significantly complicate the form of ansatzes.
To preserve the balance between the simplicity of ansatzes and the simplicity of the corresponding reduced systems,
sometimes it is necessary to transform ansatzes only partially.

Elements of optimal lists of subalgebras of~$\mathfrak g$ can in general be not only single subalgebras
but also families of subalgebras parameterized by arbitrary constants
or, if the algebra~$\mathfrak g$ is infinite-dimensional, even by arbitrary functions.
Lie reductions of the system~$\mathcal L$ with respect to subalgebras from such a family
result in a class~$\mathcal C$ of reduced systems with subalgebra parameters as its arbitrary elements
rather than in a collection of single reduced systems.
Thus, the study of Lie symmetries for systems from the class~$\mathcal C$ should be realized
as the solution to the group classification problem for this class.

We would like to emphasize that
further Lie reductions of a reduced system of partial differential equations with two independent variables
in step~\ref{item:Codim1LRs} should be carried out only if this system admits point symmetries
that are not induced by point symmetries%
\footnote{%
The induction of Lie symmetries of a reduced system by Lie symmetries of the original system
was first discussed in \cite[Section~20.4]{ovsi1982A}.
}
of the original system~$\mathcal L$ and thus called \emph{hidden point symmetries}%
\footnote{%
In this context, the term \emph{hidden symmetries} was first used in~\cite{yeho2004a}.
Other terms for this notion in the literature are
\emph{additional} \cite[Example~3.5]{olve1993A},
\emph{non-induced}~\cite{fush1994a,fush1994b} or
\emph{Type-II hidden}~\cite{abra2006a,abra2006b} symmetries.
The first example of such symmetries was given in~\cite{kapi1978a}
but become known after its discussion in \cite[Example~3.5]{olve1993A}.
A~systematic study of them is rather seldom
and has been carried out only for a few famous systems of differential equations,
in particular, for
the Navier--Stokes equations describing flows of an incompressible viscous fluid~\cite{fush1994a,fush1994b},
the (1+1)-dimensional generalized Burgers equations $u_t+uu_x+f(t,x)u_{xx}=0$~\cite{poch2017a},
the two-dimensional degenerate Burgers equation $u_t+uu_x-u_{yy}=0$~\cite{vane2021a},
the Boiti--Leon--Pempinelli system~\cite{malt2024a},
the (1+2)-dimensional ultraparabolic Fokker--Planck equation $u_t+xu_y=u_{xx}$~\cite{vane2021a}
as well as the dispersionless Nizhnik equation in the present paper.
Interesting particular examples of hidden symmetries of several hydrodynamic models
were presented in~\cite[Chapter~1]{andr1998A}.
}
of~$\mathcal L$ associated with the codimension-one reduction under consideration.
See the description of the optimized procedure of step-by-step reductions with involving hidden symmetries
in \cite[Section~B]{kova2023b}.
The study of Lie and general point symmetries of the derived reduced systems,
identifying hidden symmetries of~$\mathcal L$ among them
and using such hidden symmetries for further Lie reduction of the corresponding reduced systems
is a necessary part of the comprehensive Lie reduction procedure.

\emph{It is useless and counterproductive to consider the other step-by-step Lie reductions,
whose second steps are based on induced symmetries of reduced systems.}
There are at least two sources of inconveniences in the course of such reductions,
which implicitly lead to the consideration of multiple essentially equivalent reductions.

To make the first source evident, consider a particular case, where the system~$\mathcal L$
admits two commuting Lie-symmetry vector fields~$Q^1$ and~$Q^2$ such that
the subalgebras $\mathfrak s_1^\mu:=\langle Q^1+\mu Q^2\rangle$ of~$\mathfrak g$
parameterized by an arbitrary constant $\mu$ are pairwise $G$-inequivalent
and each of them satisfies the local transversality condition and is thus appropriate for Lie reduction of~$\mathcal L$.
It is obvious that for any~$\mu$, the vector field~$Q^2$ belongs to the normalizer of~$\mathfrak s_1^\mu$ in~$\mathfrak g$.
Therefore, it induces a Lie-symmetry vector field~$\hat Q^{2,\mu}$
of the reduced system~$\mathcal L_1^\mu$ for $\mathfrak s_1^\mu$-invariant solutions of~$\mathcal L$.
Suppose that the algebra~$\langle\hat Q^{2,\mu}\rangle$ also satisfies the local transversality condition.
Thus, we have the infinite family of two-step reductions,
where for each~$\mu$,
the system~$\mathcal L$ is first reduced to the system~$\mathcal L_1^\mu$ using the algebra $\mathfrak s_1^\mu$
and then the system~$\mathcal L_1^\mu$ is further reduced using the algebra $\langle\hat Q^{2,\mu}\rangle$.
Moreover, the first steps of these reductions are definitely not equivalent to each other.
Nevertheless, each of these two-step reductions is equivalent to
the same one-step Lie reduction of the system~$\mathcal L$
with respect to the two-dimensional subalgebra $\langle Q^1,Q^2\rangle$ of~$\mathfrak g$.
For invariance algebras of more complicated structure,
equivalences between multi-step reductions are in general not so obvious,
and establishing them requires an additional analysis.

\looseness=1
The second source is that $\hat G_1$-inequivalent (one-dimensional) subalgebras of the maximal Lie invariance algebra~$\mathfrak a_1$
of a reduced system~$\hat{\mathcal L}_1$ of partial differential equations
may correspond to equivalent (two-dimensional) subalgebras of~$\mathfrak g$;
recall that by $\hat G_1$ we denote the point symmetry group of~$\hat{\mathcal L}_1$,
see this and other related notations in the above description of the optimized procedure of Lie reduction.
More specifically, suppose that the system~$\hat{\mathcal L}_1$ is obtained by the Lie reduction
of the original system~$\mathcal L$ with respect to a one-dimensional subalgebra~$\mathfrak s_1=\langle Q^0\rangle$ of~$\mathfrak g$,
and vector fields~$Q^1$ and~$Q^2$ belong to the normalizer~${\rm N}_{\mathfrak g}(\mathfrak s_1)$ of~$\mathfrak s_1$ in~$\mathfrak g$,
thus inducing elements~$\hat Q^1$ and~$\hat Q^2$ of~$\mathfrak a_1$.
In addition, suppose that the subalgebras~$\langle Q^0,Q^1\rangle$ and~$\langle Q^0,Q^2\rangle$
are $G$-equivalent and the equivalence is established only by a transformation~$\Phi\in G$
that does not belong to the stabilizer~${\rm St}_G(\mathfrak s_1)$ of~$\mathfrak s_1$ in~$G$.
Then the transformation~$\Phi$ does not induce a point symmetry of~$\hat{\mathcal L}_1$,
and thus the subalgebras~$\langle\hat Q^1\rangle$ and~$\langle\hat Q^2\rangle$ of~$\mathfrak a_1$
are in general $\hat G_1$-inequivalent.
In terms of reductions and solutions, this means that
the inequivalent two-step Lie reductions with the first step using the subalgebra~$\mathfrak s_1$ of~$\mathfrak g$
and the second step using the subalgebras~$\langle\hat Q^1\rangle$ and~$\langle\hat Q^2\rangle$ of~$\mathfrak a_1$
result in the $G$-equivalent families
of the $\langle Q^0,Q^1\rangle$- and the $\langle Q^0,Q^2\rangle$-invariant solutions, respectively.
See Remark~\ref{rem:OnInessEquivOfMultistepReductions} for a nontrivial example of the described situation,
which arises in the course of studying Lie reductions of the dispersionless Nizhnik equation~\eqref{eq:dN}.

\emph{This is why the best strategy is to completely avoid multi-step reductions not involving hidden symmetries.} 

We do not include the classification of three-dimensional subalgebras of~$\mathfrak g$
in step~\ref{item:ClassificationOf1And2DSubalgs}
since in general, it is a much more complicated problem than those for dimensions one and two
and it is not required for the Lie reduction procedure in its entity.
Only a small number of three-dimensional subalgebras satisfy the selection criterion
from step~\ref{item:LRofCodim3}, if they exist at all.
This is why it is better to merely classify the selected subalgebras directly in step~\ref{item:LRofCodim3}.
For example, the maximal Lie invariance algebra of the dispersionless Nizhnik equation~\eqref{eq:dN}
contains no three-dimensional subalgebras that are appropriate to step~\ref{item:LRofCodim3}.
In a similar way, we may also consider two-dimensional subalgebras of~$\mathfrak g$
but the reached simplification is not essential in comparison with their complete classification.

The system~$\mathcal L$ can possess families of trivial or obvious solutions
that can be easily guessed without applying Lie reduction or other methods.
Moreover, these families can contain solutions that are invariant with respect to subalgebras of~$\mathfrak g$
whose dimensions are greater than or equal to the number of independent variables,
and thus such solutions can repeatedly arise
in the course of performing the Lie reduction procedure for the system~$\mathcal L$.
It is beneficial to find these families of solutions before step~\ref{item:Codim1LRs}
and exclude their elements under the further listing of solutions.
Similar solution families can be constructed in step~\ref{item:Codim1LRs}
and should be treated analogously.
Section~\ref{sec:dNTrivialSolutions}, the solution family~\eqref{eq:dNs1.3Rho1InvarSolutions}
and the treatment of trivial solutions in Section~\ref{sec:dNLieReductionsOfCodim2}
illustrate the above remark.

In addition to classical integration methods,
a number of other techniques can be applied to finding exact solutions
of a reduced system~$\mathcal R$ of ordinary differential equations.
An obvious approach is to use Lie symmetries of the system~$\mathcal R$ for at least lowering its order,
see Section~\ref{sec:dNLieReductionsOfCodim2Collection2}.
One can try to construct first integrals of~$\mathcal R$ by means of the direct method~\cite{anco2002a,anco2002b}
supposing a certain ansatz for the associated integrating multipliers,
see~\cite[Section~6]{malt2024a} for the application of this technique
to reduced systems of ordinary differential equations for the Boiti--Leon--Pempinelli system.
One can also look for objects that are related to the original system~$\mathcal L$
within the framework of symmetry analysis of differential equations
and induce analogous objects for the system~$\mathcal R$.
These objects include not only Lie and general point symmetries,
first integrals and integrating multipliers but also
Lagrangian and Hamiltonian structures and (linear and nonlinear) Lax representations.
Induced objects can then be involved in obtaining exact solutions of~$\mathcal R$.
See, e.g., Section~\ref{sec:dNLieReductionsOfCodim2Collection2}
for using induced nonlinear Lax representations.

\section{Structure of Lie invariance algebra}\label{sec:LieInvAlgebra}

The maximal Lie invariance (pseudo)algebra~$\mathfrak g$ of the dispersionless Nizhnik equation~\eqref{eq:dN}
is infinite-dimensional and is spanned by the vector fields
\begin{gather}\label{eq:dNMIA}
\begin{split}&
D^t(\tau)=\tau\p_t+\tfrac13\tau_tx\p_x+\tfrac13\tau_ty\p_y-\tfrac1{18}\tau_{tt}(x^3+y^3)\p_u,\quad
D^{\rm s}=x\p_x+y\p_y+3u\p_u,\\[.5ex] &
P^x(\chi)=\chi\p_x-\tfrac12\chi_tx^2\p_u,\quad
P^y(\rho)=\rho\p_y-\tfrac12\rho_ty^2\p_u,\\[.5ex] &
R^x(\alpha)=\alpha x\p_u,\quad
R^y(\beta)=\beta y\p_u,\quad
Z(\sigma)=\sigma\p_u,
\end{split}
\end{gather}
where $\tau$, $\chi$, $\rho$, $\alpha$, $\beta$ and $\sigma$ run through the set of smooth functions of~$t$,
cf.~\cite{moro2021a}.
Moreover, the contact invariance algebra~$\mathfrak g_{\rm c}$ of the equation~\eqref{eq:dN} coincides with
the first prolongation~$\mathfrak g_{(1)}$ of the algebra~$\mathfrak g$,
and generalized symmetries of this equation at least up to order five are exhausted,
modulo the equivalence of generalized symmetries, by its Lie symmetries.

Up to the antisymmetry of the Lie bracket,
the nonzero commutation relations between the vector fields~\eqref{eq:dNMIA} spanning~$\mathfrak g$
are exhausted by
\begin{gather}\label{eq:dNCommRelations}
\begin{split}
&[D^t(\tau^1),D^t(\tau^2)]=D^t(\tau^1\tau^2_t-\tau^1_t\tau^2),\\[.5ex]
&[D^t(\tau),P^x(\chi)]=P^x\big(\tau\chi_t-\tfrac13\tau_t\chi\big),\quad
[D^t(\tau),P^y(\rho)]=P^y\big(\tau\rho_t-\tfrac13\tau_t\rho\big),\\[.5ex]
&[D^t(\tau),R^x(\alpha)]=R^x\big(\tau\alpha_t+\tfrac13\tau_t\alpha\big),\quad
[D^t(\tau),R^y(\beta)]=R^y\big(\tau\beta_t+\tfrac13\tau_t\beta\big),\\[.5ex]
&[D^t(\tau),Z(\sigma)]=Z(\tau\sigma_t),\\[.5ex]
&[D^{\rm s},P^x(\chi)]=-P^x(\chi),\quad
[D^{\rm s},P^y(\rho)]=-P^y(\rho),\\[.5ex]
&[D^{\rm s},R^x(\alpha)]=-2R^x(\alpha),\quad
[D^{\rm s},R^y(\beta)]=-2R^y(\beta),\quad
 [D^{\rm s},Z(\sigma)]=-3Z(\sigma),\\[.5ex]
&[P^x(\chi^1),P^x(\chi^2)]=-R^x(\chi^1\chi^2_t-\chi^1_t\chi^2),\quad
 [P^y(\rho^1),P^y(\rho^2)]=-R^y(\rho^1\rho^2_t-\rho^1_t\rho^2),\\[.5ex]
&[P^x(\chi),R^x(\alpha)]=Z(\chi\alpha),\quad
[P^y(\rho),R^y(\beta)]=Z(\rho\beta).
\end{split}
\end{gather}

The maximal Lie invariance (pseudo)algebra~$\mathfrak g_{\rm L}$ of the system~\eqref{eq:dNLaxPair}
is obtained from the algebra~$\mathfrak g$ in a predictable way.
Each vector field from~$\mathfrak g$ is extended to the additional dependent variable~$v$,
and supplementing the extended algebra with the vector field~$\bar P^v:=\p_v$ leads to the algebra~$\mathfrak g_{\rm L}$.
Thus, the latter algebra is spanned by the vector fields
\begin{gather*}
\begin{split}&
\bar D^t(\tau)=\tau\p_t+\tfrac13\tau_tx\p_x+\tfrac13\tau_ty\p_y-\tfrac1{18}\tau_{tt}(x^3+y^3)\p_u,\quad
\bar D^{\rm s}=x\p_x+y\p_y+3u\p_u+\tfrac32v\p_v,\\ &
\bar P^x(\chi)=\chi\p_x-\tfrac12\chi_tx^2\p_u,\quad
\bar P^y(\rho)=\rho\p_y-\tfrac12\rho_ty^2\p_u,\\ &
\bar R^x(\alpha)=\alpha x\p_u,\quad
\bar R^y(\beta)=\beta y\p_u,\quad
\bar Z(\sigma)=\sigma\p_u,\quad
\bar P^v=\p_v,
\end{split}
\end{gather*}
where $\tau$, $\chi$, $\rho$, $\alpha$, $\beta$ and~$\sigma$ are again arbitrary smooth functions of~$t$.
Although the vector fields~$\bar D^t(\tau)$, $\bar P^x(\chi)$, $\bar P^y(\rho)$, $\bar R^x(\alpha)$, $\bar R^y(\beta)$
and $\bar Z(\sigma)$ from the algebra~$\mathfrak g_{\rm L}$
are formally of the same form
as the vector fields~$D^t(\tau)$, $P^x(\chi)$, $P^y(\rho)$, $R^x(\alpha)$, $R^y(\beta)$ and $Z(\sigma)$ from the algebra~$\mathfrak g$,
in fact they are defined in a different space.
Up to the antisymmetry of the Lie bracket,
the nonzero commutation relations between vector fields spanning~$\mathfrak g_{\rm L}$
are exhausted by the counterparts of the commutation relations~\eqref{eq:dNCommRelations}
and one more commutation relation involving the vector field~$\bar P^v$,
$[\bar P^v,\bar D^{\rm s}]=\tfrac32\bar P^v$.

In~\cite{boyk2024a}, we recomputed the algebra~$\mathfrak g$
as well as first obtained the algebra~$\mathfrak g_{\rm L}$ using the packages
{\sf DESOLV} \cite{carm2000a} and {\sf Jets} \cite{BaranMarvan,marv2009a} for {\sf Maple};
the latter package was also applied for computing the algebra~$\mathfrak g_{\rm c}$
and generalized symmetries of~\eqref{eq:dN} up to order five.

\section{Point-symmetry pseudogroups}\label{sec:PointSymGroup}

The point-symmetry pseudogroups~$G$ and~$G_{\rm L}$ of the equation~\eqref{eq:dN}
and its nonlinear Lax representation~\eqref{eq:dNLaxPair}
were computed in~\cite{boyk2024a} using
the original megaideal-based version of the algebraic method that was suggested in~\cite{malt2024a}.

\begin{theorem}\label{thm:dNCompletePointSymGroup}
The point-symmetry pseudogroup~$G$ of the dispersionless Nizhnik equation~\eqref{eq:dN}
is generated by the transformations of the form
\begin{gather}\label{eq:dNPointSymForm}
\begin{split}
&\tilde t=T(t),\quad
\tilde x=CT_t^{1/3}x+X^0(t),\quad
\tilde y=CT_t^{1/3}y+Y^0(t),\\
&\tilde u=C^3u-\frac{C^3T_{tt}}{18T_t}(x^3+y^3)
-\frac{C^2}{2T_t^{1/3}}(X^0_tx^2+Y^0_ty^2)+W^1(t)x+W^2(t)y+W^0(t)
\end{split}
\end{gather}
and the transformation~$\mathscr J$: $\tilde t=t$, $\tilde x=y$, $\tilde y=x$, $\tilde u=u$.
Here $T$, $X^0$, $Y^0$, $W^0$, $W^1$ and $W^2$ are arbitrary smooth functions of~$t$
with $T_t\neq0$, and $C$ is an arbitrary nonzero constant.
\end{theorem}

Moreover, it was also proved in~\cite{boyk2024a}
that the contact-symmetry pseudogroup~$G_{\rm c}$ of the equation~\eqref{eq:dN}
coincides with the first prolongation~$G_{(1)}$ of the pseudogroup~$G$.

Each transformation $\Phi$ from the point-symmetry pseudogroup~$G$ of the dispersionless Nizhnik equation~\eqref{eq:dN}
can be represented as a composition
of transformations from subgroups each of which is parameterized by a single functional or discrete parameter,
\begin{equation}\label{eq:dNElementaryTransformations}\arraycolsep=0ex
\begin{array}{llllll}
\mathscr D^t(T)      \colon\quad & \tilde t=T,\quad & \tilde x=T_t^{1/3}x,\quad & \tilde y=T_t^{1/3}y,\quad & \tilde u=u-\frac{1}{18}T_{tt}{T_t}^{-1}(x^3+y^3),   \\[1ex]
\mathscr D^{\rm s}(C)\colon\quad & \tilde t=t,\quad & \tilde x=Cx,        \quad & \tilde y=Cy,        \quad & \tilde u=C^3u,                                      \\[1ex]
\mathscr P^x(X^0)    \colon\quad & \tilde t=t,\quad & \tilde x=x+X^0,     \quad & \tilde y=y,         \quad & \tilde u=u-\frac16X^0_t\big(3x^2+3X^0x+(X^0)^2\big),\\[1ex]
\mathscr P^y(Y^0)    \colon\quad & \tilde t=t,\quad & \tilde x=x,         \quad & \tilde y=y+Y^0,     \quad & \tilde u=u-\frac16Y^0_t\big(3y^2+3Y^0y+(Y^0)^2\big),\\[1ex]
\mathscr R^x(W^1)    \colon\quad & \tilde t=t,\quad & \tilde x=x,         \quad & \tilde y=y,         \quad & \tilde u=u+W^1x,                                    \\[1ex]
\mathscr R^y(W^2)    \colon\quad & \tilde t=t,\quad & \tilde x=x,         \quad & \tilde y=y,         \quad & \tilde u=u+W^2y,                                    \\[1ex]
\mathscr Z(W^0)      \colon\quad & \tilde t=t,\quad & \tilde x=x,         \quad & \tilde y=y,         \quad & \tilde u=u+W^0,                                     \\[1ex]
\mathscr J           \colon\quad & \tilde t=t,\quad & \tilde x=y,         \quad & \tilde y=x,         \quad & \tilde u=u,                                         \\[1ex]
\end{array}
\end{equation}
where $T$, $X^0$, $Y^0$, $W^0$, $W^1$ and $W^2$ are arbitrary smooth functions of~$t$
with $T_t\neq0$, and $C$ is an arbitrary nonzero constant.
We will call transformations from the families~\eqref{eq:dNElementaryTransformations}
elementary point symmetry transformations of the equation~\eqref{eq:dN}.
Note that the subgroups~$\{\mathscr D^t(T)\}$, $\{\mathscr D^{\rm s}(C)\}$,
$\{\mathscr P^x(X^0)\}$, $\{\mathscr P^y(Y^0)\}$,
$\{\mathscr R^x(W^1)\}$, $\{\mathscr R^y(W^2)\}$ and $\{\mathscr Z(W^0)\}$ of~$G$
are associated with the subalgebras~$\{D^t(\tau)\}$, $\langle D^{\rm s}\rangle$,
$\{P^x(\chi)\}$, $\{P^y(\rho)\}$, $\{R^x(\alpha)\}$,
$\{R^y(\beta)\}$ and $\{Z(\sigma)\}$ of~$\mathfrak g$, respectively.
Here all the parameter functions run through the specified sets of their values.
A representation of a~transformation~$\Phi$ of the form~\eqref{eq:dNPointSymForm}
as a composition of elementary point symmetry transformations of the equation~\eqref{eq:dN} is
\[
\Phi=\mathscr D^t(T)
\circ\mathscr D^{\rm s}(C)
\circ\mathscr P^x(\tilde X^0)
\circ\mathscr P^y(\tilde Y^0)
\circ\mathscr R^x(\tilde W^1)
\circ\mathscr R^y(\tilde W^2)
\circ\mathscr Z(\tilde W^0)
\]
with
\begin{gather*}
\tilde X^0=\frac{X^0}{CT_t^{1/3}},\quad
\tilde Y^0=\frac{Y^0}{CT_t^{1/3}},\\
\tilde W^0=\frac{W^0}{C^3}+\frac{X^0_t(X^0)^2+Y^0_t(Y^0)^2}{6C^3T_t},\quad
\tilde W^1=\frac{W^1}{C^3}+\frac{X^0_tX^0}{2C^2T_t^{2/3}},\quad
\tilde W^2=\frac{W^2}{C^3}+\frac{Y^0_tY^0}{2C^2T_t^{2/3}}.
\end{gather*}

\begin{corollary}\label{cor:dNDiscrSyms}
The identity component~$G_{\rm id}$ of the point-symmetry pseudogroup~$G$ of the dispersionless Nizhnik equation~\eqref{eq:dN}
consists of the transformations of the form~\eqref{eq:dNPointSymForm} with $T_t>0$ and $C>0$.
A complete list of discrete point symmetry transformations of the equation~\eqref{eq:dN}
that are independent up to composing with each other and with transformations from~$G_{\rm id}$
is exhausted by three commuting involutions, which can be chosen to be
the permutation~$\mathscr J$ of the variables~$x$ and~$y$
and two transformations $\mathscr I^{\rm i}:=\mathscr D^t(-t)$ and $\mathscr I^{\rm s}:=\mathscr D^{\rm s}(-1)$
alternating the signs of $(t,x,y)$ and of $(x,y,u)$, respectively.
\end{corollary}

Therefore, the quotient group~$G/G_{\rm id}$ of the pseudogroup~$G$
with respect to its identity component~$G_{\rm id}$ is isomorphic to the group $\mathbb Z_2\times\mathbb Z_2\times\mathbb Z_2$.

\begin{theorem}\label{thm:dNCompletePointSymGroupOfLaxRepresentation}
The point-symmetry pseudogroup~$G_{\rm L}$ of the nonlinear Lax representation~\eqref{eq:dNLaxPair}
is generated by the transformations of the form
\[
\begin{split}
&\tilde t=T(t),\quad
 \tilde x=A^{2/3}T_t^{1/3}x+X^0(t),\quad
 \tilde y=A^{2/3}T_t^{1/3}y+Y^0(t),\\
&\tilde u=A^2u-\frac{A^2T_{tt}}{18T_t}(x^3+y^3)
 -\frac{A^{4/3}}{2T_t^{1/3}}(X^0_tx^2+Y^0_ty^2)+W^1(t)x+W^2(t)y+W^0(t),\\
&\tilde v=Av+B,
\end{split}
\]
and the transformation~$\mathscr J$: $\tilde t=t$, $\tilde x=y$, $\tilde y=x$, $\tilde u=u$, $\tilde v=v$.
Here $T$, $X^0$, $Y^0$, $W^0$, $W^1$ and $W^2$ are arbitrary smooth functions of~$t$
with $T_t\neq0$, and $A$ and~$B$ are arbitrary constants with $A\ne0$.
\looseness=-1
\end{theorem}

It follows from Theorem~\ref{thm:dNCompletePointSymGroupOfLaxRepresentation}
that each elementary point symmetry transformation of the equation~\eqref{eq:dN}, except $\mathscr D^{\rm s}(C)$,
is trivially extended to an elementary point symmetry transformation of the system~\eqref{eq:dNLaxPair}
with the identity $v$-component, $\tilde v=v$.
The extension of~$\mathscr D^{\rm s}(C)$ with $C>0$ is given by $\tilde v=C^{3/2}v$,
whereas the transformations~$\mathscr D^{\rm s}(C)$ with $C<0$ have no counterparts in~$G_{\rm L}$.
We denote the extension of an elementary point symmetry transformation of the equation~\eqref{eq:dN}
by the same symbol supplemented with bar.
The pseudogroup~$G_{\rm L}$ contains two kinds of elementary transformations that have no counterparts in~$G$,
shifts with respect to~$v$, $\bar{\mathscr P}^v(B)$: $(\tilde t,\tilde x,\tilde y,\tilde u,\tilde v)=(t,x,y,u,v+B)$,
and the transformation~$\bar{\mathscr I}^v$ alternating the signs of~$v$, $(\tilde t,\tilde x,\tilde y,\tilde u,\tilde v)=(t,x,y,u,-v)$.
The point transformation~$\bar{\mathscr I}^{\rm s}\colon(\tilde t,\tilde x,\tilde y,\tilde u,\tilde v)=(t,-x,-y,-u,v)$,
which is the trivial extension of the discrete point symmetry transformation~$\mathscr I^{\rm s}$
of the equation~\eqref{eq:dN} to~$v$,
maps the nonlinear Lax representation~\eqref{eq:dNLaxPair} of the equation~\eqref{eq:dN}
to an equivalent nonlinear Lax representation of the same equation,
\[
v_t=-\frac13\left(v_x^3+\frac{u_{xy}^3}{v_x^3}\right)+u_{xx}v_x+\frac{u_{xy}u_{yy}}{v_x},\quad
v_y=\frac{u_{xy}}{v_x}.
\]

\begin{corollary}\label{cor:dNLaxPairDiscrSyms}
A complete list of discrete point symmetry transformations of the system~\eqref{eq:dNLaxPair}
that are independent up to composing with each other and with continuous point symmetry transformations of this system
is exhausted by three commuting involutions, which can be chosen to be
the permutation~$\bar{\mathscr J}$ of the variables~$x$ and~$y$, $(\tilde t,\tilde x,\tilde y,\tilde u,\tilde v)=(t,y,x,u,v)$,
and two transformations~$\bar{\mathscr I}^{\rm i}$ and~$\bar{\mathscr I}^v$
alternating the signs of $(t,x,y)$ and of $v$, respectively,
$\bar{\mathscr I}^{\rm i}\colon(\tilde t,\tilde x,\tilde y,\tilde u,\tilde v)=(-t,-x,-y,u,v)$ and
$\bar{\mathscr I}^v\colon(\tilde t,\tilde x,\tilde y,\tilde u,\tilde v)=(t,x,y,u,-v)$.
\end{corollary}

Therefore, analogously to the pseudogroup~$G$,
the quotient group of the point-symmetry pseudogroup~$G_{\rm L}$
of the nonlinear Lax representation~\eqref{eq:dNLaxPair} of the dispersionless Nizhnik equation~\eqref{eq:dN}
with respect to its identity component is isomorphic to the group $\mathbb Z_2\times\mathbb Z_2\times\mathbb Z_2$.

\section{Classification of one- and two-dimensional subalgebras}\label{sec:dNClassificationOfSubalgs}

To carry out Lie reductions of codimension one and two for the equation~\eqref{eq:dN} in the optimal way,
we should classify one- and two-dimensional subalgebras of the algebra~$\mathfrak g$ up to $G_*$-equivalence.
Instead of the classical approach for finding inner automorphisms~\cite[Section~3.3]{olve1993A},
we act on~$\mathfrak g$ by~$G$ via pushing forward of vector fields by elements of~$G$.
Recall that this way is more convenient for computing in the infinite-dimensional case~\cite{bihl2012b,card2011a}.
Moreover, it also allows us to properly use the entire complete point-symmetry group~$G$
and not be limited to its connected component of the identity transformation.
Thus the non-identity adjoint actions of elementary transformations from~$G$
on vector fields spanning~$\mathfrak g$ are merely
\begin{gather*}
\mathscr D^t_*(T)D^t(\tau)=D^t\big(\hat T^{-1}_t\tau(\hat T)\big),\\
\mathscr D^t_*(T)P^x(\chi)=P^x\big(\hat T^{-1/3}_t\chi(\hat T)\big),\quad
\mathscr D^t_*(T)P^y(\rho)=P^y\big(\hat T^{-1/3}_t\rho(\hat T)\big),\\
\mathscr D^t_*(T)R^x(\alpha)=R^x\big(\hat T^{1/3}_t\alpha(\hat T)\big),\quad
\mathscr D^t_*(T)R^y(\beta)=R^y\big(\hat T^{1/3}_t\beta(\hat T)\big),\quad
\mathscr D^t_*(T)Z(\sigma)=Z\big(\sigma(\hat T)\big),\!
\\[1.5ex]
\mathscr D^{\rm s}_*(C)P^x(\chi)=P^x(C\chi),\quad
\mathscr D^{\rm s}_*(C)P^y(\rho)=P^y(C\rho),\\[.5ex]
\mathscr D^{\rm s}_*(C)R^x(\alpha)=R^x(C^2\alpha),\quad
\mathscr D^{\rm s}_*(C)R^y(\beta)=R^y(C^2\beta),\quad
\mathscr D^{\rm s}_*(C)Z(\sigma)=Z(C^3\sigma),
\\[1.5ex]
\mathscr P^x_*(X^0)D^t(\tau)=D^t(\tau)+P^x\big(\tau X^0_t-\tfrac13\tau_tX^0\big)+\tfrac12R^x\big(X^0(\tau X^0_t)_t-\tfrac13\tau_{tt}(X^0)^2-\tau(X^0_t)^2\big)\\
\hphantom{\mathscr P^x_*(X^0)D^t(\tau)=}-\tfrac16Z\big((X^0)^2(\tau X^0_t)_t-\tfrac13\tau_{tt}(X^0)^3-\tau X^0(X^0_t)^2\big),
\\[1.5ex]
\mathscr P^y_*(Y^0)D^t(\tau)=D^t(\tau)+P^y(\tau Y^0_t-\tfrac13\tau_tY^0)+\tfrac12R^y\big(Y^0(\tau Y^0_t)_t-\tfrac13\tau_{tt}(Y^0)^2-\tau(Y^0_t)^2\big)\\
\hphantom{\mathscr P^y_*(Y^0)D^t(\tau)=}-\tfrac16Z\big((Y^0)^2(\tau Y^0_t)_t-\tfrac13\tau_{tt}(Y^0)^3-\tau Y^0(Y^0_t)^2\big),
\\[1ex]
\mathscr R^x_*(W^1)D^t(\tau)=D^t(\tau)+R^x(\tau W^1_t+\tfrac13\tau_tW^1),\\[.5ex]
\mathscr R^y_*(W^2)D^t(\tau)=D^t(\tau)+R^y(\tau W^2_t+\tfrac13\tau_tW^2),\quad
\mathscr Z_*(W^0)D^t(\tau)=D^t(\tau)+Z(\tau W^0_t),
\\[1.5ex]
\mathscr P^x_*(X^0)D^{\rm s}=D^{\rm s}-P^x(X^0),\quad
\mathscr R^x_*(W^1)D^{\rm s}=D^{\rm s}-2R^x(W^1),\\[.5ex]
\mathscr P^y_*(Y^0)D^{\rm s}=D^{\rm s}-P^y(Y^0),\quad
\mathscr R^y_*(W^2)D^{\rm s}=D^{\rm s}-2R^y(W^2),\quad
\mathscr Z_*(W^0)D^{\rm s}=D^{\rm s}-3Z(W^0),\!\!
\\[1.5ex]
\mathscr P^x_*(X^0)P^x(\chi)=P^x(\chi)+R^x\big(\chi_tX^0-\chi X^0_t\big)-\tfrac12Z\big(\chi_t(X^0)^2-\chi X^0X^0_t\big),\\[.5ex]
\mathscr P^x_*(X^0)R^x(\alpha)=R^x(\alpha)-Z(\alpha X^0),\quad
\mathscr R^x_*(W^1)P^x(\chi)=P^x(\chi)+Z(\chi W^1),
\\[1.5ex]
\mathscr P^y_*(Y^0)P^y(\rho)=P^y(\rho)+R^y(\rho_tY^0-\rho Y^0_t)-\tfrac12Z(\rho_t(Y^0)^2-\rho Y^0_tY^0),\\[.5ex]
\mathscr P^y_*(Y^0)R^y(\beta)=R^y(\beta)-Z(\beta Y^0),\quad
\mathscr R^y_*(W^2)P^y(\rho)=P^y(\rho)+Z(\rho W^2),
\\[1.5ex]
\mathscr J_*P^x(\chi)=P^y(\chi),\quad
\mathscr J_*P^y(\rho)=P^x(\rho),\quad
\mathscr J_*R^x(\alpha)=R^y(\alpha),\quad
\mathscr J_*R^y(\beta)=R^x(\beta),
\end{gather*}
where~$\hat T$ is the inverse of the function~$T$.
At the same time, a part of adjoint actions can be computed via mimicking the classical approach
if the corresponding Lie series has a finite number of nonzero terms.

\begin{lemma}\label{lem:dN1DInequivSubalgs}
A complete list of $G$-inequivalent one-dimensional subalgebras of the algebra~$\mathfrak g$
is exhausted by the following subalgebra families:
\begin{gather*}
\mathfrak s_{1.1}^\delta=\big\langle D^t(1)+\delta D^{\rm s}\big\rangle,\quad
\mathfrak s_{1.2}       =\big\langle D^{\rm s}              \big\rangle,\quad
\mathfrak s_{1.3}^\rho  =\big\langle P^x(1)+P^y(\rho)       \big\rangle,\quad
\mathfrak s_{1.4}^\beta =\big\langle P^x(1)+R^y(\beta)      \big\rangle,\\
\mathfrak s_{1.5}^\beta =\big\langle R^x(1)+R^y(\beta)      \big\rangle,\quad
\mathfrak s_{1.6}       =\big\langle Z(t)                   \big\rangle,\quad
\mathfrak s_{1.7}       =\big\langle Z(1)                   \big\rangle,
\end{gather*}
where $\delta\in\{0,1\}$ $(\!{}\bmod G)$, and $\rho$ and $\beta$ run through the set of smooth functions of~$t$ with $\rho\ne0$.
\end{lemma}

\begin{proof}
Let $\mathfrak s_1=\langle Q\rangle$ be a one-dimensional subalgebra of~$\mathfrak g$
spanned by a nonvanishing vector field
$Q=D^t(\tau)+\lambda D^{\rm s}+P^x(\chi)+P^y(\rho)+R^x(\alpha)+R^y(\beta)+Z(\sigma)$ from~$\mathfrak g$.
Here $\tau$, $\chi$, $\rho$, $\alpha$, $\beta$ and $\sigma$ are arbitrary smooth functions of $t$
and $\lambda$ is an arbitrary constant that do not simultaneously vanish.\looseness=-1

If the function~$\tau$ is nonzero, then we use~$\mathscr D^t_*(T)$ with $T_t=1/\tau$ to set $\tau=1$
and, preserving the notation of the parameter functions, successively act on the (currently modified) vector field~$Q$
by $\mathscr P^x_*(X^0)\circ\mathscr P^y_*(Y^0)$ with $X^0_t-\lambda X^0=-\chi$ and $Y^0_t-\lambda Y^0=-\rho$ to set $\chi=\rho=0$,
by $\mathscr R^x_*(W^1)\circ\mathscr R^y_*(W^2)$ with $W^1_t-2\lambda W^1=-\alpha$ and $W^2_t-2\lambda W^2=-\beta$ to set $\alpha=\beta=0$ and
by $\mathscr Z_*(W^0)$ with $W^0_t-3\lambda W^0=-\sigma$ to set $\sigma=0$.
Thus, we obtain $Q=D^t(1)+\lambda D^{\rm s}$.
If $\lambda\ne0$, we can set $\lambda=1$ by simultaneously scaling~$t$ and the entire~$Q$
and, if necessary, alternating their signs.
In other words, the subalgebra~$\mathfrak s_1$ with~$\tau\ne0$
is $G$-equivalent to a one in the family $\{\mathfrak s_{1.1}^0,\,\mathfrak s_{1.1}^1\}$.

Suppose that $\tau=0$ and $\lambda\ne0$.
Changing the basis element~$Q$, we first set $\lambda=1$.
Then, preserving the notation of the parameter functions and successively act on the (currently modified) vector field~$Q$
by $\mathscr P^x_*(\chi)\circ\mathscr P^y_*(\rho)$ to set $\chi=\rho=0$,
by $\mathscr R^x_*(\frac12\alpha)\circ\mathscr R^y_*(\frac12\beta)$ to set $\alpha=\beta=0$ and
by $\mathscr Z_*(\frac13\sigma)$ to set $\sigma=0$,
which leads to the subalgebra~$\mathfrak s_{1.2}$.

Let $\tau=0$, $\lambda=0$ and $\chi\rho\ne0$.
Analogously to the above cases, a chain of simplifying successive actions is
$\mathscr D^t_*(T)$ with $T_t=\chi^{-3}$,
$\mathscr P^x_*(X^0)\circ\mathscr P^y_*(Y^0)$ with $X^0_t=\alpha$ and $\rho Y^0_t-\rho_tY^0=\beta$ and
by $\mathscr R^x_*(W^1)$ with $W^1=-\sigma$,
which gives $\chi=1$, $\alpha=\beta=0$ and $\sigma=0$.
Thus, we have the subalgebra~\smash{$\mathfrak s^\rho_{1.3}$}.

Let $\tau=0$ and $\lambda=0$ and exactly one of the parameter functions $\chi$ and $\rho$ is nonzero.
Up to the permutation of~$x$ and~$y$, we can assume without loss of generality that $\chi\ne0$ and $\rho=0$.
Similarly to the previous case, we set $\chi=1$, $\alpha=0$ and $\sigma=0$,
and obtain the subalgebra~\smash{$\mathfrak s^\beta_{1.4}$}.

Further we assume $\tau=0$, $\lambda=0$ and $\chi=\rho=0$.

If $(\alpha,\beta)\ne(0,0)$, then due to the possibility of permuting~$x$ and~$y$,
we can assume, without loss of generality, $\alpha\ne0$
and set $\alpha=1$ and $\sigma=0$ modulo the $G$-equivalence,
which gives the subalgebra~\smash{$\mathfrak s^\beta_{1.5}$}.

Otherwise, $\alpha=\beta=0$ and $\sigma\ne0$.
The consideration splits into two cases $\sigma_t\ne0$ and $\sigma_t=0$,
where the subalgebra~$\mathfrak s_1$ is $G$-equivalent to $\mathfrak s_{1.6}$ and~$\mathfrak s_{1.7}$, respectively.
\end{proof}

\begin{lemma}\label{lem:dN2DInequivSubalgs}
A complete list of $G$-inequivalent two-dimensional subalgebras of the algebra~$\mathfrak g$
is exhausted by the non-abelian algebras
\begin{gather*}
\mathfrak s_{2.1}^\lambda=\big\langle D^t(1),\,D^t(t)+\lambda D^{\rm s}\big\rangle,\quad
\mathfrak s_{2.2}^\nu    =\big\langle D^t(1),\,D^t(t)-\tfrac13 D^{\rm s}+P^x(1)+P^y(\nu)\big\rangle,\\
\mathfrak s_{2.3}^\nu    =\big\langle D^t(1),\,D^t(t)+\tfrac16 D^{\rm s}+R^x(1)+R^y(\nu)\big\rangle,\quad
\mathfrak s_{2.4}        =\big\langle D^t(1),\,D^t(t)+Z(1)\big\rangle,\\
\mathfrak s_{2.5}^{\lambda\mu}=\big\langle D^t(1)+\lambda D^{\rm s},\,P^x({\rm e}^{(\lambda-1)t})+\mu P^y({\rm e}^{(\lambda-1)t})\big\rangle,\\
\mathfrak s_{2.6}^{\lambda\delta}=\big\langle D^t(1)+\lambda D^{\rm s},\,P^x({\rm e}^{(\lambda-1)t})+\delta R^y({\rm e}^{(2\lambda-1)t})\big\rangle,\\
\mathfrak s_{2.7}^{\lambda\nu}=\big\langle D^t(1)+\lambda D^{\rm s},\,R^x({\rm e}^{(2\lambda-1)t})+\nu R^y({\rm e}^{(2\lambda-1)t})\big\rangle,\quad
\mathfrak s_{2.8}^\lambda     =\big\langle D^t(1)+\lambda D^{\rm s},\,Z({\rm e}^{(3\lambda-1)t})\big\rangle,\\
\mathfrak s_{2.9}^{\tilde\rho}=\big\langle D^{\rm s},\,P^x(1)+P^y(\tilde\rho)\big\rangle,\quad
\mathfrak s_{2.10}^\beta      =\big\langle D^{\rm s},\,R^x(1)+R^y(\beta)\big\rangle,\\
\mathfrak s_{2.11}            =\big\langle D^{\rm s},\,Z(t)\big\rangle,\quad
\mathfrak s_{2.12}            =\big\langle D^{\rm s},\,Z(1)\big\rangle,\quad
\end{gather*}
and the abelian algebras
\begin{gather*}
\mathfrak s_{2.13}       =\big\langle D^t(1),\,D^{\rm s}\big\rangle,\quad
\mathfrak s_{2.14}^{\delta\nu\delta'}=\big\langle D^t(1)+\delta D^{\rm s},\,
P^x({\rm e}^{\delta t})+\nu P^y({\rm e}^{\delta t})+\delta'R^y({\rm e}^{2\delta t})\big\rangle,
\\
\mathfrak s_{2.15}^{\delta\nu}=\big\langle D^t(1)+\delta D^{\rm s},\,R^x({\rm e}^{2\delta t})+\nu R^y({\rm e}^{2\delta t})\big\rangle,\quad
\mathfrak s_{2.16}^\delta     =\big\langle D^t(1)+\delta D^{\rm s},\,Z({\rm e}^{3\delta t}) \big\rangle,\\
\mathfrak s_{2.17}^{\rho\alpha\beta}=\big\langle P^x(1)+R^y(\beta),\,P^y(\rho)+R^x(\rho\beta)\big\rangle,\\
\mathfrak s_{2.18}^{\rho\beta\sigma}=\big\langle P^x(1)+P^y(\rho),\,-R^x(\rho\beta)+R^y(\beta)+Z(\sigma)\big\rangle_{(\beta,\sigma)\ne(0,0)},\\
\mathfrak s_{2.19}^{\beta^1\beta^2} =\big\langle P^x(1)+R^y(\beta^1),\,R^y(\beta^2)\big\rangle_{\beta^2\ne0},\quad
\mathfrak s_{2.20}^{\beta\sigma}    =\big\langle P^x(1)+R^y(\beta),\,Z(\sigma)\big\rangle_{\sigma\ne0},\\
\mathfrak s_{2.21}^{\alpha\beta^1\beta^2} =\big\langle R^x(1)+R^y(\beta^1),\,R^x(\alpha)+R^y(\beta^2)\big\rangle_{\beta^2\ne\alpha\beta^1},\\
\mathfrak s_{2.22}^{\alpha\beta\sigma}    =\big\langle R^x(1)+R^y(\beta),\,R^x(\alpha)+R^y(\alpha\beta)+Z(\sigma)\big\rangle_{\alpha_t\ne0},\\
\mathfrak s_{2.23}^{\beta\sigma}          =\big\langle R^x(1)+R^y(\beta),\,Z(\sigma)\big\rangle_{\sigma\ne0},\ \
\mathfrak s_{2.24}^\sigma =\big\langle Z(t),\,Z(\sigma)\big\rangle_{\sigma_{tt}\ne0},\ \
\mathfrak s_{2.25}^\sigma =\big\langle Z(1),\,Z(\sigma)\big\rangle_{\sigma_t\ne0},
\end{gather*}
where $\rho$, $\tilde\rho$, $\alpha$, $\beta$, $\beta^1$, $\beta^2$ and $\sigma$
run through the set of smooth functions of~$t$ with $\rho\ne0$,
$\lambda\in\mathbb R$, $\mu\in[-1,1]\setminus\{0\}$ and $\nu\in[-1,1]$ $(\!{}\bmod G)$,
$\delta,\delta'\in\{0,1\}$ $(\!{}\bmod G)$,
and the conditions indicated after the corresponding subalgebras should be satisfied as well.
\end{lemma}

\begin{proof}
Consider a two-dimensional subalgebra $\mathfrak s_2=\langle Q^1,Q^2\rangle$ of~$\mathfrak g$
spanned by two (linearly independent) vector fields
\begin{gather*}
Q^i=D^t(\tau^i)+\lambda^iD^{\rm s}+P^x(\chi^i)+P^y(\rho^i)+R^x(\alpha^i)+R^y(\beta^i)+Z(\sigma^i),\quad i=1,2,
\end{gather*}
from~$\mathfrak g$
with arbitrary smooth functions $\tau^i$, $\chi^i$, $\rho^i$, $\alpha^i$, $\beta^i$ and $\sigma^i$ of $t$
and arbitrary constants $\lambda^i$ such that
the tuples $(\tau^i,\lambda^i,\chi^i,\rho^i,\alpha^i,\beta^i,\sigma^i)$, $i=1,2$, are linearly independent.
Moreover, since $[Q^1,Q^2]\in\langle Q^1,Q^2\rangle$, up to changing the basis $(Q^1,Q^2)$,
we can assume that either $[Q^1,Q^2]=Q^1$ or $[Q^1,Q^2]=0$
if the subalgebra~$\mathfrak s_2$ is non-abelian or abelian, respectively.
Consider these cases separately.
For each of the obtained families of subalgebras,
we do not indicate a~final tuning of its basis elements, which involves permuting or scaling these elements
or omitting superfluous indices.

\medskip

\noindent {\bf I.}
The commutation relation $[Q^1,Q^2]=Q^1$ implies $\lambda^1=0$.

First suppose that the functions~$\tau^1$ and $\tau^2$ are linearly independent.
Then the projections $\tau^i\p_t$ of~$Q^i$, $i=1,2$ on the $t$-line span a two-dimensional Lie algebra
of vector fields on the real line with $[\tau^1\p_t,\tau^2\p_t]=\tau^1\p_t$.
In view of the classical Lie theorem, there exists a point transformation $\tilde t=T(t)$ of~$t$
that pushes forward the vector fields $\tau^1\p_t$ and $\tau^2\p_t$ to $\p_{\tilde t}$ and $\tilde t\p_{\tilde t}$.
This means that the action by~$\mathscr D^t_*(T)$ allows us to set $\tau^1=1$ and $\tau^2=t$.
Following the first case of the proof of Lemma~\ref{lem:dN1DInequivSubalgs},
we can further set $\chi^1$, $\rho^1$, $\alpha^1$, $\beta^1$ and $\sigma^1$ to 0.
Re-denote $(\lambda^2,\chi^2,\rho^2,\alpha^2,\beta^2,\sigma^2)$ as $(\lambda,\chi,\rho,\alpha,\beta,\sigma)$.
Under the derived constraints, the commutation relation $[Q^1,Q^2]=Q^1$
is equivalent to the equations $\chi_t=\rho_t=\alpha_t=\beta_t=\sigma_t=0$,
i.e., all these subalgebra parameters are constants.
The pushforward by a transformation~$\Phi$ from~$G$ does not change the vector field $Q^1=D^t(1)$ up to its scaling
if and only if $T=at+b$ for some constants~$a$ and~$b$
and the parameter functions~$X^0$, $Y^0$, $W^0$, $W^1$ and $W^2$ are constants,
whereas the constant $C$ is not additionally constrained,\looseness=-1
\begin{gather*}
\Phi\colon\quad
\tilde t=at+b,\quad
\tilde x=Cx+X^0,\quad
\tilde y=Cy+Y^0,\quad
\tilde u=C^3u+W^1x+W^2y+W^0,
\end{gather*}
i.e., $\Phi=\mathscr D^t_*(at+b)\circ\mathscr P^x_*(X^0)\circ\mathscr P^y_*(Y^0)
\circ\mathscr R^x_*(W^1/C)\circ\mathscr R^y_*(W^2/C)\circ\mathscr Z_*(W^0)\circ\mathscr D^{\rm s}_*(C)$.

Pushing forward $\mathfrak s_2$ by such~$\Phi$ with $a=C=1$ in addition, we have $\Phi_*Q^1=Q^1$ and
\begin{gather*}
\begin{split}
\Phi_*Q^2={}&Q^2-bQ^1
-(\lambda+\tfrac13)\big(P^x(X^0)+P^y(Y^0)\big)
-(2\lambda-\tfrac13)\big(R^x(W^1)+R^y(W^2)\big)\\&
-Z\big(\alpha X^0+\beta Y^0-\chi W^1-\rho W^2-(2\lambda-\tfrac13)(W^1X^0+W^2Y^0)+3\lambda W^0\big).
\end{split}
\end{gather*}
The last equation implies that for general values of~$\lambda$,
we can set $\chi=\rho=\alpha=\beta=\sigma=0$, obtaining the subalgebra family~$\{\mathfrak s^{\lambda}_{2.1}\}$.
Subalgebras of the considered kind that are not $G$-equivalent to elements of the family~$\{\mathfrak s^{\lambda}_{2.1}\}$
correspond to the special values $-\frac13$, $\frac16$ and $0$ of $\lambda$,
where in addition $(\chi,\rho)\ne(0,0)$, $(\alpha,\beta)\ne(0,0)$ and $\sigma\ne0$, respectively.
In each of these cases, the other parameters in the tuple $(\chi,\rho,\alpha,\beta,\sigma)$
can be set to zero by~$\Phi_*$ as above.
Using the permutation of $x$ and~$y$, we replace
the inequality $(\chi,\rho)\ne(0,0)$ by $\chi\ne0$ and $|\rho|\leqslant|\chi|$ and
the inequality $(\alpha,\beta)\ne(0,0)$ by $\alpha\ne0$ and $|\beta|\leqslant|\alpha|$.
Acting by~$\mathscr D^t_*(at)$ or $\mathscr D^{\rm s}_*(C)$,
we scale the nonzero parameter among $\chi$, $\alpha$ or $\sigma$ to~1,
which leads to the subalgebras~$\mathfrak s_{2.2}^\rho$ with $|\rho|\leqslant1$,
\smash{$\mathfrak s_{2.3}^\beta$} with $|\beta|\leqslant1$ or $\mathfrak s_{2.4}$,
respectively.

Now, let the functions~$\tau^1$ and $\tau^2$ are linearly dependent but not simultaneously zero.
The commutation relation $[Q^1,Q^2]=Q^1$ implies that $\tau^1=0$ and $\tau^2\ne0$ in this case.
Following the first case of the proof of Lemma~\ref{lem:dN1DInequivSubalgs},
we can set $Q^2=D^t(1)+\lambda D^{\rm s}$, where we re-denote $\lambda^2$ by $\lambda$.
Recalling again the commutation relation $[Q^1,Q^2]=Q^1$,
we obtain
$\chi^1=\nu_1{\rm e}^{(\lambda-1)t}$,
$\rho^1=\nu_2{\rm e}^{(\lambda-1)t}$,
$\alpha^1=\nu_3{\rm e}^{(2\lambda-1)t}$,
$\beta^1=\nu_4{\rm e}^{(2\lambda-1)t}$ and
$\sigma^1=\nu_5{\rm e}^{(3\lambda-1)t}$
with constants $\nu_1$, \dots, $\nu_5$.
For further simplification, we can apply only the pushforwards that do not change
the form of the basis element~$Q^2$ up to its linearly combining with~$Q^1$.
The pushforward by an element~$\Phi$ of~$G$ has this property if and only if
the transformation~$\Phi$ or $\mathscr J\circ\Phi$ is of the form~\eqref{eq:dNPointSymForm} with
\begin{gather*}
T=t+b,\quad
X^0=\kappa_1{\rm e}^{\lambda t}-\kappa_0C\nu_1{\rm e}^{(\lambda-1)t},\quad
Y^0=\kappa_2{\rm e}^{\lambda t}-\kappa_0C\nu_2{\rm e}^{(\lambda-1)t},\\
W^1=\kappa_3{\rm e}^{2\lambda t}-(\kappa_0C^3\nu_3-\lambda\kappa_0C^2\nu_1\kappa_1){\rm e}^{(2\lambda-1)t}
-\frac{\lambda-1}2\kappa_0^{\,2}C^3\nu_1^{\,2}{\rm e}^{2(\lambda-1)t},\\
W^2=\kappa_4{\rm e}^{2\lambda t}-(\kappa_0C^3\nu_4-\lambda\kappa_0C^2\nu_2\kappa_2){\rm e}^{(2\lambda-1)t}
-\frac{\lambda-1}2\kappa_0^{\,2}C^3\nu_2^{\,2}{\rm e}^{2(\lambda-1)t},\\
\begin{split}
W^0={}&\kappa_5{\rm e}^{3\lambda t}-\kappa_0(C^3\nu_5+\nu_1\kappa_3+\nu_2\kappa_4){\rm e}^{(3\lambda-1)t}
+\frac12\kappa_0^{\,2}C^3(\nu_1\nu_3+\nu_2\nu_4){\rm e}^{(3\lambda-2)t}\\&
-\frac\lambda2\kappa_0^{\,2}C^2(\nu_1^{\,2}\kappa_1+\nu_2^{\,2}\kappa_2){\rm e}^{(3\lambda-2)t}
+\frac{\lambda-1}6\kappa_0^{\,3}C^3(\nu_1^{\,3}+\nu_2^{\,3}){\rm e}^{3(\lambda-1)t},
\end{split}
\end{gather*}
where $C$, $b$ and $\kappa_0$, \dots, $\kappa_5$ are arbitrary constants with~$C\ne0$.
In addition, we can multiply~$Q^1$ by an arbitrary nonzero constant.
As a result, we set
\begin{itemize}\itemsep=-0.3ex
\item[$\circ$] $\nu_1=1$, \ $|\nu_2|\leqslant1$, \ $\nu_3=\nu_4=\nu_5=0$ \ if \ $\nu_1\nu_2\ne0$,
\item[$\circ$] $\nu_1=1$, \ $\nu_2=0$, \ $\nu_3=\nu_5=0$, \ $\nu_4\in\{0,1\}$ \ if \ $\nu_1\nu_2=0$, \ $(\nu_1,\nu_2)\ne(0,0)$,
\item[$\circ$] $\nu_3=1$, \ $|\nu_4|\leqslant1$, \ $\nu_5=0$ \ if \ $\nu_1=\nu_2=0$, \ $(\nu_3,\nu_4)\ne(0,0)$,
\item[$\circ$] $\nu_5=1$ \ otherwise.
\end{itemize}
After re-denoting the respective parameters,
this corresponds to the subalgebras $\mathfrak s_{2.5}^{\lambda\mu}$,
$\mathfrak s_{2.6}^{\lambda\delta}$, $\mathfrak s_{2.7}^{\lambda\nu}$ and $\mathfrak s_{2.8}^\lambda$.

If $\tau^1=\tau^2=0$, then $\lambda^2\ne0$ and
we follow the second case of the proof of Lemma~\ref{lem:dN1DInequivSubalgs} and set
$\chi^2=\rho^2=\alpha^2=\beta^2=\sigma^2=0$, which gives $Q^2=\lambda^2D^{\rm s}$.
The commutation relation $[Q^1,Q^2]=Q^1$ implies that there are three possible cases,
$\lambda^2=1$ and $\alpha^1=\beta^1=\sigma^1=0$,
$\lambda^2=\tfrac12$ and $\chi^1=\rho^1=\sigma^1=0$ or
$\lambda^2=\tfrac13$ and  $\chi^1=\rho^1=\alpha^1=\beta^1=0$.
Permuting $x$ and $y$ if necessary and acting by $\mathscr D^t_*(T)$
with an appropriate value of the parameter function~$T$, we can set
$\chi^1=1$, $\alpha^1=1$ or $\sigma^1\in\{t,1\}$
in the first, the second or the third cases,
which leads to the subalgebras \smash{$\mathfrak s_{2.9}^{\rho^1}$}, \smash{$\mathfrak s_{2.10}^{\beta^1}$}
or $\mathfrak s_{2.11}$ and $\mathfrak s_{2.12}$, respectively.

\medskip

\noindent {\bf II.}
Suppose that the subalgebra~$\mathfrak s_2$ is abelian, $[Q^1,Q^2]=0$.
Then the functions~$\tau^1$ and $\tau^2$ are necessarily linearly dependent.

Let in addition the tuples $(\tau^1,\lambda^1)$ and $(\tau^2,\lambda^2)$ are linearly independent.
Linearly combining~$Q^1$ and~$Q^2$, we can set $\tau^1\ne0$, $\lambda^1=0$, $\tau^2=0$ and $\lambda^2\ne0$.
The successive action by~$\mathscr D^t_*(T)$ with $T_t=1/\tau^1$ and the new value of~$\tau^1$
allows us to set $\tau^1=1$.
Following the second case of the proof of Lemma~\ref{lem:dN1DInequivSubalgs},
we set $\chi^2=\rho^2=\alpha^2=\beta^2=\sigma^2=0$.
Then the commutation relation $[Q^1,Q^2]=0$ implies $\chi^1=\rho^1=\alpha^1=\beta^1=\sigma^1=0$,
and thus we have the subalgebra~$\mathfrak s_{2.13}$.

If the tuples $(\tau^1,\lambda^1)$ and $(\tau^2,\lambda^2)$ are linearly dependent
and the functions~$\tau^1$ and $\tau^2$ do not simultaneously vanish,
then we linearly combine~$Q^1$ and~$Q^2$ to set $\tau^1\ne0$, $\tau^2=0$ and $\lambda^2=0$.
According to the first case of the proof of Lemma~\ref{lem:dN1DInequivSubalgs},
we can reduce~$Q^1$ to the form $D^t(1)+\lambda^1D^{\rm s}$.
In view of the commutation relation $[Q^1,Q^2]=0$,
the parameter functions in~$Q^2$ are
$\chi^2=\nu_1{\rm e}^{\lambda t}$,
$\rho^2=\nu_2{\rm e}^{\lambda t}$,
$\alpha^2=\nu_3{\rm e}^{2\lambda t}$,
$\beta^2=\nu_4{\rm e}^{2\lambda t}$ and
$\sigma^2=\nu_5{\rm e}^{3\lambda t}$
 with constants $\nu_1$, \dots, $\nu_5$.
The pushforward by an element~$\Phi$ of~$G$, which is necessarily of the form~\eqref{eq:dNPointSymForm},
does not change the form of the basis element~$Q^1$ up to its linearly combining with~$Q^2$
if and only if the parameters of~$\Phi$ are
\begin{gather*}
T=t+b,\quad
X^0=(\kappa_1+\kappa_0C\nu_1t){\rm e}^{\lambda t},\quad
Y^0=(\kappa_2+\kappa_0C\nu_2t){\rm e}^{\lambda t},\\
W^1=\left(\kappa_3+\kappa_0C^3\nu_3t-\lambda\kappa_0C^2\nu_1\kappa_1t
-\kappa_0^{\,2}C^3\nu_1^{\,2}t-\frac\lambda2\kappa_0^{\,2}C^3\nu_1^{\,2}t^2
\right){\rm e}^{2\lambda t},\\
W^2=\left(\kappa_4+\kappa_0C^3\nu_4t-\lambda\kappa_0C^2\nu_2\kappa_2t
-\kappa_0^{\,2}C^3\nu_2^{\,2}t-\frac\lambda2\kappa_0^{\,2}C^3\nu_2^{\,2}t^2
\right){\rm e}^{2\lambda t},\\
\begin{split}
W^0={}&\left(\kappa_5+\kappa_0(C^3\nu_5+\nu_1\kappa_3+\nu_2\kappa_4)t
+\frac12\kappa_0^{\,2}C^3(\nu_1\nu_3+\nu_2\nu_4)t^2
-\frac\lambda2\kappa_0^{\,2}C^2(\nu_1^{\,2}\kappa_1+\nu_2^{\,2}\kappa_2)t^2\right.\\&\left.
-\frac12\kappa_0^{\,3}C^3(\nu_1^{\,3}+\nu_2^{\,3})t^2
-\frac\lambda6\kappa_0^{\,3}C^3(\nu_1^{\,3}+\nu_2^{\,3})t^3
\right){\rm e}^{3\lambda t},
\end{split}
\end{gather*}
where $C$, $b$ and $\kappa_0$, \dots, $\kappa_5$ are arbitrary constants with $C\ne0$.
We can also push forward~$\mathfrak s_2$ by~$\mathscr J$
or multiply~$Q^2$ by an arbitrary nonzero constant.
Hence we can set
\begin{itemize}\itemsep=0ex
\item[$\circ$] $\nu_1=1$, \ $|\nu_2|\leqslant1$, \ $\nu_3=0$, \ $\nu_4\in\{0,1\}$, \ $\nu_5=0$ \ if \ $(\nu_1,\nu_2)\ne0$,
\item[$\circ$] $\nu_3=1$, \ $|\nu_4|\leqslant1$, \ $\nu_5=0$ \ if \ $\nu_1=\nu_2=0$, \ $(\nu_3,\nu_4)\ne0$,
\item[$\circ$] $\nu_5=1$ \ otherwise,
\end{itemize}
which corresponds, up to re-denoting parameters,
to the subalgebras $\mathfrak s_{2.14}^{\delta\nu\delta'}$,
$\mathfrak s_{2.15}^{\delta\nu}$ and $\mathfrak s_{2.16}^\delta$.

If $\tau^1=\tau^2=0$, then it follows from the commutation relation $[Q^1,Q^2]=0$ that
also $\lambda^1=\lambda^2=0$ since otherwise the vector fields~$Q^1$ and~$Q^2$ are linearly dependent,
and $\chi^1\chi^2_t-\chi^1_t\chi^2=0$, $\rho^1\rho^2_t-\rho^1_t\rho^2=0$,
$\chi^1\alpha^2-\chi^2\alpha^1+\rho^1\beta^2-\rho^2\beta^1=0$,
i.e., the parameter functions~$\chi^1$ and~$\chi^2$ (resp.\ $\rho^1$ and~$\rho^2$) are linearly dependent.
Suppose that the tuples $(\chi^1,\rho^1)$ and $(\chi^2,\rho^2)$ are linearly independent.
Linearly combining~$Q^1$ and~$Q^2$, we make $\chi^1\rho^2\ne0$ and~$\chi^2=\rho^1=0$.
Then the action by $\mathscr D^t_*(T)$ with $T_t=(\chi^1)^{-3}$ allows us to set $\chi^1=1$,
after which $\alpha^2=\rho^2\beta^1$,
and we obtain the subalgebra~$\mathfrak s_{2.17}^{\rho\alpha\beta}$.
If the tuples $(\chi^1,\rho^1)$ and $(\chi^2,\rho^2)$ are linearly dependent but not simultaneously zero,
then we linearly combine~$Q^1$ and~$Q^2$ and, if necessary, permute~$x$ and~$y$
to make $\chi^1\ne0$ and $\chi^2=\rho^2=0$.
Successively acting
by $\mathscr D^t_*(T)$ with $T_t=(\chi^1)^{-3}$,
by $\mathscr P^x_*(X^0)$ with $X^0_t=\alpha^1$ and
by $\mathscr R^x_*(W^1)$ with $W^1=-\sigma^1$, we set $\chi^1=1$, $\alpha^1=0$ and $\sigma^1=0$,
which results in $\alpha^2=-\rho^1\beta^2$.
The further simplifications are $\beta^1=0$ if $\rho^1\ne0$,
$\sigma^2=0$ if $\rho^1=0$ and $\beta^2\ne0$,
and no meaningful simplification is possible if $\rho^1=\beta^2=0$.
This gives the subalgebras
$\mathfrak s_{2.18}^{\rho\beta\sigma}$,
$\mathfrak s_{2.19}^{\beta^1\beta^2}$ and
$\mathfrak s_{2.20}^{\beta\sigma}$, respectively.
In the case $\chi^1=\rho^1=\chi^2=\rho^2=0$,
the consideration splits according to the additional conditions that
\begin{itemize}\itemsep=0ex
\item[$\circ$] $\alpha^1\beta^2\ne\alpha^2\beta^1$,
\item[$\circ$] $\alpha^1\beta^2=\alpha^2\beta^1$ but the tuples $(\alpha^1,\beta^1)$ and $(\alpha^2,\beta^2)$ are linearly independent,
\item[$\circ$] the tuples $(\alpha^1,\beta^1)$ and $(\alpha^2,\beta^2)$ are linearly dependent but not simultaneously zero,
\item[$\circ$] $\alpha^1=\beta^1=\alpha^2=\beta^2=0$, and $\sigma^1_t$ and $\sigma^2_t$ are linearly independent,
\item[$\circ$] $\alpha^1=\beta^1=\alpha^2=\beta^2=0$, and $\sigma^1_t$ and $\sigma^2_t$ are linearly dependent,
\end{itemize}
and after obvious simplifications,
we obtain the subalgebras
$\mathfrak s_{2.21}^{\alpha\beta^1\beta^2}$,
$\mathfrak s_{2.22}^{\alpha\beta\sigma}$,
$\mathfrak s_{2.23}^{\beta\sigma}$,
$\mathfrak s_{2.24}^\sigma$ and
$\mathfrak s_{2.25}^\sigma$.
\end{proof}

\begin{remark}
The statements of Lemmas~\ref{lem:dN1DInequivSubalgs} and~\ref{lem:dN2DInequivSubalgs}
should be interpreted in the following way.
Subalgebras from different families or within each of the families parameterized only by constants
are definitely $G$-inequivalent.
At the same time, there is an inessential equivalence between subalgebras within each of the families parameterized by functions
that does not allow us to further simplify the general form of subalgebras from the family.
For example, subalgebras~$\mathfrak s_{1.3}^\rho$ and~$\smash{\mathfrak s_{1.3}^{\tilde\rho}}$
are $G$-inequivalent if and only if
$\tilde\rho(t)=\rho(at+b)$ for some $a,b\in\mathbb R$ or
$\tilde\rho=(\rho(\hat T))^{-1}$, where $\hat T$ is the inverse of a solution~$T$ of the equation $T_t=c\rho^{-3}$ for some $c\in\mathbb R$.
\end{remark}

\begin{remark}
For the purpose of Lie reduction of the equation~\eqref{eq:dN}
to differential equations with less number of independent variables,
it would suffice to only classify one-dimensional subalgebras of rank one
and two-dimensional subalgebras of rank two.
Nevertheless, we decided to present the respective complete classifications
since they require not much more effort than the above partial classifications do.
Moreover, this is instructive given the fact
that the number of correct classifications of subalgebras of Lie algebras
(especially infinite-dimensional ones) in the literature is not great.
\end{remark}

Due to the analogy of the structures of $(\mathfrak g,G)$ and $(\mathfrak g_{\rm L},G_{\rm L})$,
we can easily obtain the classifications of one- and two-dimensional subalgebras of the algebra~$\mathfrak g_{\rm L}$
using Lemmas~\ref{lem:dN1DInequivSubalgs} and~\ref{lem:dN2DInequivSubalgs}, respectively.
Here we consider only one-dimensional subalgebras of~$\mathfrak g_{\rm L}$.

\begin{lemma}\label{lem:dNLaxPair1DInequivSubalgs}
A complete list of $G_{\rm L}$-inequivalent one-dimensional subalgebras of the algebra~$\mathfrak g_{\rm L}$
is exhausted by the following algebras:
\begin{gather*}
\bar{\mathfrak s}_{1.1}^{\delta\delta'}=\big\langle \bar D^t(1)+\delta\bar D^{\rm s}+\delta'\bar P^v\big\rangle_{\delta\delta'=0},\quad
\bar{\mathfrak s}_{1.2}                =\big\langle \bar D^{\rm s}                                  \big\rangle,\quad
\bar{\mathfrak s}_{1.3}^{\rho\delta}   =\big\langle \bar P^x(1)+\bar P^y(\rho)+\delta\bar P^v       \big\rangle,\\
\bar{\mathfrak s}_{1.4}^{\beta\delta}  =\big\langle \bar P^x(1)+\bar R^y(\beta)+\delta\bar P^v      \big\rangle,\quad
\bar{\mathfrak s}_{1.5}^{\beta\delta}  =\big\langle \bar R^x(1)+\bar R^y(\beta)+\delta\bar P^v      \big\rangle,\\
\bar{\mathfrak s}_{1.6}^{\delta}       =\big\langle \bar Z(t)+\delta\bar P^v                        \big\rangle,\quad
\bar{\mathfrak s}_{1.7}^{\delta}       =\big\langle \bar Z(1)+\delta\bar P^v                        \big\rangle,\quad
\bar{\mathfrak s}_{1.8}                =\big\langle \bar P^v                                        \big\rangle,
\end{gather*}
where $\delta,\delta'\in\{0,1\}$, and $\rho$ and $\beta$ run through the set of smooth functions of~$t$ with $\rho\ne0$.
\end{lemma}

\section{Trivial solutions}\label{sec:dNTrivialSolutions}

It is obvious that the equation~\eqref{eq:dN} is identically satisfied
on the solution set of the differential constraint $u_{xy}=0$
or, equivalently, on the set of functions of~$(t,x,y)$
with additive separation of the variables~$x$ and~$y$.
In other words, the equation~\eqref{eq:dN} has the solutions of the form
\begin{gather}\label{eq:dNTrivialSolutions}
\solution u=w(t,x)+\tilde w(t,y),
\end{gather}
where $w$ and~$\tilde w$ are sufficiently smooth functions of their arguments.
Calling these solutions trivial is justified by the fact that
the equation~\eqref{eq:dN} is a potential equation for the dispersionless Nizhnik system
$p_t=(h^1p)_x+(h^2p)_y$, $h^1_y=p_x$, $h^2_x=p_y$
with the relation $p=u_{xy}$, $h^1=u_{xx}$, $h^2=u_{yy}$,
and thus a solution is of the form~\eqref{eq:dNTrivialSolutions} for the equation~\eqref{eq:dN}
if and only if it corresponds to a solution of the dispersionless Nizhnik system with zero principal component~$p$.

Within the family~\eqref{eq:dNTrivialSolutions},
there is the subfamily of solutions satisfying the differential constraints
$u_{xy}=u_{xxxx}=u_{yyyy}=0$ and $u_{xxx}=u_{yyy}$ and thus having the form
\begin{gather}\label{eq:dNMoreTrivialSolutions}
u=W^5(t)(x^3+y^3)+W^3(t)x^2+W^4(t)y^2+W^1(t)x+W^2(t)y+W^0(t),
\end{gather}
where the coefficients~$W^0$, \dots, $W^5$ are arbitrary sufficiently smooth functions of~$t$.
The solutions from the subfamily~\eqref{eq:dNMoreTrivialSolutions} are even more trivial
than general elements of the family~\eqref{eq:dNTrivialSolutions}
since each solution of the form~\eqref{eq:dNMoreTrivialSolutions}
is $G$-equivalent to the constant zero solution $u=0$.

The above trivial solutions of the equation~\eqref{eq:dN}
often arise in the course of its Lie reductions.
Identifying such solutions among constructed ones and neglecting them
in addition to listing solutions up to the $G$-equivalence allow us
to better arrange the found families of invariant solutions.
Note that modulo the $G$-equivalence,
we can arbitrarily change or neglect summands of the form $W^1(t)x+W^2(t)y+W^0(t)$
in any solution of the equation~\eqref{eq:dN}.

\section{Lie reductions of codimension one}\label{sec:dNLieReductionsOfCodim1}

Among subalgebras listed in Lemma~\ref{lem:dN1DInequivSubalgs},
only subalgebras~$\mathfrak s_{1.1}^\delta$, $\mathfrak s_{1.2}$, $\mathfrak s_{1.3}^\rho$ and~$\mathfrak s_{1.4}^\beta$
are appropriate to be used for Lie reduction of the equation~\eqref{eq:dN}.
We collect $G$-inequivalent codimension-one Lie reductions of the equation~\eqref{eq:dN} in Table~\ref{tab:LieReductionsOfCodim1}.
There, for each of the above one-dimensional subalgebras of~$\mathfrak g$,
we present a constructed ansatz for $u$,
the corresponding reduced partial differential equation in two independent variables,
where $w=w(z_1,z_2)$ is the new unknown function of the invariant independent variables~$(z_1,z_2)$.
The subscripts~1 and~2 of~$w$ denote the differentiation with respect to~$z_1$ and~$z_2$, respectively.

\begin{table}[!ht]
\begin{center}\small
\caption{\footnotesize $G$-inequivalent Lie reductions with respect to one-dimensional subalgebras of~$\mathfrak g$.
\strut}\label{tab:LieReductionsOfCodim1}
\renewcommand{\arraystretch}{1.5}
\begin{tabular}{|l|c|c|c|l|}
\hline
\hfil $\subset\mathfrak g$ &  $u$      &   $z_1$ &  $z_2$        &\hfil Reduced equation      \\
\hline
$\mathfrak s_{1.1}^\delta$ & ${\rm e}^{3\delta t}w-\frac16\delta(x^3+y^3)$ & ${\rm e}^{-\delta t}x$ & ${\rm e}^{-\delta t}y$ &
$(w_{11}w_{12})_1+(w_{12}w_{22})_2=3\delta w_{12}$\\
$\mathfrak s_{1.2}$        & $x^3w$ &  $t$ & $y/x$  & $(z_2w_{22}-2w_2)_1=\big((z_2w_{22}-2w_2)w_{22}\big)_2$ \\[-.7ex]
&&&&\ \ ${}-(z_2\p_2-2)\big((z_2w_{22}-2w_2)(z_2^{\,2}w_{22}-4z_2w_2+6w)\big)$\\
$\mathfrak s_{1.3}^\rho$   & $w-\frac16\rho_t\rho^{-1}y^3$ &  $t$ & $\rho^{-1}y-x$  &  $w_{122}+2(1-\rho^{-3})w_{22}w_{222}=0$\\
$\mathfrak s_{1.4}^\beta$  & $w+\beta xy$ &  $t$ & $y$ & $\beta w_{222}=\beta_1$\\
\hline
\end{tabular}
\end{center}
\end{table}

We study each of the listed reduced equations separately,
indexing it by the number of the corresponding one-dimensional subalgebra of~$\mathfrak g$.

\medskip\par\noindent{\bf 1.1.}\
$\mathfrak s_{1.1}^\delta=\big\langle D^t(1)+\delta D^{\rm s}\big\rangle$, $\delta\in\{0,1\}$ $(\!{}\bmod G)$.
The maximal Lie invariance algebra of reduced equation 1.1$^\delta$ is%
\footnote{%
In contrast to reduced equation 1.1$^1$,
its counterpart with $\delta=0$ loses the property of maximal rank
on the submanifold~$\mathcal M_0$ of the manifold~$\mathcal M$.
Here the manifold~$\mathcal M$ is defined by this equation
in the jet space ${\rm J}^3(\mathbb R^2_{z_1,z_2}\times\mathbb R_w)$ and
the submanifold~$\mathcal M_0$ is singled out in~$\mathcal M$
by the consistent system $w_{11}=w_{12}=w_{22}=0$, $w_{112}=w_{122}=w_{111}+w_{222}=0$.
It can be proved by the classical infinitesimal method that
the maximal Lie invariance algebra of the complement $\mathcal M\setminus\mathcal M_0$
of~$\mathcal M_0$ in~$\mathcal M$ coincides with the algebra~$\mathfrak a_{1.1}^0$.
At the same time, the submanifold~$\mathcal M_0$ is also invariant with respect to this algebra.
Therefore, the maximal Lie invariance algebra of reduced equation 1.1$^0$
is indeed the algebra~$\mathfrak a_{1.1}^0$.
}
\begin{gather*}
\mathfrak a_{1.1}^0=\langle D^z,\,w\p_w,\,\p_{z_1},\,\p_{z_2},\,z_1\p_w,\,z_2\p_w,\,\p_w\rangle
\quad\mbox{if}\quad \delta=0,
\\
\mathfrak a_{1.1}^1=\langle\tilde D^z,\,\p_{z_1},\,\p_{z_2},\,z_1\p_w,\,z_2\p_w,\,\p_w\rangle
\quad\mbox{if}\quad \delta=1.
\end{gather*}
Here $D^z:=z_1\p_{z_1}+z_2\p_{z_2}$ and $\tilde D^z:=z_1\p_{z_1}+z_2\p_{z_2}+3w\p_w$.
All their elements are induced by Lie symmetries of the original equation~\eqref{eq:dN}.
Indeed, the normalizer of the subalgebra~$\mathfrak s_{1.1}^\delta$ in~$\mathfrak g$ is
\begin{gather*}
{\rm N}_{\mathfrak g}(\mathfrak s_{1.1}^0)=\langle D^t(1),D^t(t),D^{\rm s},P^x(1),P^y(1),R^x(1),R^y(1),Z(1)\rangle
\quad\mbox{if}\quad \delta=0,
\\
{\rm N}_{\mathfrak g}(\mathfrak s_{1.1}^1)=\langle D^t(1),D^{\rm s},
P^x({\rm e}^t),P^y({\rm e}^t),R^x({\rm e}^{2t}),R^y({\rm e}^{2t}),Z({\rm e}^{3t})\rangle
\quad\mbox{if}\quad \delta=1,
\end{gather*}
see a similar computation in \cite[Section~3]{vane2021a}.
The Lie-symmetry vector fields $D^t(1)+\delta D^{\rm s}$, $D^{\rm s}$,
$P^x({\rm e}^{\delta t})$, $P^y({\rm e}^{\delta t})$, $R^x({\rm e}^{2\delta t})$, $R^y({\rm e}^{2\delta t})$, $Z({\rm e}^{3\delta t})$
and, for $\delta=0$, $3D^t(t)$
of the equation~\eqref{eq:dN}
induce the Lie-symmetry vector fields 0, $\tilde D^z$,
$\p_{z_1}$, $\p_{z_2}$, $z_1\p_w$, $z_2\p_w$, $\p_w$ and, for $\delta=0$, $D^z$
of reduced equation 1.1, respectively.

Therefore, any two-step Lie reduction of the equation~\eqref{eq:dN} to an ordinary differential equation,
where the first step is reduction~1.1 and the second step is a Lie reduction of reduced equation~1.1,
is equivalent to a direct Lie reduction to an ordinary differential equation
using a two-dimensional subalgebra of~$\mathfrak g$.
This means that there is no need to carry out Lie reductions of reduced equation~1.1.

For each $\delta\in\{0,1\}$,
let us compute the point-symmetry group~$G_{1.1}^\delta$ of the reduced equation 1.1$^\delta$
by the algebraic method.
Up to the antisymmetry of the Lie bracket,
the nonzero commutation relations between the basis vector fields of the algebra~$\mathfrak a:=\mathfrak a_{1.1}^\delta$
are exhausted by
\begin{gather*}
\begin{split}&
[D^z,\p_{z_1}]=-\p_{z_1},\quad
[D^z,\p_{z_2}]=-\p_{z_2},\quad
[D^z,z_1\p_w]=z_1\p_w,\quad
[D^z,z_2\p_w]=z_2\p_w,\\[.5ex]&
[w\p_w,z_1\p_w]=-z_1\p_w,\quad
[w\p_w,z_2\p_w]=-z_2\p_w,\quad
[w\p_w,\p_w]=-\p_w,\\[.5ex]&
[\p_{z_1},z_1\p_w]=\p_w,\quad
[\p_{z_2},z_2\p_w]=\p_w,
\end{split}
\end{gather*}
and
\begin{gather*}
\begin{split}&
[\tilde D^z,\p_{z_1}]=-\p_{z_1},\quad
[\tilde D^z,\p_{z_2}]=-\p_{z_2},\quad
[\tilde D^z,z_1\p_w]=-2z_1\p_w,\quad
[\tilde D^z,z_2\p_w]=-2z_2\p_w,\\[.5ex]&
[\tilde D^z,\p_w]=-3\p_w,\quad
[\p_{z_1},z_1\p_w]=\p_w,\quad
[\p_{z_2},z_2\p_w]=\p_w
\end{split}
\end{gather*}
if $\delta=0$ and $\delta=1$, respectively.
We first find megaideals of the algebra~$\mathfrak a$
applying techniques that do not require the knowledge of the automorphism group ${\rm Aut}(\mathfrak a)$ \cite{bihl2015a,popo2003a}.
Then we use the constructed megaideals for simplifying the computation of ${\rm Aut}(\mathfrak a)$
and obtain the remaining megaideals.
As a result, the complete list of proper megaideals of~$\mathfrak a$ is as follows:
\begin{gather*}
\mathfrak m_1:=\mathfrak a'=\big\langle\p_{z_1},\,\p_{z_2},\,z_1\p_w,\,z_2\p_w,\,\p_w\big\rangle,\quad
\mathfrak m_2:=\mathfrak a''=\mathfrak z(\mathfrak m_1)=\big\langle\p_w\big\rangle,
\\
\mathfrak m_3:=\mathrm C_{\mathfrak a}(\mathfrak m_2)=\big\langle D^z\big\rangle\dotplus\mathfrak m_1,\quad
\mathfrak m_4:=\big\langle D^z+2w\p_w\big\rangle\dotplus\mathfrak m_1 \quad\mbox{if}\quad\delta=0,
\\
\mathfrak m_3:=\big\langle z_1\p_w,\,z_2\p_w,\,\p_w\big\rangle,\quad
\mathfrak m_4:=\big\langle \p_{z_1},\,\p_{z_2},\,\p_w\big\rangle \quad\mbox{if}\quad\delta=1.
\end{gather*}
Denote $\mathfrak m_0:=\mathfrak a$.
Let a point transformation~$\Phi$: $(\tilde z_1,\tilde z_2,\tilde w)=(Z^1,Z^2,W)$
in the space with the coordinates $(z_1,z_2,w)$,
where $(Z^1,Z^2,W)$ is a tuple of smooth functions of $(z_1,z_2,w)$ with nonvanishing Jacobian,
preserve the equation~\eqref{eq:dN}.
Necessary conditions for this are $\Phi_*\mathfrak m_k\subseteq\mathfrak m_k$, $k=0,\dots,4$.
Hereafter the indices~$i$ and~$j$ run from~1 to~2, and we assume summation with respect to repeated indices.
The conditions $\Phi_*\p_w\in\mathfrak m_2$ and $\Phi_*(z_i\p_w)\in\mathfrak m_1$
imply that
\[
Z^i=a_{ij}z_j+a_{i0}, \quad W=cw+W^0(z_1,z_2),
\]
where~$a_{ij}$, $a_{i0}$ and~$c$ are constants with $c\det(a_{ij})\ne0$,
and $W^0$ is a smooth function of~$(z_1,z_2)$.
Then the conditions
$\Phi_*\p_{z_i}\in\mathfrak m_4$ if $\delta=1$ or
$\Phi_*\p_{z_i}\in\mathfrak m_1$ and $\Phi_*D^z\in\mathfrak m_3$ if $\delta=0$
in addition give that $W^0=b_iz_i+b_0$ for some constants~$b_i$ and~$b_0$.
Since no further constraints on~$\Phi$ can be derived within the framework of the algebraic method,
we continue the computation with the direct method, obtaining
$a_{11}=a_{22}$, $a_{12}=a_{21}$, $a_{11}a_{12}=0$, $(a_{11},a_{12})\ne(0,0)$
and, if $\delta=1$, $c=a^3$, where $a$ is the nonzero value among~$a_{11}$ and~$a_{12}$.
This means that there are exactly two independent,
up to composing with each other and with continuous point symmetry transformations of the equation~\eqref{eq:dN},
discrete point symmetry transformations of this equation.
They are the involutions, the one that permutes the independent variables~$z_1$ and~$z_2$,
${(\tilde z_1,\tilde z_2,\tilde w)=(z_2,z_1,w)}$,
and the one that simultaneously alternates the signs of all the variables,
$(\tilde z_1,\tilde z_2,\tilde w)=(-z_1,-z_2,-w)$.
These transformations are respectively induced by the discrete point symmetries~$\mathscr J$ and~$\mathscr I^{\rm s}$
of the original equation~\eqref{eq:dN}.
The fact of inducing all Lie symmetries of reduced equations~1.1$^\delta$
follows from that for the corresponding Lie invariance algebras~$\mathfrak a_{1.1}^\delta$.
Therefore, for each $\delta\in\{0,1\}$
the group~$G_{1.1}^\delta$ is entirely induced by the stabilizer of~$\mathfrak s_{1.1}^\delta$ in~$G$.

The subalgebra~$\mathfrak s_{1.1}^\delta$ of~$\mathfrak g$
is associated with the subalgebra(s)~$\bar{\mathfrak s}_{1.1}^{\delta\delta'}$ of~$\mathfrak g_{\rm L}$
with $\delta\delta'=0$.
The extension of ansatz~1.1 to~$v$ is \smash{$v={\rm e}^{\frac32\delta t}q+\delta't$}.
Here and in the next case,
$q=q(z_1,z_2)$ is the invariant unknown function that replaces~$v$,
and, as for~$w$, the subscripts~1 and~2 of~$q$ denote the differentiation with respect to~$z_1$ and~$z_2$, respectively.

The corresponding family of reduced systems for the nonlinear Lax representation~\eqref{eq:dNLaxPair}
is associated with the subalgebra family~$\bar{\mathfrak s}_{1.1}^{\delta\delta'}$
from Lemma~\ref{lem:dNLaxPair1DInequivSubalgs}
and consists of the systems
\begin{gather}\label{eq:dNReducedEq1.1Lax}
\frac13(q_1^3+q_2^3)+q_1w_{11}+q_2w_{22}=\frac32\delta q+\delta',\quad w_{12}=-q_1q_2,
\end{gather}
each of which can be interpreted 
as a nonlinear Lax representation for reduced equation~1.1 with the same value of~$\delta$,
cf.\ the introductive part of~\cite[Section~4.1]{moro2021a}.
In other words, reduced equation~1.1 is the compatibility condition of the system~\eqref{eq:dNReducedEq1.1Lax}
with respect to~$q$.
Note that for~$\delta=0$ we thus construct two inequivalent nonlinear Lax representations,
with $\delta'=0$ and with $\delta'=1$.

Reduced equation~1.1$^0$ is just the stationary dispersionless Nizhnik equation.
Note that its counterpart with dispersion was studied in~\cite{marv2003a}.

\begin{remark}\label{rem:OnInessEquivOfMultistepReductions}
A complete list of $G_{1.1}^1$-inequivalent one-dimensional subalgebras of the algebra~$\mathfrak a_{1.1}^1$
is exhausted by the subalgebras
\[
\mathfrak b_1=\langle\tilde D^z\rangle,\quad
\mathfrak b_2^{\nu\kappa\varsigma}=\langle\p_{z_1}+\nu\p_{z_2}+\kappa z_1\p_w+\varsigma z_2\p_w\rangle,\quad
\mathfrak b_3^\nu=\langle z_1\p_w+\nu z_2\p_w\rangle,\quad
\mathfrak b_4=\langle\p_w\rangle,
\]
where $\nu\in[-1,1]$, and $\varsigma\ne0$ if $\nu\in\{-1,1\}$ and $(\kappa,\varsigma)\ne(0,0)$,
cf.\ \cite[Proposition~2]{moro2021a}.
A further gauging of subalgebra parameters is possible only within the second family $\{\mathfrak b_2^{\nu\kappa\varsigma}\}$,
where one of the parameters~$\kappa$ or~$\varsigma$, if nonzero, can be set to be equal to~1 up to the $G_{1.1}^1$-equivalence.
At the same time, the subalgebra~$\mathfrak b_2^{\nu\kappa\varsigma}$ is induced by
the subalgebra
\[
\check{\mathfrak b}_2^{\nu\kappa\varsigma}=\langle D^t(1)+\delta D^{\rm s},\,
P^x({\rm e}^{\delta t})+\nu P^y({\rm e}^{\delta t})+\kappa R^x({\rm e}^{2\delta t})+\varsigma R^y({\rm e}^{2\delta t})
\rangle$ of~$\mathfrak g,
\]
which is $G$-equivalent to the subalgebra~$\mathfrak s_{2.14}^{\delta\nu\delta'}$,
where $\delta'=0$ if $\kappa=\varsigma$ and $\delta'=1$ otherwise.
In other words, if $\tilde\nu=\nu$, the tuples $(\kappa,\varsigma)$ and $(\tilde\kappa,\tilde\varsigma)$
are not proportional with a nonzero multipliers and
$\kappa-\varsigma$ and $\tilde\kappa-\tilde\varsigma$ are simultaneously either are equal to zero or are not,
then the subalgebras~$\mathfrak b_2^{\nu\kappa\varsigma}$ and~\smash{$\mathfrak b_2^{\tilde\nu\tilde\kappa\tilde\varsigma}$}
of~$\mathfrak a_{1.1}^1$ are $G_{1.1}^1$-inequivalent,
whereas the associated subalgebras
$\check{\mathfrak b}_2^{\nu\kappa\varsigma}$ and~\smash{$\check{\mathfrak b}_2^{\tilde\nu\tilde\kappa\tilde\varsigma}$}
are $G$-equivalent.
This is why the inequivalent two-step reductions, where the first step is reduction~1.1$^1$
and the second step is based on subalgebras~$\mathfrak b_2^{\nu\kappa\varsigma}$
and~\smash{$\mathfrak b_2^{\tilde\nu\tilde\kappa\tilde\varsigma}$ of~$\mathfrak a_{1.1}^1$}
with the above constraints on the subalgebra parameters,
definitely results in $G$-equivalent families of invariant solutions
of the dispersionless Nizhnik equation~\eqref{eq:dN}, cf.\ \cite[Section~4.1.1.2]{moro2021a}.
The above phenomenon has not been described in the literature.
\end{remark}

\noindent{\bf 1.2.}\
$\mathfrak s_{1.2}=\big\langle D^{\rm s}\big\rangle$.
The same conclusion on the superfluousness of two-step Lie reduction can be made for reduced equation 1.2.
Its maximal Lie invariance algebra is
\begin{gather*}
\mathfrak a_{1.2}=\big\langle\breve D(\tau)\big\rangle
\quad\mbox{with}\quad \breve D(\tau):=\tau\p_{z_1}-\left(\tau_1w+\frac1{18}\tau_{11}(z_2^{\,3}+1)\right)\p_w.
\end{gather*}
Here and in what follows the parameter function~$\tau$ runs through the set of smooth functions of~$z_1$.
It is obvious that the normalizer of the subalgebra~$\mathfrak s_{1.2}$ in~$\mathfrak g$ is
$
{\rm N}_{\mathfrak g}(\mathfrak s_{1.2})=\langle D^t(\tau),D^{\rm s}\rangle.
$
The Lie-symmetry vector field $D^t(\tau)$ of the equation~\eqref{eq:dN}
induces the element of~$\mathfrak a_{1.2}$ with the same value of the parameter function~$\tau$,
whereas $D^{\rm s}$ is mapped to 0.
In other words, the entire maximal Lie invariance algebra~$\mathfrak a_{1.2}$ of reduced equation 1.2
is induced by ${\rm N}_{\mathfrak g}(\mathfrak s_{1.2})$.
Therefore, similarly to reduced equation~1.1,
further Lie reductions of reduced equation 1.2 are needless.

In fact, all point symmetries of reduced equation 1.2
are induced by point symmetries of the original equation~\eqref{eq:dN}.
To show this, we compute the point-symmetry group  of reduced equation~1.2
using the most general version of the algebraic method.
Again, consider a point transformation~$\Phi$: $(\tilde z_1,\tilde z_2,\tilde w)=(Z^1,Z^2,W)$
in the space with the coordinates $(z_1,z_2,w)$,
where $(Z^1,Z^2,W)$ is a tuple of smooth functions of $(z_1,z_2,w)$ with nonvanishing Jacobian.
The algebra~$\mathfrak a:=\mathfrak a_{1.2}$ is infinite-dimensional and contains no proper megaideals.
This is why the only convenient necessary condition for the transformation~$\Phi$ to preserve the equation~\eqref{eq:dN}
is $\Phi_*\mathfrak a\subseteq\mathfrak a$,
which expands to $\Phi_*\breve D(\tau)=\breve D(\tilde\tau)$.
Componentwise splitting the latter condition with each of the specific values $\tau=z_1^{\,i}$, $i=0,\dots,3$,
$\Phi_*\breve D(z_1^{\,i})=\breve D(\tilde\tau^i)$,
and recombining the derived determining equations for the components of~$\Phi$,
we in particular obtain the equation
$\tilde\tau^3-3z_1\tilde\tau^2+3z_1^{\,2}\tilde\tau^1-z_1^{\,3}\tilde\tau^0=0$.
Since at most one function among $\tilde\tau^i$, $i=0,\dots,3$, can be constant,
this equation can be solved with respect to~$Z^1$, giving $Z^1=Z^1(z_1)$.
It is obvious that reduced equation~1.2 admits the discrete point symmetry
$\breve{\mathscr I}^1$: $(\tilde z_1,\tilde z_2,\tilde w)=(-z_1,z_2,-w)$,
which is induced by the discrete point symmetry~$\mathscr I^{\rm i}\circ\mathscr I^{\rm s}$ of the equation~\eqref{eq:dN}.
Up to factoring out the transformation~$\breve{\mathscr I}^1$ and Lie symmetries of reduced equation~1.2,
each of which is also induced, we can assume that $Z^1=z_1$.
For this restricted form of~$\Phi$, we have $\Phi_*\breve D(\tau)=\breve D(\tau)$.
Splitting this condition componentwise and with respect to the parameter function~$\tau$
and its derivatives~$\tau_{z_1}$ and~$\tau_{z_1z_1}$ leads to the equations
$Z^2_{z_1}=Z^2_w=0$, $W_{z_1}=0$, $wW_w=W$ and $(z_2^{\,3}+1)W_w=(Z^2)^3+1$.
Therefore, $Z^2=Z^2(z_2)$ and $W=W^1(z_2)w$ with $Z^2_{z_2}\ne0$ and $W^1:=\big((Z^2)^3+1\big)/(z_2^{\,3}+1)$.
To derive more constraints on~$\Phi$, we should continue the computation with the direct method.
This only gives two solutions, $Z^2=z_2$ and $Z^2=1/z_2$,
which correspond to the identity transformation and the discrete point symmetry
$\breve{\mathscr J}$: $(\tilde z_1,\tilde z_2,\tilde w)=(z_1,z_2^{-1},z_2^{-3}w)$.
The transformation~$\breve{\mathscr J}$ is induced by the discrete point symmetry~$\mathscr J$ of the equation~\eqref{eq:dN}.
Therefore, the entire point-symmetry group of reduced equation~1.2
is induced by the stabilizer of~$\mathfrak s_{1.2}$ in~$G$.

The associated subalgebra of~$\mathfrak g_{\rm L}$ for the subalgebra~$\mathfrak s_{1.2}$ of~$\mathfrak g$
is $\bar{\mathfrak s}_{1.2}$.
The $v$-component of the extension of ansatz~1.2 is $v=|x|^{3/2}q$.
The corresponding reduced system for the nonlinear Lax representation~\eqref{eq:dNLaxPair}~is
\begin{gather*}
q_1=\frac\varepsilon3\left(\frac32q-z_2q_2\right)^3+\frac\varepsilon3q_2^{\,3}
+(z_2^{\,2}w_{22}-4z_2w_2+6w)\left(\frac32q-z_2q_2\right)+w_{22}q_2,
\\
\varepsilon q_2\left(\frac32q-z_2q_2\right)=z_2w_{22}-2w_2
\end{gather*}
with $\varepsilon=\sgn x$,
which can be interpreted, after solving with respect to~$(q_1,q_2)$,
as a nonlinear Lax representation for reduced equation~1.2,
cf.\ \cite[Section~4.2]{moro2021a} up to typos.

\medskip\par\noindent{\bf 1.3.}\
$\mathfrak s_{1.3}^\rho=\big\langle P^x(1)+P^y(\rho)\big\rangle$ with $\rho=\rho(t)\ne0$.

If $\rho\equiv1$, then reduced equation~1.3 degenerates to $w_{122}=0$,
and its general solution is $w=f(z_2)+\varrho^1(z_1)z_2+\varrho^0(z_1)$,
where~$\varrho^0$ and~$\varrho^1$ are arbitrary functions of~$z_1=t$,
and $f$ is an arbitrary sufficiently smooth function of~$z_2=y-x$, cf.\ \cite[Eq.~(60)]{moro2021a}.
Up to the $G$-equivalence, the coefficients~$\varrho^0$ and~$\varrho^1$ can be assumed to vanish.
The corresponding family of solutions of~\eqref{eq:dN}~is
\begin{gather}\label{eq:dNs1.3Rho1InvarSolutions}
\solution u=f(y-x).
\end{gather}

For any $\rho$ with $\rho\not\equiv1$,
meaning that $\rho\not\equiv1$ on each open interval of the domain of~$\rho$, we can use
the change of independent variables $\tilde z_1=2\int(1-\rho^{-3})\,{\rm d}z_1$, $\tilde z_2=z_2$
for modifying ansatz~1.3$^\rho$ and reduced equation~1.3$^\rho$ to the form
\begin{gather}\nonumber
u=w(\tilde z_1,\tilde z_2)-\frac{\rho_t}{6\rho}y^3, \quad
\tilde z_1=2\int\frac{\rho^3-1}{\rho^3}{\rm d}t,\quad
\tilde z_2=\frac y\rho-x,
\\\label{eq:dNs1.3RhoNe1ModifiedRedEq}
w_{\tilde z_1\tilde z_2\tilde z_2}+w_{\tilde z_2\tilde z_2}w_{\tilde z_2\tilde z_2\tilde z_2}=0.
\end{gather}
Thus, the class of reduced equations~1.3$^\rho$
associated with the subalgebra family $\{\mathfrak s_{1.3}^\rho\mid\rho\not\equiv1\}$
collapses to the unary class, whose single element is the equation~\eqref{eq:dNs1.3RhoNe1ModifiedRedEq}.
In other words, the $G$-inequivalent subalgebras $\mathfrak s_{1.3}^\rho$ with $\rho\not\equiv1$
lead to pairwise similar reduced equations, which take the same form~\eqref{eq:dNs1.3RhoNe1ModifiedRedEq}
if appropriate ansatzes are chosen.
Nevertheless, we prefer to use ansatzes~1.3$^\rho$ from Table~\ref{tab:LieReductionsOfCodim1}
since otherwise Case~1.3 splits into two cases,
and without the above explanation, the modified ansatz looks artificial.

The substitution~$w_{\tilde z_2\tilde z_2}=h$ maps the modified reduced equation~\eqref{eq:dNs1.3RhoNe1ModifiedRedEq}
to the inviscid Burgers equation \[h_{\tilde z_1}+hh_{\tilde z_2}=0,\]
which is the simplest nonlinear transport equation, called also Hopf's equation.
An implicit representation of the general solution of this equation is well known, $F(h,\tilde z_2-h\tilde z_1)=0$,
where $F$ is an arbitrary nonconstant sufficiently smooth function of its arguments.
Modulo the $G$-equivalence, we can assume that $w$ is a fixed second antiderivative of~$h$ with respect to~$\tilde z_2$.
As a result, we construct a family of solutions of~\eqref{eq:dN}
expressed in terms of quadratures with an implicitly defined function,
\begin{gather}\label{eq:dNs1.3InvarSolutions}
\solution u=
\int\left(\int h(\tilde z_1,\tilde z_2)\,{\rm d}\tilde z_2\right){\rm d}\tilde z_2
-\frac{\rho_t}{6\rho}y^3, \quad
\tilde z_1:=2\int\frac{\rho^3-1}{\rho^3}{\rm d}t,\quad
\tilde z_2:=\frac y\rho-x,
\end{gather}
where $\rho$ is an arbitrary sufficiently smooth function of~$t$
that does not coincide with the constant functions~0 and~1,
and the function $h=h(\tilde z_1,\tilde z_2)$ is implicitly defined by the equation ${F(h,\tilde z_2-h\tilde z_1)=0}$
with an arbitrary nonconstant sufficiently smooth function~$F$ of its arguments.

Lie and generalized symmetries, cosymmetries, conservation-law characteristics and conservation laws
of the inviscid Burgers equation were exhaustively described in Sections~3 and~6 of~\cite{vane2021a},
see also \cite[Appendix]{baik1991b} for the first computation of the generalized symmetries of this equation.
Via the substitution~$h(\tilde z_1,\tilde z_2)=w_{22}(z_1,z_2)$
with $\tilde z_1=2\int(1-\rho^{-3})\,{\rm d}z_1$, $\tilde z_2=z_2$,
this results in finding many hidden symmetry-like objects for the equation~\eqref{eq:dN}.

The normalizer ${\rm N}_{\mathfrak g}(\mathfrak s_{1.3}^\rho)$ of the subalgebra~$\mathfrak s_{1.3}^\rho$ in~$\mathfrak g$
depends on the value of~$\rho$,
\begin{gather*}
{\rm N}_{\mathfrak g}(\mathfrak s_{1.3}^\rho)=
\big\langle D^{\rm s},\,P^x(1),\,P^y(\rho),\,R^y(\beta)-R^x(\rho\beta),\,Z(\sigma)\big\rangle
\quad\mbox{if}\quad \rho_t\ne0,
\\
{\rm N}_{\mathfrak g}(\mathfrak s_{1.3}^\rho)=
\big\langle D^t(1),\,D^t(t),\,D^{\rm s},\,P^x(1),\,P^y(\rho),\,R^y(\beta)-R^x(\rho\beta),\,Z(\sigma)\big\rangle
\quad\mbox{if}\quad \rho_t=0,
\end{gather*}
where $\beta$ and~$\sigma$ run through the set of smooth functions of~$t$.
The maximal Lie invariance algebra of the modified reduced equation~\eqref{eq:dNs1.3RhoNe1ModifiedRedEq} is
\begin{gather*}
\begin{split}
\mathfrak a_{1.3}=\big\langle&
\p_{\tilde z_1},\,
\tilde z_1\p_{\tilde z_1}-w\p_w,\,
\tilde z_1^2\p_{\tilde z_1}+\tilde z_1\tilde z_2\p_{\tilde z_2}+(\tilde z_1w+\tfrac16\tilde z_2^{\,3})\p_w,\\&
\tilde z_2\p_{\tilde z_2}+3w\p_w,\,
\p_{\tilde z_2},\,
\tilde z_1\p_{\tilde z_2}+\tfrac12\tilde z_2^{\,2}\p_w,\,
\tilde\alpha(\tilde z_1)\tilde z_2\p_w,\,
\tilde\sigma(\tilde z_1)\p_w
\big\rangle,
\end{split}
\end{gather*}
where $\tilde\alpha$ and~$\tilde\sigma$ run through the set of smooth functions of~$\tilde z_1$.
The vector fields~$D^{\rm s}$, $P^x(1)+P^y(\rho)$, $P^y(\rho)$, $R^y(\beta)-R^x(\rho\beta)$, $Z(\sigma)$
and, if $\rho_t=0$, $D^t(1)$ and~$D^t(t)$ from \smash{${\rm N}_{\mathfrak g}(\mathfrak s_{1.3}^\rho)$}
induce the Lie-symmetry vector fields
$\tilde z_2\p_{\tilde z_2}+3w\p_w$, 0, $\p_{\tilde z_2}$,
$\tilde\alpha\tilde z_2\p_w$ with $\tilde\alpha(\tilde z_1)=\rho(t)\beta(t)$,
$\tilde\sigma\p_w$ with $\tilde\sigma(\tilde z_1)=\sigma(t)$ and, if $\rho_t=0$, 
$\p_{\tilde z_1}$ and $\tilde z_1\p_{\tilde z_1}+\frac13\tilde z_2\p_{\tilde z_2}$
of the modified reduced equation~\eqref{eq:dNs1.3RhoNe1ModifiedRedEq}, respectively.
All the elements of~$\mathfrak a_{1.3}$
from the set complement of the linear span of the above vector fields from~$\mathfrak a_{1.3}$
are genuinely hidden symmetries of the equation~\eqref{eq:dN}.
Note that whether the vector fields
$\p_{\tilde z_1}$ and $\tilde z_1\p_{\tilde z_1}+\frac13\tilde z_2\p_{\tilde z_2}$
are genuinely hidden symmetries of~\eqref{eq:dN} depends on the value of the parameter function~$\rho$,
which does not appear in the modified reduced equation~\eqref{eq:dNs1.3RhoNe1ModifiedRedEq}
and in its maximal Lie invariance algebra~$\mathfrak a_{1.3}$.

In view of the representation~\eqref{eq:dNs1.3InvarSolutions}
for all $\mathfrak s_{1.3}^\rho$-invariant solutions of the equation~\eqref{eq:dN},
we do not carry out further Lie reductions of the modified reduced equation~\eqref{eq:dNs1.3RhoNe1ModifiedRedEq}
with respect to subalgebras of~$\mathfrak a_{1.3}$
although most of them are associated with hidden Lie symmetries of~\eqref{eq:dNs1.3RhoNe1ModifiedRedEq}.
At the same time, in Section~\ref{sec:dNLieReductionsOfCodim2Collection1}
we exhaustively study essential direct Lie reductions of~\eqref{eq:dNs1.3RhoNe1ModifiedRedEq}
with respect to two-dimensional subalgebras of~$\mathfrak g$
that can be interpreted as two-step Lie reductions with reduction~1.3 as the first step.

\medskip\par\noindent{\bf 1.4.}\
$\mathfrak s_{1.4}^\beta=\big\langle P^x(1)+R^y(\beta)\big\rangle$ with $\beta=\beta(t)$.
Each reduced equation~1.4$^\beta$ is trivial, see~\cite[Section~4.4]{moro2021a}.
Its general solution is an arbitrary sufficiently smooth function of~$(z_1,z_2)$ if $\beta=0$
and $w=\frac16\beta_1\beta^{-1}z_2^3+\varrho^2(z_1)z_2^2+\varrho^1(z_1)z_2+\varrho^0(z_1)$
with arbitrary sufficiently smooth functions~$\varrho^0$, $\varrho^1$ and~$\varrho^2$ of~$z_1=t$ if $\beta\ne0$.
The case $\beta=0$ leads to the solution family $u=w(t,y)$ of~\eqref{eq:dN},
which is parameterized by an arbitrary sufficiently smooth function $w$ of~$(t,y)$
and is hence a subfamily of the family~\eqref{eq:dNTrivialSolutions}.
In the case $\beta\ne0$, the coefficients~$\varrho^0$, $\varrho^1$ and~$\varrho^2$ can be assumed,
up to the $G$-equivalence, to vanish.
This leads to the following simple solutions of the equation~\eqref{eq:dN}:
\begin{gather}\label{eq:InvSolutions1.4}
\solution u=\frac{\beta_t}{6\beta}y^3+\beta xy,
\end{gather}
where $\beta$ is an arbitrary sufficiently smooth function of~$t$.

We can modify ansatz~1.4$^\beta$ with $\beta\ne0$ to $u=\tilde w(z_1,z_2)+\beta xy+\frac16\beta_t\beta^{-1}y^3$
with the same $z_1=t$ and~$z_2=y$, which simplifies reduced equation~1.4$^\beta$
to $\tilde w_{222}=0$.
Therefore, analogously to the subalgebras~$\mathfrak s_{1.3}^\rho$ with $\rho\ne1$,
the subalgebras from the family $\{\mathfrak s_{1.4}^\beta\mid \beta\ne0\}$,
which is parameterized by an arbitrary nonvanishing function~$\beta$ of~$t$,
also correspond, under this ansatz choice, to the same reduced equation.

Depending on the value of the parameter function~$\beta$,
the normalizer ${\rm N}_{\mathfrak g}(\mathfrak s_{1.4}^\beta)$
of the subalgebra~$\mathfrak s_{1.4}^\beta$ in~$\mathfrak g$ is
\begin{gather*}
\big\langle D^t(1),\,D^t(t),\,D^{\rm s},\,P^x(1),\,P^y(\rho),\,R^y(\breve\beta),\,Z(\sigma)\big\rangle
\quad\mbox{if}\quad \beta=0,
\\
\big\langle D^t(1),\,D^t(t)+\tfrac23D^{\rm s},\,P^x(1),\,P^y(\rho)+R^x(\rho\beta),\,R^y(\breve\beta),\,Z(\sigma)\big\rangle
\quad\mbox{if}\quad \beta\ne0,\ \beta_t=0,
\\
\big\langle D^t(1)+\kappa D^{\rm s},\,P^x(1),\,P^y(\rho)+R^x(\rho\beta),\,R^y(\breve\beta),\,Z(\sigma)\big\rangle
\quad\mbox{if}\quad \beta_t\ne0,\ \beta_t=\kappa\beta,
\\
\big\langle D^t(t\!+\!\mu)+(\kappa\!+\!\tfrac23)D^{\rm s},\,P^x(1),\,P^y(\rho)+R^x(\rho\beta),\,R^y(\breve\beta),\,Z(\sigma)\big\rangle
\!\quad\mbox{if}\quad \beta_t\ne0,\ (t\!+\!\mu)\beta_t=\kappa\beta,
\\
\big\langle P^x(1),\,P^y(\rho)+R^x(\rho\beta),\,R^y(\breve\beta),\,Z(\sigma)\big\rangle
\quad\mbox{otherwise},
\end{gather*}
where $\rho$, $\breve\beta$ and~$\sigma$ run through the set of smooth functions of~$t$.
In the second, the third and the fourth cases,
$\beta=1$, $(\beta,\kappa)=({\rm e}^t,1)$ and $(\beta,\mu)=(|t|^\kappa,0)$ modulo the $G$-equivalence, respectively.

For any value of the parameter function~$\beta$,
the vector fields~$P^x(1)+R^y(\beta)$, $P^y(\rho)+R^x(\rho\beta)$, $R^y(\breve\beta)$ and $Z(\sigma)$
from \smash{${\rm N}_{\mathfrak g}(\mathfrak s_{1.4}^\beta)$} induce the Lie-symmetry vector fields
$0$, $\rho\p_{z_2}-\tfrac12\rho_tz_2^{\,2}\p_w$, $\breve\beta z_2\p_w$ and $\sigma\p_w$
of reduced equations~1.4$^\beta$,
where \smash{$\breve\beta$}, $\rho$ and $\sigma$ are arbitrary smooth functions of $z_1=t$.
For particular values of~$\beta$ with extension of~\smash{${\rm N}_{\mathfrak g}(\mathfrak s_{1.4}^\beta)$},
there are the following additional independent inductions:
\begin{gather*}
\p_{z_1},\ z_1\p_{z_1}+\tfrac13z_2\p_{z_2},\ z_2\p_{z_2}+3w\p_w\quad\mbox{by}\quad D^t(1),\ D^t(t),\ D^{\rm s}
\quad\mbox{if}\quad\beta=0,
\\[.5ex]
\p_{z_1},\ z_1\p_{z_1}+z_2\p_{z_2}+2w\p_w\quad\mbox{by}\quad D^t(1),\ D^t(t)+\tfrac23D^{\rm s}
\quad\mbox{if}\quad\beta\ne0,\ \beta_t=0,
\\[.5ex]
\p_{z_1}+\kappa z_2\p_{z_2}+3\kappa w\p_w\quad\mbox{by}\quad D^t(1)+\kappa D^{\rm s}
\quad\mbox{if}\quad\beta_t\ne0,\ \beta_t=\kappa\beta,
\\[.5ex]
(z_1+\mu)\p_{z_1}+(\kappa+1)z_2\p_{z_2}+(3\kappa+2)w\p_w\quad\mbox{by}\quad D^t(t+\mu)+(\kappa+\tfrac23)D^{\rm s}
\\ \qquad\mbox{if}\quad\beta_t\ne0,\ (t+\mu)\beta_t=\kappa\beta,
\end{gather*}
where $\kappa$ and~$\mu$ are arbitrary constants with $\kappa\ne0$.
If we use the modified ansatz in the case $\beta\ne0$,
the expressions for the analogous induced Lie-symmetry vector fields are formally the same
up to replacing~$w$ by~$\tilde w$, except the vector field $\rho\p_{z_2}-\tfrac12\rho_tz_2^{\,2}\p_w$
that should be replaced by $\rho\p_{z_2}-\tfrac12(\rho_t+\beta_t\beta^{-1}\rho)z_2^{\,2}\p_{\tilde w}$.

Since reduced equation~1.4$^0$ is in fact the identity,
it admits any point transformation in the space with coordinates $(z_1,z_2,w)$ as its point symmetry,
and any vector field in this space is its Lie-symmetry vector field.
The maximal Lie invariance of the modified reduced equation $\tilde w_{222}=0$ for the case $\beta\ne0$
coincides with the span of the vector fields
$\p_{z_1}$,
$\p_{z_2}$, $z_2\p_{z_2}$, $z_2^{\,2}\p_{z_2}+2z_2\tilde w\p_{\tilde w}$,
$\tilde w\p_{\tilde w}$, $\p_{\tilde w}$, $z_2\p_{\tilde w}$ and $z_2^{\,2}\p_{\tilde w}$
over the (pseudo)ring smooth functions of~$z_1$.
Moreover, it is obvious that these reduced equations also possess very wide sets of other symmetry-like objects.
Hence for each~$\beta$,
the equation~\eqref{eq:dN} admits many hidden symmetry-like objects
associated with reduction~1.4$^\beta$
but they are not of interest in view of the triviality of reduced equations~1.4$^\beta$.

\section{Lie reductions of codimension two}\label{sec:dNLieReductionsOfCodim2}

The subalgebras
$\mathfrak s_{2.1}^\lambda$ with $\lambda=-1/3$,
$\mathfrak s_{2.7}^{\lambda\nu}$, $\mathfrak s_{2.8}^\lambda$,
$\mathfrak s_{2.10}^\beta$, $\mathfrak s_{2.11}$, $\mathfrak s_{2.12}$,
$\mathfrak s_{2.15}^{\delta\nu}$, $\mathfrak s_{2.16}^\delta$,
$\mathfrak s_{2.18}^{\rho\beta\sigma}$,
$\mathfrak s_{2.19}^{\beta^1\beta^2}$,
$\mathfrak s_{2.20}^{\beta\sigma}$,
$\mathfrak s_{2.21}^{\alpha\beta^1\beta^2}$,
$\mathfrak s_{2.22}^{\alpha\beta\sigma}$,
$\mathfrak s_{2.23}^{\beta\sigma}$,
$\mathfrak s_{2.24}^\sigma$ and
$\mathfrak s_{2.25}^\sigma$
cannot be used for codimension-two Lie reductions of the equation~\eqref{eq:dN}
since the rank of these subalgebras is less than two.

The subalgebras~$\mathfrak s_{2.9}^0$ and~\smash{$\mathfrak s_{2.17}^{\rho\alpha\beta}$}
contain the subalgebras~$\mathfrak s_{1.4}^0$ and~\smash{$\mathfrak s_{1.4}^\beta$}, respectively,
and the one-dimensional subalgebras
of~\smash{$\mathfrak s_{2.14}^{\delta0\delta'}$} and~\smash{$\mathfrak s_{2.6}^{\lambda\delta}$}
that are spanned by the respective second basis elements
are $G$-equivalent to~\smash{$\mathfrak s_{1.4}^\beta$} for some~$\beta$.
All $\mathfrak s_{1.4}^\beta$-invariant solutions
were constructed in Section~\ref{sec:dNLieReductionsOfCodim1}.
This is why we can neglect
the subalgebras~\smash{$\mathfrak s_{2.6}^{\lambda\delta}$}, \smash{$\mathfrak s_{2.9}^0$},
\smash{$\mathfrak s_{2.14}^{\delta0\delta'}$} and~\smash{$\mathfrak s_{2.17}^{\rho\alpha\beta}$}
in the course of carrying out codimension-two Lie reductions of the equation~\eqref{eq:dN}.
The same claim is true for the subalgebras~$\mathfrak s_{2.5}^{\lambda1}$, $\mathfrak s_{2.9}^1$
and $\mathfrak s_{2.14}^{\delta1\delta'}$
due to their relation to the subalgebra~$\mathfrak s_{1.3}^1$.

For the subalgebras~$\mathfrak s_{2.5}^{\lambda\mu}$, $\mathfrak s_{2.9}^\rho$
and $\mathfrak s_{2.14}^{\delta\nu\delta'}$ with $\mu,\nu\ne0,1$ and $\rho\ne0,1$,
the similar claim is relevant only partially since
the representation~\eqref{eq:dNs1.3InvarSolutions} for \smash{$\mathfrak s_{1.3}^\rho$}-invariant solutions
is not explicit and involves quadratures of the general solution of the inviscid Burgers equation.
At the same time, the Lie reductions associated with these subalgebras are simpler
and essentially differ from the remaining $G$-inequivalent Lie reductions.
Moreover, we were able to construct general solutions of all the corresponding reduced equations
either in an explicit form in terms of elementary functions and the Lambert $W$ function
or in a parametric form.
This is why we consider the above Lie reductions first.

The remaining subalgebras from the list in Lemma~\ref{lem:dN2DInequivSubalgs} are
$\mathfrak s_{2.1}^\lambda$, $\mathfrak s_{2.2}^\nu$, $\mathfrak s_{2.3}^\nu$,
$\mathfrak s_{2.4}$ and $\mathfrak s_{2.13}$.
They constitute the second collection of subalgebras to be considered
within the framework of Lie reductions.
All the corresponding invariant solutions are stationary.
The integration of the involved reduced equations is more complicated,
and the construction of general or even particular solutions in certain closed form
is possible only for some of them.
For each of the subalgebras in the second collection,
we consider its counterparts among inequivalent two-dimensional subalgebras
of the algebra~$\mathfrak g_{\rm L}$ and
associated $G_{\rm L}$-inequivalent Lie reductions
of the nonlinear Lax representation~\eqref{eq:dNLaxPair}.
Some of these reductions help us in finding exact solutions to the associated reduced equations.

Below, for each of the subalgebras
$\mathfrak s_{2.5}^{\lambda\mu}$ with $\mu\ne0,1$,
$\mathfrak s_{2.9}^\rho$ with $\rho\ne0,1$ and
$\mathfrak s_{2.14}^{\delta\nu\delta'}$ with $\nu\ne0,1$ (the first collection)
and
$\mathfrak s_{2.1}^\lambda$ with $\lambda\ne-1/3$, $\mathfrak s_{2.2}^\nu$, $\mathfrak s_{2.3}^\nu$,
$\mathfrak s_{2.4}$ and $\mathfrak s_{2.13}$ (the second collection),
we present an ansatz for the related invariant solutions
and the corresponding reduced equation,
where $\varphi=\varphi(\omega)$ is the new unknown functions of the single invariant variable~$\omega$.
We also compute the subalgebra normalizers and induced symmetries (both infinitesimal and discrete)
of reduced equations.
In reduced nonlinear Lax representations,
$\psi=\psi(\omega)$ is one more unknown functions of the single invariant variable~$\omega$,
which is associated with the unknown function~$q$ in the original nonlinear Lax representation~\eqref{eq:dNLaxPair}.

\subsection{The first collection of reductions}\label{sec:dNLieReductionsOfCodim2Collection1}

\noindent{\bf 2.5.}\
$\mathfrak s_{2.5}^{\lambda\mu}=\big\langle D^t(1)+\lambda D^{\rm s},\,
P^x({\rm e}^{(\lambda-1)t})+\mu P^y({\rm e}^{(\lambda-1)t})\big\rangle$, \ $\mu\ne0,1$, \ $|\mu|\leqslant1$ $(\!{}\bmod G)$:
\begin{gather*}
u={\rm e}^{3\lambda t}\varphi-\frac{\lambda-1}6(x^3+y^3),\quad
\omega={\rm e}^{-\lambda t}(y-\mu x);\\
2(\mu^3-1)\varphi_{\omega\omega}\varphi_{\omega\omega\omega}-\omega\varphi_{\omega\omega\omega}
+(3\lambda-2)\varphi_{\omega\omega}=0.
\end{gather*}

For any value of the parameter tuple $(\lambda,\mu)$,
the subalgebra~$\mathfrak s_{2.5}^{\lambda\mu}$ has the same normalizer in~$\mathfrak g$,
\[
{\rm N}_{\mathfrak g}(\mathfrak s_{2.5}^{\lambda\mu})=\big\langle
D^t(1),\,D^{\rm s},\,P^x({\rm e}^{(\lambda-1)t})+\mu P^y({\rm e}^{(\lambda-1)t}),\,
-\mu R^x({\rm e}^{2\lambda t})+R^y({\rm e}^{2\lambda t}),\,Z({\rm e}^{3\lambda t})\big\rangle.
\]
Reduced equation~2.5$^{\lambda\mu}$ is invariant with respect to the algebra
$\mathfrak a_{2.5}=\langle\omega\p_\omega+3\varphi\p_\varphi,\p_\varphi,\omega\p_\varphi\rangle$
and the point transformation $(\tilde\omega,\tilde\varphi)=(-\omega,-\varphi)$
and, therefore, with respect to the Lie group~$G_{2.5}$
that consists of the point transformations $\tilde\omega=a_1\omega$, $\tilde\varphi=a_1^{\,3}\varphi+a_2\omega+a_3$,
where~$a_1$, $a_2$ and~$a_3$ are arbitrary constants with $a_1\ne0$.
The vector fields
$D^t(1)+\lambda D^{\rm s}$, $D^{\rm s}$,
$P^x({\rm e}^{(\lambda-1)t})+\mu P^y({\rm e}^{(\lambda-1)t})$,
$-\mu R^x({\rm e}^{2\lambda t})+R^y({\rm e}^{2\lambda t})$ and $Z({\rm e}^{3\lambda t})$
from \smash{${\rm N}_{\mathfrak g}(\mathfrak s_{2.5}^{\lambda\mu})$}
induce the Lie-symmetry vector fields
$0$, $\omega\p_\omega+3\varphi\p_\varphi$, $0$, $\omega\p_\varphi$ and~$\p_\varphi$
of reduced equation~2.5$^{\lambda\mu}$, respectively.
This means that the entire algebra~$\mathfrak a_{2.5}$
is induced by elements of \smash{${\rm N}_{\mathfrak g}(\mathfrak s_{2.5}^{\lambda\mu})$}.
Alternating the signs of $(\omega,\varphi)$ is a discrete point symmetry of reduced equation~2.5$^{\lambda\mu}$
and is induced by the discrete point symmetry transformation $\mathscr I^{\rm s}:=\mathscr D^{\rm s}(-1)$
of the original equation~\eqref{eq:dN}.
Hence the group~$G_{2.5}$ is entirely induced by the point symmetry group~$G$
of the original equation~\eqref{eq:dN}.

For any values of $(\lambda,\mu)$, reduced equation~2.5$^{\lambda\mu}$ is satisfied
by all $\varphi$ with $\varphi_{\omega\omega}=0$
but such values of~$\varphi$ are $G_{2.5}$-equivalent to 0
and, moreover, corresponds to solutions of the equation~\eqref{eq:dN}
that are $G$-equivalent to the zero solution $u=0$.
Further we assume that $\varphi_{\omega\omega}\ne0$.

There are several values of~$\lambda$,
for which the general solutions of the corresponding reduced equation~2.5$^{\lambda\mu}$
can be represented in the closed form.
These are $\lambda=2/3$, $\lambda=1/3$, $\lambda=5/6$ and $\lambda=1$.

Solving of reduced equation~2.5$^{\lambda\mu}$ with $\lambda=2/3$ degenerates to
the independent consideration of two equations, the elementary equation~$\varphi_{\omega\omega\omega}=0$
and the equation $2(\mu^3-1)\varphi_{\omega\omega}=\omega$.
Solutions of the first equation correspond to solutions of the equation~\eqref{eq:dN}
that are $G$-equivalent to the zero solution $u=0$,
whereas the second equation is a particular case of the constraint $2(\mu^3-1)\varphi_{\omega\omega}=-3(\lambda-1)\omega$,
whose solution set is contained in that of reduced equation~2.5$^{\lambda\mu}$ for any~$\lambda$,
and the associated solutions of the equation~\eqref{eq:dN} take, modulo the $G$-equivalence, the~form
\begin{gather}\label{eq:ReducedEq2.5QubicSolution}
\solution
u=\frac\kappa{4(\mu^3-1)}(y-\mu x)^3+\frac\kappa6(x^3+y^3)
\end{gather}
with $\kappa:=-(\lambda-1)$.
Note that the solution~\eqref{eq:ReducedEq2.5QubicSolution} with fixed values of~$\kappa$ and~$\mu$
is in fact invariant with respect to the three-dimensional subalgebra
$\big\langle D^t(1),\, D^{\rm s},\,P^x({\rm e}^{-\kappa t})+\mu P^y(\nu{\rm e}^{-\kappa t})\big\rangle$ of~$\mathfrak g$.

Below $\lambda\ne2/3$.
We can integrate reduced equation~2.5$^{\lambda\mu}$ once in two different ways.
The first way uses the fact that the left-hand side of this equation is a total derivative with respect to~$\omega$,
thus leading to the equation
\begin{gather}\label{eq:ReducedEq2.5IntA}
(\mu^3-1)(\varphi_{\omega\omega})^2-\omega\varphi_{\omega\omega}+(3\lambda-1)\varphi_{\omega}+c_1=0,
\end{gather}
where $c_1$ is the integration constant.
The second way is to consider reduced equation~2.5$^{\lambda\mu}$ as a first-order ordinary differential equation
with respect to~$\varphi_{\omega\omega}$,
which integrates to
\begin{gather}\label{eq:ReducedEq2.5IntB}
\omega=-\frac23\frac{\mu^3-1}{\lambda-1}\varphi_{\omega\omega}+c_2|\varphi_{\omega\omega}|^{\frac1{3\lambda-2}}
\quad\mbox{if}\quad \lambda\ne1,
\\\nonumber
\omega=-2(\mu^3-1)\varphi_{\omega\omega}\ln|\varphi_{\omega\omega}|+c_2\varphi_{\omega\omega}\quad\mbox{if}\quad \lambda=1,
\end{gather}
where $c_2$ is another arbitrary constant,
and the integrated equations can be easily solved as algebraic equations with respect to~$\varphi_{\omega\omega}$
for four values of~$\lambda$, $\lambda=1/3$, $\lambda=5/6$, $\lambda=4/3$ and $\lambda=1$.
Some of these values are also singled out in the course of the further integration.

Moreover, the case $\lambda=1/3$ is singular from the point of view of gauging the constant~$c_1$
by point symmetries of the corresponding reduced equation.
To be specific, in contrast to the other values of~$\lambda$, we cannot set~$c_1$ to be equal to zero
and can only assume that $c_1\in\{-1,0,1\}$.
The equation~\eqref{eq:ReducedEq2.5IntA} with $\lambda=1/3$ is easily solved.
Its general solution is
\[
\varphi=\frac{\omega^3+\varepsilon\sqrt{(\omega^2+c_1)^3}}{12(\mu^3-1)}
+\frac{\varepsilon c_1}{4(\mu^3-1)}
\Big(\omega\ln\big|\omega+\sqrt{\omega^2+c_1}\,\big|-\sqrt{\omega^2+c_1}\,\Big)+c_2\omega+c_3,
\]
where $\varepsilon=\pm1$, and $c_2$ and~$c_3$ are integration constants,
which can be set to be equal zero modulo the $G_{2.5}$-equivalence.
The corresponding family of solutions of the equation~\eqref{eq:dN} is
\[
\solution
u={\rm e}^t\left(
\frac{\omega^3+\varepsilon\sqrt{(\omega^2+c_1)^3}}{12(\mu^3-1)}
+\frac{\varepsilon c_1}{4(\mu^3-1)}
\Big(\omega\ln\big|\omega+\sqrt{\omega^2+c_1}\,\big|-\sqrt{\omega^2+c_1}\,\Big)
\right)+\frac{x^3+y^3}9,
\]
where $\omega={\rm e}^{-t/3}(y-\mu x)$, $\varepsilon=\pm1$,
$\mu$ is an arbitrary constant with $\mu\ne0,1$,
and $c_1\in\{-1,0,1\}\pmod{G_{2.5}}$.

For $\lambda=5/6$, $\lambda=4/3$ and $\lambda=1$,
the general solutions of the corresponding reduced equations~2.5$^{\lambda\mu}$
can also be represented in closed form,
where for convenience we use other integration constants $b_1$, $b_2$ and~$b_3$,
and $\varepsilon=\pm1$:
\begin{gather*}
\lambda=\frac56\colon\quad
\varphi=\frac{\mu^3-1}{b_1}\omega^2
+\frac{4\varepsilon}{15b_1^{\,3}}
\big(4(\mu^3-1)^2-b_1\omega\big)^{5/2}+b_2\omega+b_3,
\\
\lambda=\frac43\colon\quad
\varphi=
-\frac{\omega^3}{12(\mu^3-1)}
+\frac{b_1^{\,2}\omega^2}{16(\mu^3-1)^2}
+\varepsilon b_1\frac{\big(b_1^{\,2}-8(\mu^3-1)\omega\big)^{5/2}}{1920(\mu^3-1)^4}
+b_2\omega+b_3,
\\[1ex]
\lambda=1\colon\quad
\varphi=-\omega^3\frac{18z^2+15z+4}{216(\mu^3-1)z^3}+b_2\omega+b_3,
\ \
z\in\big\{W_0(\tilde\omega),W_{-1}(\tilde\omega)\big\},\ \
\tilde\omega:=-\frac{b_1\omega}{2(\mu^3-1)},
\end{gather*}
where~$W_0$ and~$W_{-1}$ are the principal real and the other real branches of the Lambert $W$ function, respectively.
The solutions with $b_1=0$ correspond to solutions of the equation~\eqref{eq:dN}
that belong, up to the $G$-equivalence, to the family~\eqref{eq:ReducedEq2.5QubicSolution}.
Hence we can assume that $b_1\ne0$ and thus set the gauges $b_1=1,\ b_2=b_3=0 \pmod{G_{2.5}}$.
This leads to the following $G$-inequivalent solutions of the equation~\eqref{eq:dN}
with $\varepsilon=\pm1$ and an arbitrary constant $\mu\ne0,1$:
\begin{gather*}
\solution
u=(\mu^3-1){\rm e}^{\frac56t}(y-\mu x)^2
+\frac{4\varepsilon}{15}{\rm e}^{\frac52t}
\big(4(\mu^3-1)^2-{\rm e}^{-\frac56t}(y-\mu x)\big)^{5/2}
+\frac{x^3+y^3}{36},
\\[1ex]
\solution
u=-\frac{(y-\mu x)^3}{12(\mu^3-1)}
+{\rm e}^{\frac43t}\frac{(y-\mu x)^2}{16(\mu^3-1)^2}
+\varepsilon{\rm e}^{4t}\frac{\big(1-8(\mu^3-1){\rm e}^{-\frac43t}(y-\mu x)\big)^{5/2}}{1920(\mu^3-1)^4}
-\frac{x^3+y^3}{18},
\\[1ex]
\solution
u=-(y-\mu x)^3\frac{18z^2+15z+4}{216(\mu^3-1)z^3},\quad
z\in\big\{W_0(\tilde\omega),W_{-1}(\tilde\omega)\big\},\quad
\tilde\omega:=-\frac{{\rm e}^{-t}(y-\mu x)}{2(\mu^3-1)}.
\end{gather*}

For any $\lambda\ne2/3,1/3,1$, the general solution of reduced equation~2.5$^{\lambda\mu}$
can be represented in a parametric form in a uniform way.
Considering the derivative~$\varphi_{\omega\omega}$
in the equations~\eqref{eq:ReducedEq2.5IntA} and~\eqref{eq:ReducedEq2.5IntB} as a parameter
and denoting it by~$s$, we rewrite these equations as
\[
\omega=-\frac23\frac{\mu^3-1}{\lambda-1}s+c_2|s|^{\frac1{3\lambda-2}},\quad
\varphi_{\omega}=-\frac{(\mu^3-1)s^2-\omega s+c_1}{3\lambda-1}.
\]
The associated parametric expression for~$\varphi=\int\varphi_{\omega}\,{\rm d}\omega$ is given by
\begin{gather}\label{eq:ReducedEq2.5ParamSol}
\begin{split}&
\varphi=
(\mu^3-1)\frac{4(\mu^3-1)s-3\omega}{27\lambda(2\lambda-1)}s^2
-\frac{(\mu^3-1)s-\omega}{3\lambda(3\lambda-1)}\omega s
-\frac{c_1\omega}{3\lambda-1}+c_3\quad\mbox{if}\quad \lambda\ne0,\frac12,\frac13,
\\&
\varphi=
\frac2{27}(\mu^3-1)^2s^3-\frac59c_2(\mu^3-1)|s|^{3/2}+\frac{c_2^{\,2}}2\sgn(s)\ln|s|+c_1\omega+c_3
\quad\mbox{if}\quad \lambda=0,
\\&
\varphi=
\frac8{27}(\mu^3-1)^2s^3+\frac43c_2(\mu^3-1)\ln|s|+\frac{4c_2^{\,2}}{3s^3}+2c_1\omega+c_3
\quad\mbox{if}\quad \lambda=\frac12.
\end{split}
\end{gather}
Note that the value $c_2=0$ corresponds to solutions of the form~\eqref{eq:ReducedEq2.5QubicSolution}
and can be excluded from the consideration.
Thus, we can assume $c_2\ne0$ and thus set $c_2=1,\ c_1=c_3=0 \pmod{G_{2.5}}$.
This leads to the following $G$-inequivalent solutions of the equation~\eqref{eq:dN}:
\[
\solution
u={\rm e}^{3\lambda t}\varphi-\frac{\lambda-1}6(x^3+y^3),
\]
where $\lambda\ne2/3,1/3,1$, $\mu\ne0,1$,
$\varphi$ is defined by the appropriate equation from~\eqref{eq:ReducedEq2.5ParamSol}
with $c_2=1$, $c_1=c_3=0$, $\omega:={\rm e}^{-\lambda t}(y-\mu x)$
and the function $s=s(\omega)$ is implicitly defined as a solution of the Lambert's transcendental equation
\begin{gather}\label{eq:ReducedEq2.5Lambert'sTranscendentalEq}
|s|^{\frac1{3\lambda-2}}-\frac23\frac{\mu^3-1}{\lambda-1}s=\omega.
\end{gather}
In fact, the elementary solvability of this equation for $\lambda\in\{5/6,4/3\}$
as a quadratic equation with respect to a degree of~$s$ has been used above
for deriving explicit solutions of the equation~\eqref{eq:dN}.
The equation~\eqref{eq:ReducedEq2.5Lambert'sTranscendentalEq}
can also be solved for some other values of~$\lambda$
as algebraic equations with respect to certain degrees of~$s$,
which results in explicit expressions for the general solutions
of the corresponding reduced equation~2.5$^{\lambda\mu}$.
Thus, $(3\lambda-2)^{-1}=-1/2,-2,3,1/3$ for ${\lambda=0,1/2,7/9,5/3}$, respectively,
and the corresponding equations~\eqref{eq:ReducedEq2.5Lambert'sTranscendentalEq} are cubic equations
with respect to  certain degrees of~$s$ and, therefore, can be solved, e.g., using the Cardano formula.

\medskip\par\noindent{\bf 2.9.}\
$\mathfrak s_{2.9}^\rho=\big\langle D^{\rm s},\,P^x(1)+P^y(\rho)\big\rangle$, \
$\rho\not\equiv1$ for any open interval in the domain of~$\rho$ and
$\rho(t)\ne0$ for any~$t$  in this domain:
\[
u=\frac{(y-\rho x)^3}{\rho^3}\varphi-\frac{\rho_t}{6\rho}y^3,\quad \omega=t;\qquad
\varphi_\omega=-12\frac{\rho^3-1}{\rho^3}\varphi^2.
\]
The normalizer of the subalgebra~$\mathfrak s_{2.9}^\rho$ in~$\mathfrak g$ is
\begin{gather*}
{\rm N}_{\mathfrak g}(\mathfrak s_{2.9}^\rho)=\big\langle D^{\rm s},\,P^x(1)+P^y(\rho)\big\rangle
\quad\mbox{if}\quad\rho_t\ne0,
\\
{\rm N}_{\mathfrak g}(\mathfrak s_{2.9}^\rho)=\big\langle D^t(1),\,D^t(t),\,D^{\rm s},\,P^x(1)+P^y(\rho)\big\rangle
\quad\mbox{if}\quad\rho_t=0.
\end{gather*}
Since reduced equation~2.9$^\rho$ is a first-order ordinary differential equation,
its maximal Lie invariance algebra~$\mathfrak a_{2.9}^\rho$ is infinite-dimensional.
The normalizer~${\rm N}_{\mathfrak g}(\mathfrak s_{2.9}^\rho)$ induces merely
the zero subalgebra of the algebra~$\mathfrak a_{2.9}^\rho$
and its subalgebra~$\langle\p_\omega,\omega\p_\omega-\varphi\p_\varphi\rangle$
if $\rho_t\ne0$ and $\rho_t=0$, respectively.
Therefore, the original equation~\eqref{eq:dN} admits an infinite number
of linearly independent hidden symmetries that are associated with reduced equation~2.9$^\rho$.
Nevertheless, these hidden symmetries are not of great interest
in view of the trivial integrability of reduced equation~2.9$^\rho$.
Separating the variables in the reduced equation, we integrate it and
substitute the obtained expression for~$\varphi$ into the ansatz,
which gives the following a family of solutions of~\eqref{eq:dN} (cf.\ reduction~1.3): 
\[
\solution
u=\left(\int\frac{\rho^3-1}{\rho^3}{\rm d}t\right)^{-1}\frac{(y-\rho x)^3}{12\rho^3}-\frac{\rho_t}{6\rho}y^3.
\]

\noindent{\bf 2.14.}\
$\mathfrak s_{2.14}^{\delta\nu\delta'}=\big\langle D^t(1)\!+\!\delta D^{\rm s},\,
P^x({\rm e}^{\delta t})\!+\!\nu P^y({\rm e}^{\delta t})\!+\!\delta'R^y({\rm e}^{2\delta t})\big\rangle$,
$\delta,\delta'\!\in\!\{0,1\}$, $\nu\!\ne\!0,1$, $|\nu|\!\leqslant\!1$ $(\!{}\bmod G)$:
\begin{gather*}
u={\rm e}^{3\delta t}\varphi-\frac\delta6(x^3+y^3)+\frac{\delta'}{2\nu}{\rm e}^{\delta t}y^2,\quad
\omega={\rm e}^{-\delta t}(y-\nu x);\\
2\nu(\nu^3-1)\varphi_{\omega\omega}\varphi_{\omega\omega\omega}=\delta'\varphi_{\omega\omega\omega}-3\nu\delta\varphi_{\omega\omega}.
\end{gather*}

Depending on values of $(\delta,\delta')$,
the normalizer of the subalgebra~$\mathfrak s_{2.14}^{\delta\nu\delta'}$ in~$\mathfrak g$ is one of the following:
\begin{gather*}
{\rm N}_{\mathfrak g}(\mathfrak s_{2.14}^{0\nu0})=\langle D^t(1),\,D^t(t),\,D^{\rm s},\,
P^x(1),\,P^y(1),\,-\nu R^x(1)+R^y(1),\,Z(1)\rangle,
\\
{\rm N}_{\mathfrak g}(\mathfrak s_{2.14}^{0\nu1})=\langle D^t(1),\,D^t(t)+\tfrac23D^{\rm s},\,
P^x(1),\,P^y(1)+R^x(1),\,-\nu R^x(1)+R^y(1),\,Z(1)\rangle,
\\
{\rm N}_{\mathfrak g}(\mathfrak s_{2.14}^{1\nu0})=\langle D^t(1),\,D^{\rm s},\,
P^x({\rm e}^t),\,P^y({\rm e}^t),\,-\nu R^x({\rm e}^{2t})+R^y({\rm e}^{2t}),\,Z({\rm e}^{3t})\rangle,
\\
{\rm N}_{\mathfrak g}(\mathfrak s_{2.14}^{1\nu1})=\langle D^t(1)+D^{\rm s},\,
P^x({\rm e}^t),\,P^y({\rm e}^t)+R^x({\rm e}^{2t}),\,-\nu R^x({\rm e}^{2t})+R^y({\rm e}^{2t}),\,Z({\rm e}^{3t})\rangle.
\end{gather*}
Therefore, the vector fields
$D^t(1)+\delta D^{\rm s}$,
$P^x({\rm e}^{\delta t})$,
$ P^y({\rm e}^{\delta t})+\delta'R^x({\rm e}^{2\delta t})$,
$-\nu R^x({\rm e}^{2\delta t})+R^y({\rm e}^{2\delta t})$ and
$Z({\rm e}^{3\delta t})$
belong to ${\rm N}_{\mathfrak g}(\mathfrak s_{2.14}^{\delta\nu\delta'})$
with the corresponding value of $(\delta,\delta')$
and induce the Lie-symmetry vector fields
$0$, $-\nu\p_\omega$, $\p_\omega-\delta'\nu\omega\p_\varphi$, $\omega\p_\varphi$ and~$\p_\varphi$
of reduced equation~2.14$^{\delta\nu\delta'}$, respectively.
For any values of $(\delta,\nu,\delta')$,
reduced equation~2.14$^{\delta\nu\delta'}$ is invariant
with respect to the algebra $\mathfrak a_{2.14}=\langle\p_\omega,\p_\varphi,\omega\p_\varphi\rangle$
and, therefore, with respect to the corresponding Lie group~$G_{2.14}$,
which consists of the point transformations $\tilde\omega=\omega+a_1$, $\tilde\varphi=\varphi+a_2\omega+a_3$,
where $a_1$, $a_2$ and $a_3$ are arbitrary constants.
The group~$G_{2.14}$ is entirely induced by the point symmetry group~$G$
of the original equation~\eqref{eq:dN}.
For any values of $(\delta,\nu,\delta')$, reduced equation~2.14$^{\delta\nu\delta'}$ is satisfied
by all~$\varphi$ with $\varphi_{\omega\omega}=0$
but such values of~$\varphi$ are $G_{2.14}$-equivalent to 0
and, moreover, correspond to solutions of the equation~\eqref{eq:dN}
that are $G$-equivalent to the zero solution $u=0$.
Further we can consider only solutions with $\varphi_{\omega\omega}\ne0$.

Consider the case $\delta=0$.
Reduced equation~2.14$^{0\nu\delta'}$ degenerates to
the elementary equation~$\varphi_{\omega\omega\omega}=0$
whose maximal Lie invariance algebra is well known,
$\mathfrak a_{2.14}^{0\nu\delta'}=\langle\p_\omega,\,\omega\p_\omega,\,\omega^2\p_\omega+2\omega\varphi\p_\varphi,\,
\p_\varphi,\,\omega\p_\varphi,\,\omega^2\p_\varphi,\,\varphi\p_\varphi\rangle$.
In addition to the Lie-symmetry vector fields~$\p_\omega$, $\p_\varphi$ and $\omega\p_\varphi$,
which are induced as in the general case,
elements of the normalizer ${\rm N}_{\mathfrak g}(\mathfrak s_{2.14}^{0\nu\delta'})$ induce
$\omega\p_\omega+2\varphi\p_\varphi$ if $\delta'=1$ and
$\omega\p_\omega$ and $\varphi\p_\varphi$ if $\delta'=0$.
Any element of~$\mathfrak a_{2.14}^{0\nu\delta'}$
involving at least one of the basis vector fields $\omega^2\p_\omega+2\omega\varphi\p_\varphi$, $\omega^2\p_\varphi$
and, if $\delta'=1$, $\omega\p_\omega+a\varphi\p_\varphi$ with $a\ne2$
is a hidden symmetry of the equation~\eqref{eq:dN}.
All the corresponding solutions of the equation~\eqref{eq:dN}
are $G$-equivalent to either the zero solution $u=0$
or solutions of the form~\eqref{eq:InvSolutions1.4} with $\beta=\const$.

Reduced equation~2.14$^{1\nu0}$ is factored out to
$\big(2(\nu^3-1)\varphi_{\omega\omega\omega}+3\big)\varphi_{\omega\omega}=0$.
Therefore, its solution set is the disjoint union of the solution sets of the equations
$2(\nu^3-1)\varphi_{\omega\omega\omega}+3=0$ and $\varphi_{\omega\omega}=0$.
This implies that the maximal Lie invariance algebra~$\mathfrak a_{2.14}^{1\nu0}$
of reduced equation~2.14$^{1\nu0}$ is the intersection of the maximal Lie invariance algebras
of the above equations,
$\mathfrak a_{2.14}^{1\nu0}=\langle\p_\omega,\,\p_\varphi,\,\omega\p_\varphi,\,\omega\p_\omega+3\varphi\p_\varphi\rangle$.
The entire algebra~$\mathfrak a_{2.14}^{1\nu0}$ is induced by ${\rm N}_{\mathfrak g}(\mathfrak s_{2.14}^{1\nu0})$.
Thus, the case $(\delta,\delta')=(1,0)$ leads, modulo the $G$-equivalence,
to the solutions of the equation~\eqref{eq:dN}
that are of the form~\eqref{eq:ReducedEq2.5QubicSolution} with $\kappa=-1$ and $\mu=\nu$.

Consider the case $\delta\delta'\ne0$.
Thus, $\delta=1$.
We neglect the gauge $\delta'=1$, set $\delta'$ to another value, $\delta'=-\nu(\nu^3-1)$,
and denote $\kappa:=-3(\nu^3-1)^{-1}$.
As a result, we need to solve the equation
$\big((\varphi_{\omega\omega})^2+\varphi_{\omega\omega}-\kappa\varphi_\omega\big)_\omega=0$.
We integrate it once, deriving $(\varphi_{\omega\omega})^2+\varphi_{\omega\omega}-\kappa\varphi_\omega-c_1=0$,
$c_1$ is an integration constant, and solve the integrated equation with respect to~$\varphi_{\omega\omega}$,
\[
\varphi_{\omega\omega}=-\frac12\pm\frac12\sqrt{4\kappa\varphi_\omega+1+4c_1}.
\]
The maximal Lie invariance algebra of reduced equation~2.14$^{1\nu\delta'}$
coincides with the common invariance algebra~$\mathfrak a_{2.14}$ of case~2.14.
Modulo the induced $G_{2.14}$-equivalence, we can set \mbox{$1+4c_1=0$}.
Separating the variables in the resulting equation, denoting $z:=-1\pm\sqrt{4\kappa\varphi_\omega}$
and integrating once more, we obtain $z+\ln|z|=\kappa(\omega+c_2)$,
where $c_2$ is another integration constant
that also can be set to be equal zero up to the $G_{2.14}$-equivalence.
In other words, we derive the equation
\[
\omega=F(\varphi_\omega):=\left.\frac{z+\ln|z|}\kappa\right|_{z=-1\pm\sqrt{4\kappa\varphi_\omega}},
\]
and its general solution can be represented in parametric form as 
\begin{gather*}
\omega=F(\zeta)\quad\mbox{with}\quad\zeta:=\frac{(z+1)^2}{4\kappa},
\\
\varphi=\int\zeta\frac{{\rm d} F}{{\rm d}\zeta}(\zeta)\,{\rm d}\zeta
=\int\frac{(z+1)^3}{4\kappa^2z}\,{\rm d}z
=\frac1{4\kappa^2}\left(\frac{z^3}3+\frac32z^2+2z\right)+\frac\omega{4\kappa}+c_3,
\end{gather*}
where $c_3$ is one more integration constant
that also can be set to be equal to zero up to the $G_{2.14}$-equivalence.
The summand $\omega/(4\kappa)$ can be neglected using the $G_{2.14}$-equivalence as well.
Since $z{\rm e}^z=\pm{\rm e}^{\kappa\omega}$, we have that
$z\in\big\{W_0({\rm e}^{\kappa\omega}),\,W_0(-{\rm e}^{\kappa\omega}),\,W_{-1}(-{\rm e}^{\kappa\omega})\big\}$,
where~$W_0$ and~$W_{-1}$ again denote the principal real and the other real branches of the Lambert $W$ function, respectively.
As a result, we show that any solution of reduced equation~2.14$^{\delta\nu\delta'}$ with $\delta\delta'\ne0$
is $G_{2.14}$-equivalent to one of the solutions
\[
\varphi=\frac1{4\kappa^2}\left(\frac{z^3}3+\frac32z^2+2z\right), \quad
z\in\big\{W_0({\rm e}^{\kappa\omega}),\,W_0(-{\rm e}^{\kappa\omega}),\,W_{-1}(-{\rm e}^{\kappa\omega})\big\}.
\]
The corresponding $G$-equivalent solutions of the equation~\eqref{eq:dN} take the form
\[
\solution
u=\frac{(\nu^3-1)^2}{36}{\rm e}^{3t}\left(\frac{z^3}3+\frac32z^2+2z\right)
-\frac16(x^3+y^3)-\frac{\nu^3-1}2{\rm e}^ty^2,
\]
where
$z\in\big\{W_0({\rm e}^{\kappa\omega}),\,W_0(-{\rm e}^{\kappa\omega}),\,W_{-1}(-{\rm e}^{\kappa\omega})\big\}$
with $\omega:={\rm e}^{-t}(y-\nu x)$
and $\kappa:=-3(\nu^3-1)^{-1}$.

\subsection{The second collection of reductions}\label{sec:dNLieReductionsOfCodim2Collection2}

The reduced ordinary differential equations
that are obtained from the equation~\eqref{eq:dN} by Lie reductions
using the two-dimensional subalgebras from the second selected collection are cumbersome
and among them there are three one-parameter families of equations,
which complicates the computation of Lie and, moreover, point symmetries of these equations.
At the same time, there is a more essential obstacle even for computing Lie symmetries
just using the standard Lie approach augmented with specialized computer-algebra packages.
The general form of the above reduced equations is
$M(\omega,\varphi,\varphi_\omega,\varphi_{\omega\omega})\varphi_{\omega\omega\omega}
+N(\omega,\varphi,\varphi_\omega,\varphi_{\omega\omega})=0$,
where $M$ and~$N$ are respectively specific first- and second-degree polynomials
in $(\varphi,\varphi_\omega,\varphi_{\omega\omega})$ with coefficients polynomially depending on~$\omega$
that in addition satisfy the conditions
\begin{gather}\label{eq:dNLieReductionsOfCodim2Collection2ContraintsForRP}
M_{\varphi_{\omega\omega}}\ne0
\quad\mbox{or}\quad
M_{\varphi_{\omega\omega}}=0,\
M_{\varphi_\omega}(3M_{\varphi_\omega}+N_{\varphi_{\omega\omega}\varphi_{\omega\omega}})
(6M_{\varphi_\omega}+N_{\varphi_{\omega\omega}\varphi_{\omega\omega}})\ne0.
\end{gather}
Some of these equations cannot be represented in normal form due to their degeneration,
and the solution set of each of them splits into two parts
that are singled out by the constraints $M\ne0$ and $M=N=0$, respectively.
The left-hand sides of several of them even admit algebraic factorizations.
Therefore, the maximal Lie invariance algebra of such an equation $M\varphi_{\omega\omega\omega}+N=0$
is the intersection of the maximal Lie invariance algebras
of the equation $\varphi_{\omega\omega\omega}=-N/M$ with $M\ne0$ and of the (overdetermined) system $M=N=0$.
We prove that under the conditions~\eqref{eq:dNLieReductionsOfCodim2Collection2ContraintsForRP},
any Lie-symmetry vector field of the equation $\varphi_{\omega\omega\omega}=-N/M$ is
necessarily of the form $\xi\p_\omega+(\eta^1\varphi+\eta^0)\p_\varphi$,
where $\xi$, $\eta^1$ and~$\eta^0$ are functions of~$\omega$,
and then the computation of the maximal Lie invariance algebra of this equation
can be easily completed with a computer-algebra system even in the case of presence of a parameter.
After reducing the corresponding system $M=N=0$ to a passive form,
we also find its maximal Lie invariance algebra if its solution set is nonempty.
For each family of reduced equations under study in this section,
the construction of its point symmetry group is specific
and is carried out using the algebraic method by Hydon \cite{hydo1998a,hydo1998b,hydo2000b}.

\medskip\par\noindent{\bf 2.1.}\ 
$\mathfrak s_{2.1}^\lambda=\big\langle D^t(1),\,D^t(t)+\lambda D^{\rm s}\big\rangle$, $\lambda\ne-1/3$:
\begin{gather*}
u=|x|^\kappa\varphi,\quad
\smash{\omega=\frac yx, \quad \kappa:=\frac{9\lambda}{3\lambda+1}};\\[2ex]
\big(2\omega(\omega^3-1)\varphi_{\omega\omega}
-(\kappa-1)(3\omega^3-1)\varphi_\omega+\kappa(\kappa-1)\omega^2\varphi\big)
\varphi_{\omega\omega\omega}\\\ \ {}
-(\kappa-2)\big((5\omega^3-1)\varphi_{\omega\omega}{}^{\!\!2}
-(\kappa-1)\omega(11\omega\varphi_\omega-3\kappa\varphi)\varphi_{\omega\omega}
+(\kappa-1)^2(5\omega\varphi_\omega-2\kappa\varphi)\varphi_\omega\big)
=0.\!\!
\end{gather*}
In view of its definition, the parameter~$\kappa$ cannot be equal to~3
but we can neglect this fact by uniting reduction~2.1 with reduction~2.13,
which can be considered as corresponding to the values $\lambda=\pm\infty$ and $\kappa=3$,
see below.
Note that it is convenient to assume the family of reduced equations~2.1
to be parameterized~$\kappa$ instead of~$\lambda$.

The associated system $M=N=0$ is equivalent to the equation
\begin{gather*}
\varphi_\omega=0\quad\mbox{if}\quad\kappa=0,\\
\varphi_{\omega\omega}=0\quad\mbox{if}\quad\kappa=1,\\
2\omega(\omega^3-1)\varphi_{\omega\omega}-(3\omega^3-1)\varphi_\omega+2\omega^2\varphi=0\quad\mbox{if}\quad\kappa=2,\\
(w^3+1)\varphi_\omega=3\omega^2\varphi\quad\mbox{if}\quad \kappa=3,\\
(w^6-10w^3+1)\varphi_\omega=6\omega^2(w^3-5)\varphi\quad\mbox{if}\quad\kappa=6,\\
\varphi=0\quad\mbox{otherwise.}
\end{gather*}
The corresponding solutions of the original equation~\eqref{eq:dN}
belong to the family of trivial solutions~\eqref{eq:dNTrivialSolutions},
except the cases $\kappa=2$, see the consideration of this case below,
and $\kappa=6$, with the polynomial solutions $u=c_1(x^6-10x^3y^3+y^6)$ of~\eqref{eq:dN}, where $c_1$ is an arbitrary constant.

The normalizer of the subalgebra~$\mathfrak s_{2.1}^\lambda$ in~$\mathfrak g$ is
\begin{gather*}
{\rm N}_{\mathfrak g}(\mathfrak s_{2.1}^\lambda)=\big\langle D^t(1),D^t(t),D^{\rm s}\big\rangle\quad\mbox{if}\quad\lambda\ne0,1/6,\\[.2ex]
{\rm N}_{\mathfrak g}(\mathfrak s_{2.1}^0)=\big\langle D^t(1),D^t(t),D^{\rm s},Z(1)\big\rangle,\\
{\rm N}_{\mathfrak g}(\mathfrak s_{2.1}^{1/6})=\big\langle D^t(1),D^t(t),D^{\rm s},R^x(1),R^y(1)\big\rangle.
\end{gather*}
For a general value $\lambda\ne-1/3$, the vector fields $D^t(1)$, $3D^t(t)$ and $D^{\rm s}$
induce the Lie-symmetry vector fields $0$, $-\kappa\varphi\p_\varphi$ and $(3-\kappa)\varphi\p_\varphi$
of reduced equation~2.1$^\kappa$,
whereas for $\lambda=0$ and $\lambda=1/6$ (i.e., $\kappa=0$ and $\kappa=1$)
we in addition have inductions
$\p_\varphi$ by~$Z(1)$
and
$\p_\varphi$ and $\omega\p_\varphi$ by~$R^x(1)$ and $R^y(1)$,
respectively.
The discrete point symmetry transformations~$\mathscr J$ and~$\mathscr I^{\rm s}$ of the equation~\eqref{eq:dN},
see Corollary~\ref{cor:dNDiscrSyms}, induce the discrete point symmetry transformations
$(\tilde\omega,\tilde\varphi)=(\omega^{-1},|\omega|^{-\kappa}\varphi)$
and $(\tilde\omega,\tilde\varphi)=(\omega,-\varphi)$ of reduced equation~2.1$^\kappa$
for any~$\kappa$, respectively,
whereas the discrete point symmetry transformation~$\mathscr I^{\rm i}$ of the equation~\eqref{eq:dN}
corresponds to the identity transformation of $(\omega,\varphi)$.

For a general value $\lambda\ne-1/3$, the subalgebra~$\mathfrak s_{2.1}^\lambda$
has, up to the $G_{\rm L}$-equivalence, a single counterpart among subalgebras
of the algebra~$\mathfrak g_{\rm L}$,
$\bar{\mathfrak s}_{2.1}^\lambda=\big\langle\bar D^t(1),\,\bar D^t(t)+\lambda\bar D^{\rm s}\big\rangle$.
This is why ansatz~2.1 is extended to~$v$ as $v=|x|^{\kappa/2}\psi$,
and the nonlinear Lax representation~\eqref{eq:dNLaxPair} reduces to the system
\begin{gather}\label{eq:dNReduction2.1LaxPairGen}
\begin{split}&
12(\kappa-1)\big(2(\omega^3+1)\psi_\omega-\kappa\omega^2\psi\big)\varphi_\omega
-12\kappa(\kappa-1)\omega\big(2\omega\psi_\omega-\kappa\psi\big)\varphi
\\&
\qquad{}
+16\omega(\omega^3-1)\psi_\omega^{\,\,3}
-12\kappa(\omega^3-1)\psi\psi_\omega^{\,\,2}
+\kappa^3\omega\psi^3=0,
\\[.5ex]&
\omega\varphi_{\omega\omega}-(\kappa-1)\varphi_\omega
+\omega\psi_\omega^{\,\,2}-\frac\kappa2\psi\psi_\omega=0.
\end{split}
\end{gather}
For each of the values $\lambda=0$ and $\lambda=2/3$,
where $\kappa=0$ and $\kappa=2$,
there is another counterpart of the subalgebra~$\mathfrak s_{2.1}^\lambda$
among subalgebras of the algebra~$\mathfrak g_{\rm L}$
that is $G_{\rm L}$-inequivalent to the subalgebra~$\bar{\mathfrak s}_{2.1}^\lambda$,
$\bar{\mathfrak s}_{2.1'}^0=\big\langle\bar D^t(1),\,\bar D^t(t)+\frac13\bar P^v\big\rangle$ and
\smash{$\bar{\mathfrak s}_{2.1'}^{2/3}=\big\langle\bar D^t(1)+\bar P^v,\,\bar D^t(t)+\frac23\bar D^{\rm s}\big\rangle$},
respectively.
This results in one more $G_{\rm L}$-inequivalent extension of ansatz~2.1 to~$v$ for each of these values of~$\lambda$,
$v=\psi+\ln|x|$ and $v=x\psi+t$.
The corresponding reduced systems are
\begin{gather*}
\begin{split}&
3\big((\omega^3+1)\psi_\omega-\omega^2\big)\varphi_\omega
-2\omega(\omega^3-1)\psi_\omega^{\,\,3}
+3(\omega^3-1)\psi_\omega^{\,\,2}-\omega=0,
\\&
\omega\varphi_{\omega\omega}+\varphi_\omega
+\omega\psi_\omega^{\,\,2}-\psi_\omega=0
\end{split}
\end{gather*}
and
\begin{gather*}
\begin{split}&
3\big((\omega^3+1)\psi_\omega-\omega^2\psi\big)\varphi_\omega
-6\omega(\omega\psi_\omega-\psi)\varphi\\&
\qquad{}+2\omega(\omega^3-1)\psi_\omega^{\,\,3}
-3(\omega^3-1)\psi\psi_\omega^{\,\,2}+\omega\psi^3-3\omega=0,
\\[.5ex]&
\omega\varphi_{\omega\omega}-\varphi_\omega
+\omega\psi_\omega^{\,\,2}-\psi\psi_\omega=0.
\end{split}
\end{gather*}

For the specific values $\kappa\in\{0,1,2\}$ or, equivalently,~$\lambda\in\{0,1/6,2/3\}$,
we are able to construct more solutions than for the other values.

\medskip\par\noindent{$\kappa=2$.}
The solution set of reduced equation~2.1$^\kappa$ with $\kappa=2$, i.e., $\lambda=2/3$,
is a union of the solution sets of the equations $\varphi_{\omega\omega\omega}=0$ and
$2\omega(\omega^3-1)\varphi_{\omega\omega}-(3\omega^3-1)\varphi_\omega+2\omega^2\varphi=0$,
whose intersection consists only of the zero solution $\varphi=0$.
It can be proved that the maximal Lie invariance algebra~$\mathfrak a_{2.1}^2$
of reduced equation~2.1$^2$ is the intersection of the maximal Lie invariance algebras
of the above equations, $\mathfrak a_{2.1}^2=\langle\varphi\p_\varphi\rangle$,
and thus it is induced by ${\rm N}_{\mathfrak g}(\mathfrak s_{2.1}^2)$.

The solutions of the first equation are not interesting
since each related solution of the equation~\eqref{eq:dN} is $G$-equivalent
to either the zero solution $u=0$ or the solution~\eqref{eq:InvSolutions1.4} with $\beta=1$.
The general solution of the second equation is
\begin{gather*}
\varphi=c_1(\omega^{3/2}+1)^{4/3}+c_2(\omega^{3/2}-1)^{4/3}\quad\mbox{for}\quad\omega\geqslant0,\\
\varphi=(1-\omega^3)^{2/3}\left(
 c_1\cos\left(\frac43\arctan|\omega|^{3/2}\right)
+c_2\sin\left(\frac43\arctan|\omega|^{3/2}\right)\right)
\quad\mbox{for}\quad\omega\leqslant0,
\end{gather*}
where $c_1$ and $c_2$ are arbitrary constants,
and one of them, if nonzero, can be set to one by induced symmetries of reduced equation~2.1$^2$.
This leads to the following solutions of the equation~\eqref{eq:dN}:
\begin{gather*}
\solution
u=c_1(|x|^{3/2}+|y|^{3/2})^{4/3}+c_2(|y|^{3/2}-|x|^{3/2})^{4/3}\quad\mbox{for}\quad xy\geqslant0,\\
\solution
u=(x^3-y^3)^{2/3}\left(
 c_1\cos\left(\frac43\arctan\left|\frac yx\right|^{3/2}\right)
+c_2\sin\left(\frac43\arctan\left|\frac yx\right|^{3/2}\right)\right)
\quad\mbox{for}\quad xy\leqslant0,
\end{gather*}
where $c_1$ and $c_2$ are arbitrary constants,
and one of them, if nonzero, can be set to one up to the $G$-equivalence.

\medskip\par\noindent{$\kappa=1$.}
The solution set of reduced equation~2.1$^\kappa$ with $\kappa=1$, i.e., $\lambda=1/6$,
coincides with that of the equation
$2\omega(\omega^3-1)\varphi_{\omega\omega\omega}+(5\omega^3-1)\varphi_{\omega\omega}=0$
and thus consists of the functions
\[
\varphi=c_1\varphi^0(\omega)
+c_2\omega+c_3
\quad\mbox{with}\quad\varphi^0(\omega):=|\omega|^{3/2}\,{}_3F_2\!\left(\frac16,\frac12,\frac23;\frac76,\frac32;w^3\right),
\]
where ${}_pF_q(a_1,\dots,a_p;b_1,\dots,b_q;z)$ is the generalized hypergeometric function.
Hence the maximal Lie invariance algebra of reduced equation~2.1$^1$ is
$\mathfrak a_{2.1}^1=\langle\p_\varphi,\omega\p_\varphi,\varphi^0(\omega)\p_\varphi,\varphi\p_\varphi\rangle$,
and its subalgebra induced by ${\rm N}_{\mathfrak g}(\mathfrak s_{2.1}^1)$ is
$\langle\p_\varphi,\omega\p_\varphi,\varphi\p_\varphi\rangle$,
i.e., any element of $\mathfrak a_{2.1}^1$ with nonzero coefficient of $\varphi^0(\omega)\p_\varphi$
is a hidden Lie-symmetry of the equation~\eqref{eq:dN}
that is associated with reduction~2.1$^1$.

Up to induced symmetries of reduced equation~2.1$^1$,
we can set $c_1=1$, 
and $c_2=c_3=0$.
Thus, the only corresponding $G$-inequivalent solutions of the equation~\eqref{eq:dN} is
\[
\solution
u=
\sqrt{\left|\frac{y^3}x\right|}\,\,{}_3F_2\!\left(\frac16,\frac12,\frac23;\frac76,\frac32;\frac{y^3}{x^3}\right).
\]

\medskip\par\noindent{$\kappa=0$.}
The maximal Lie invariance algebra~$\mathfrak a_{2.1}^0$
of reduced equation~2.1$^0$ is equal to $\langle\p_\varphi,\varphi\p_\varphi\rangle$,
and thus it is entirely induced by ${\rm N}_{\mathfrak g}(\mathfrak s_{2.1}^0)$.
The reduced system~\eqref{eq:dNReduction2.1LaxPairGen} with $\kappa=0$ degenerates to
\begin{gather}\label{eq:dNReduction2.1LaxPairGenKappa0}
\psi_\omega\big(3(\omega^3+1)\varphi_\omega-2\omega(\omega^3-1)\psi_\omega^{\,\,2}\big)=0,
\quad
\omega\varphi_{\omega\omega}+\varphi_\omega+\omega\psi_\omega^{\,\,2}=0.
\end{gather}
It is obvious that the integration of the system~\eqref{eq:dNReduction2.1LaxPairGenKappa0} splits into two cases.
If $\psi_\omega=0$, then it is equivalent to the equation $\omega\varphi_{\omega\omega}+\varphi_\omega=0$,
whose general solution is $\varphi=c_1\ln|\omega|+c_0$
and gives only trivial solutions of the equation~\eqref{eq:dN} from the family~\eqref{eq:dNTrivialSolutions}.
Under the constraint $\psi_\omega\ne0$,
we easily exclude $\psi_\omega$ from the system~\eqref{eq:dNReduction2.1LaxPairGenKappa0}
and derive the equation $2\omega(\omega^3-1)\varphi_{\omega\omega}+(5\omega^3+1)\varphi_\omega=0$.
The general solution of this equation is
\begin{gather*}
\varphi=c_1\ln\left|\frac{\omega^{3/2}+1}{\omega^{3/2}-1}\right|+c_2\quad\mbox{for}\quad\omega\geqslant0,\qquad
\varphi=c_1\arctan|\omega|^{3/2}+c_2\quad\mbox{for}\quad\omega\leqslant0.
\end{gather*}
Up to induced symmetries of reduced equation~2.1$^0$,
we can set $c_1=1$, and $c_2=0$.
This leads to the following solutions of the equation~\eqref{eq:dN}:
\begin{gather*}
\solution
u=\ln\left|\frac{|x|^{3/2}+|y|^{3/2}}{|x|^{3/2}-|y|^{3/2}}\right|\quad\mbox{for}\quad xy\geqslant0,\qquad\qquad
\solution
u=\arctan\left|\frac yx\right|^{3/2}
\quad\mbox{for}\quad xy\leqslant0.
\end{gather*}
The independent variable~$\omega$ and the ratio $\varphi_{\omega\omega}/\varphi_\omega$
are the lowest-order differential invariants of the solvable algebra~$\mathfrak a_{2.1}^0$.
Therefore, the change of dependent variable $p=\varphi_{\omega\omega}/\varphi_\omega$
lowers the order of reduced equation~2.1$^0$ by two.
The derived equation
\[
(2\omega(\omega^3-1)p+3\omega^3-1)p_\omega+2\omega(\omega^3-1)p^3+(13\omega^3-3)p^2+22p\omega^2+10\omega=0
\]
integrates to
\[
\frac{\big(\omega(\omega^3-1)^2p^2+(3\omega^3-1)(\omega^3-1)p+\omega^2(2\omega^3-5)\big)^3}
{\big(2\omega(\omega^3-1)p+5\omega^3+1\big)^4\big(\omega p+1\big)^2}=c_1.
\]
Substituting $\varphi_{\omega\omega}/\varphi_\varphi$ for~$p$
into the last equation, we obtain the first integral of reduced equation~2.1$^0$.
We are not able to integrate further for the general value of~$c_1$.
Nevertheless, setting $c_1=0$ simplifies the equation into be integrated to
the equation
\[
\omega(\omega^3-1)^2\varphi_{\omega\omega}^{\,\,2}
+(3\omega^3-1)(\omega^3-1)\varphi_{\omega\omega}\varphi_\omega
+\omega^2(2\omega^3-5)\varphi_\omega^{\,\,2}=0,
\]
which is easily solved as a quadratic equation with respect to~$\varphi_{\omega\omega}/\varphi_\omega$
and integrated twice.
As a result, we construct the following solution of reduced equation~2.1$^\kappa$ with $\kappa=0$:
\begin{gather*}
\varphi=\int
\frac{|s+7+K|^{1/6}|7s+1+K|^{1/6}(2s+2-K)^{2/3}}{3s(s-1)}
\,\bigg|_{K=\sqrt{s^2+14s+1}}
\,{\rm d}s\,\bigg|_{s=\omega^3}.
\end{gather*}
The corresponding solution of the original equation~\eqref{eq:dN} is
\begin{gather*}
\solution
u=\int
\frac{|s+7+K|^{1/6}|7s+1+K|^{1/6}(2s+2-K)^{2/3}}{3s(s-1)}
\,\bigg|_{K=\sqrt{s^2+14s+1}}
\,{\rm d}s\,\bigg|_{s=(y/x)^3}.
\end{gather*}

For the general value of~$\kappa$, $\kappa\ne0,1,2$,
all Lie and discrete point symmetries of the associated equation $\varphi_{\omega\omega\omega}=-N/M$
are symmetries of the system $M=N=0$.
This is why the maximal Lie invariance algebra~$\mathfrak a_{2.1}^\kappa$
of reduced equation~2.1$^\kappa$ is equal to $\langle\varphi\p_\varphi\rangle$,
and thus it is entirely induced by ${\rm N}_{\mathfrak g}(\mathfrak s_{2.1}^\kappa)$.

We compute the point symmetry group~$G_{2.1}^\kappa$ of reduced equation~2.1$^\kappa$
with an arbitrary nonsingular value $\kappa\ne0,1,2$ by the algebraic method.
Let $\Phi$: $\tilde\omega=\Omega(\omega,\varphi)$, $\tilde\varphi=F(\omega,\varphi)$
with $\Omega_\omega F_\varphi-\Omega_\varphi F_\omega\ne0$
be a point symmetry transformation of this equation.
From the condition $\Phi_*\mathfrak a_{2.1}\subseteq\mathfrak a_{2.1}$,
we only derive the equations $\Omega_\varphi=0$ and $\varphi F_\varphi=aF$,
which mean that $\Omega=\Omega(\omega)$ with $\Omega_\omega\ne0$ and
$F=g(\omega)\varphi^a+f(\omega)$
with a nonzero constant~$a$, a nonvanishing function~$g$ of~$\omega$ and a function~$f$ of~$\omega$.
The further computation by the direct method is the most complicated among such computations in the present paper.
The left-hand side $L[\varphi]$ of reduced equation~2.1$^\kappa$ is a homogeneous second-degree polynomial
with respect to the unknown function~$\varphi$ and its derivatives with coefficients depending on~$\omega$.
After expanding the transformed equation with taking into account the obtained form of~$\Phi$
and collecting the coefficients of $\varphi_{\omega\omega}{}^{\!3}\varphi_\omega$,
we first derive the equation $a=1$,
which means that the transformation~$\Phi$ is affine with respect to~$\varphi$.
Then the condition of preserving reduced equation~2.1$^\kappa$ by~$\Phi$
can be written in the form $L[\Phi_*\varphi]=K(\omega)L[\varphi]$,
where $K$ is a nonvanishing function of~$\omega$.
Collecting, in the last equality, coefficients of the terms that are of degree two
with respect to the unknown function~$\varphi$ and its derivatives
leads to a system of determining equations for the functions~$\Omega$, $g$ and~$K$,
whose general solution consists of two families,
$(\Omega,g,K)=(\omega,c_1,c_1^{\,2})$ and
$(\Omega,g,K)=(\omega^{-1},c_1\omega^{-\kappa},c_1^{\,2}\omega^{-2\kappa-11})$,
where $c_1$ is an arbitrary nonzero constant.
For each of the found solutions for $(\Omega,g,K)$,
the system for~$f$ derived by collecting coefficients of the remaining terms
only has the zero solution.
As a result, for any value of~$\kappa$
the entire group~$G_{2.1}^\kappa$ is induced by the stabilizer of~$\mathfrak s_{2.1}^\lambda$ in~$G$
with the corresponding value of~$\lambda$.

Polynomial solutions of reduced equations~2.1$^\kappa$
whose degree is not greater than five and that result in nonpolynomial solutions
of the original equation~\eqref{eq:dN} are exhausted by
$\varphi=c\omega$ for $\kappa=5/2$ and
$\varphi=c(\omega^3-8/21)$ for $\kappa=9/2$,
where $c=1$ modulo the induced $G_{1.1}$-equivalence.
The first solution corresponds to
the solution $u=|x|^{3/2}y$ of~\eqref{eq:dN}, which can also be obtained and, moreover, generalized
using the multiplicative separation of variables;
make the permutation~$\mathscr J$ of~$x$ and~$y$ in the last solution of Section~\ref{sec:dNMultiplicativeSeparationOfVars}.
The first solution gives a new solution of~\eqref{eq:dN},
\[
\solution
u=|x|^{9/2}\left(\frac{y^3}{x^3}-\frac 8{21}\right).
\]

\par\noindent{\bf 2.2.}\  
$\mathfrak s_{2.2}^\nu=\big\langle D^t(1),\,D^t(t)-\tfrac13 D^{\rm s}+P^x(1)+P^y(\nu)\big\rangle$, \ $|\nu|\leqslant1$ $(\!{}\bmod G)$:
\begin{gather*}
u={\rm e}^{-x}\varphi,\quad  \omega=y-\nu x;\\
\big(2\nu(\nu^3-1)\varphi_{\omega\omega}+(3\nu^3-1)\varphi_\omega+\nu^2\varphi\big)\varphi_{\omega\omega\omega}\\
\qquad{}+(5\nu^3-1)\varphi_{\omega\omega}{}^{\!\!2}
+\nu(11\nu\varphi_\omega+3\varphi)\varphi_{\omega\omega}
+5\nu\varphi_\omega^{\,\,2}+2\varphi\varphi_\omega=0.
\end{gather*}
The associated system $M=N=0$ (see the beginning of this section)
is equivalent to the equation
$\varphi_\omega=0$ if $\nu=0$,
$2\varphi_\omega=\varphi$ if $\nu=-1$ or
$\varphi=0$ otherwise;
the corresponding solutions of the original equation~\eqref{eq:dN}
belong to the family of trivial solutions~\eqref{eq:dNTrivialSolutions}
or to the family~\eqref{eq:dNs1.3Rho1InvarSolutions}.
All Lie and discrete point symmetries of the associated equation $\varphi_{\omega\omega\omega}=-N/M$
are symmetries of the system $M=N=0$.
This is why for any value of $\nu$, the maximal Lie invariance algebra of reduced equation~2.2$^\nu$
is the algebra $\mathfrak a_{2.2}=\langle\p_\omega,\varphi\p_\varphi\rangle$,
and this equation is invariant with respect to the group~$G_{2.2}$,
which consists of the point transformations
$\tilde\omega=\omega+\tilde c_1$, $\tilde\varphi=\tilde c_2\varphi$,
where $\tilde c_1$ and~$\tilde c_2$ are arbitrary constants with $\tilde c_2\ne0$.
All the subalgebras~$\mathfrak s_{2.2}^\nu$ have the same normalizer
${\rm N}_{\mathfrak g}(\mathfrak s_{2.2}^\nu)=\langle D^t(1),D^t(t)-\tfrac13D^{\rm s},P^x(1),P^y(1)\rangle$
in~$\mathfrak g$.
The vector fields $D^t(1)$, $D^t(t)-\tfrac13D^{\rm s}$, $P^x(1)$ and $P^y(1)$
from ${\rm N}_{\mathfrak g}(\mathfrak s_{2.2}^\nu)$
induce the Lie-symmetry vector fields
$0$, $-\varphi\p_\varphi$, $-\nu\p_\omega+\varphi\p_\varphi$ and $\p_\omega$
of reduced equation~2.2$^{\nu}$, respectively,
and thus the algebra~$\mathfrak a_{2.2}$ is entirely induced by elements of ${\rm N}_{\mathfrak g}(\mathfrak s_{2.2}^{\nu})$.
The discrete point symmetry transformation
$\mathscr I^{\rm i}\circ\mathscr I^{\rm s}$: $(\tilde t,\tilde x,\tilde y,\tilde u)=(-t,x,y,-u)$
of~\eqref{eq:dN} induces the discrete point symmetry transformation $(\tilde\omega,\tilde\varphi)=(\omega,-\varphi)$
for any of reduced equations~2.2$^\nu$.
Therefore, the entire group~$G_{2.2}$ is induced by the point symmetry group~$G$ of the original equation~\eqref{eq:dN}.

We construct the point symmetry group~$G_{2.2}^\nu$ of reduced equation~2.2$^\nu$
with an arbitrary fixed~$\nu$ using the algebraic method.
Let $\Phi$: $\tilde\omega=\Omega(\omega,\varphi)$, $\tilde\varphi=F(\omega,\varphi)$
with $\Omega_\omega F_\varphi-\Omega_\varphi F_\omega\ne0$
be a point symmetry transformation of this equation.
The necessary condition $\Phi_*\mathfrak a_{2.2}\subseteq\mathfrak a_{2.2}$
implies the equations $\Omega_\omega=a_{11}$, $\varphi\Omega_\varphi=a_{21}$, $F_\omega=a_{12}F$ and $\varphi F_\varphi=a_{22}F$,
where $a_{11}$, $a_{12}$, $a_{21}$ and~$a_{22}$ are constants with $a_{11}a_{22}-a_{12}a_{21}\ne0$.
Therefore,
\[
\Omega=a_{11}\omega+a_{21}\ln|\varphi|+\tilde c_1,\quad
F=\tilde c_2{\rm e}^{a_{12}\omega}\varphi^{a_{22}},
\]
where $\tilde c_1$ and~$\tilde c_2$ are arbitrary constants with $\tilde c_2\ne0$.
We continue the computation with the direct method using the derived form for~$\Phi$,
which leads to a cumbersome overdetermined system of determining equations
for the parameters~$a_{11}$, $a_{12}$, $a_{21}$ and $a_{22}$,
whose solution depends on the value of~$\nu$.
For $\nu\ne\pm1$, we obtain the single solution $a_{12}=a_{21}=0$, $a_{11}=a_{22}=1$,
i.e., the complete point symmetry group~$G_{2.2}^\nu$ of reduced equation~2.2$^\nu$
with such values of~$\nu$ coincides with the common point symmetry group~$G_{2.2}$.
For each $\nu\in\{-1,1\}$, there is exactly one more solution
$a_{11}=-1$, $a_{21}=0$, $a_{12}=\nu$, $a_{22}=1$,
i.e., in addition to the elements of~$G_{2.2}$,
the complete point symmetry group~$G_{2.2}^\nu$ of reduced equation~2.2$^\nu$ with $\nu=\pm1$
contains the transformations
$\tilde\omega=-\omega+\tilde c_1$,
$\tilde\varphi=\tilde c_2{\rm e}^{\nu\omega}\varphi$,
where $\tilde c_1$ and~$\tilde c_2$ are again arbitrary constants with $\tilde c_2\ne0$.
This means that the group~$G_{2.2}^\nu$ with $\nu=1$ or $\nu=-1$ is generated by the elements of~$G_{2.2}$
and the discrete point symmetry transformation
$(\tilde\omega,\tilde\varphi)=(-\omega,{\rm e}^{\nu\omega}\varphi)$,
which is induced by the discrete point symmetry~$\mathscr J$ or $\mathscr J\circ\mathscr I^{\rm s}$
of the original equation~\eqref{eq:dN}, respectively.
Therefore, for any value of~$\nu$
the group~$G_{2.2}^\nu$ is entirely induced by the stabilizer of~$\mathfrak s_{2.2}^\nu$ in~$G$.

Up to the $G_{\rm L}$-equivalence, the subalgebra~$\mathfrak s_{2.2}^\nu$ with any value of~$\nu$
is prolonged in a unique way to~$v$, which leads to the subalgebra
$\bar{\mathfrak s}_{2.2}^\nu=\big\langle\bar D^t(1),\,\bar D^t(t)-\tfrac13\bar D^{\rm s}+\bar P^x(1)+\bar P^y(\nu)\big\rangle$ of the algebra~$\mathfrak g_{\rm L}$.
Therefore, up to the $G_{\rm L}$-equivalence, there is a unique extension \smash{$v={\rm e}^{-x/2}\psi$} of ansatz~2.2 to~$v$,
and the corresponding reduced system is
\begin{gather*}
\noprint{
12\big(2(\nu^3-1)\psi_\omega+\nu^2\psi\big)\varphi_{\omega\omega}
+24\nu(2\nu\psi_\omega+\psi)\varphi_\omega+12(2\nu\psi_\omega+\psi)\varphi
\\\qquad{}
+8(\nu^3-1)\psi_\omega^{\,\,3}
+12\nu^2\psi\psi_\omega^{\,\,2}
+6\nu\psi^2\psi_\omega+\psi^3=0,
\\[1ex]
}
24\psi_\omega\varphi_{\omega\omega}
-12(2\nu\psi_\omega+\psi)(\nu\varphi_\omega+\varphi)
+8(2\nu^3+1)\psi_\omega^{\,\,3}
+12\nu^2\psi\psi_\omega^{\,\,2}-\psi^3=0,
\\
\nu\varphi_{\omega\omega}+\varphi_\omega+\nu\psi_\omega^{\,\,2}
+\frac12\psi\psi_\omega=0.
\end{gather*}

It is obvious that an arbitrary function of the form $\varphi=c_1{\rm e}^{-\omega/\nu}+c_2$
if $\nu\ne0$ or an arbitrary constant if $\nu=0$ is a solution of reduced equation~2.2$^\nu$,
and these solutions lead to trivial solutions of the equation~\eqref{eq:dN}
from the family~\eqref{eq:dNTrivialSolutions}.
Further we ignore the above trivial solutions.

Reduced equation~2.2$^0$ is especially short,
$\varphi_\omega\varphi_{\omega\omega\omega}+\varphi_{\omega\omega}{}^{\!\!2}=2\varphi\varphi_\omega$,
and integrates twice to $\varphi_\omega^{\,\,3}=\varphi^3+c_1\varphi+c_2$,
where $c_1$ and~$c_2$ are the integration constants.
Separating the variables and integrating further,
we construct the general solution of reduced equation~2.2$^0$ in an implicit form with one quadrature,
\begin{gather}\label{eq:ReducedEq2.2_0ImplicitExpression}
\int\frac{{\rm d}\varphi}{(\varphi^3+c_1\varphi+c_2)^{1/3}}=\omega+c_3,
\end{gather}
where $c_3$ is one more integration constant.
The corresponding solutions of the original equation~\eqref{eq:dN} are of the form
\begin{gather}\label{eq:s2.2InvSolutions}
\solution
u={\rm e}^{-x}\varphi(y),
\end{gather}
where the function $\varphi=\varphi(y)$
is implicitly defined by~\eqref{eq:ReducedEq2.2_0ImplicitExpression},
where $\omega=y$, $c_3=0 \pmod G$,
and, up to the $G$-equivalence,
the constant~$c_1$, if it is nonzero, can be set to be equal $\pm1$
or the constant~$c_2$, if it is nonzero, can be set to be equal one.
The solution family~\eqref{eq:s2.2InvSolutions} can be extended
using the multiplicative separation of variables,
see Section~\ref{sec:dNMultiplicativeSeparationOfVars}.
The formula~\eqref{eq:ReducedEq2.2_0ImplicitExpression} obviously leads to explicit solutions
of reduced equation~2.2$^0$ only if $c_1=c_2=0$,
which gives $\varphi=\tilde c_3{\rm e}^\omega$ with an arbitrary constant~$\tilde c_3$.
All the corresponding solutions of the equation~\eqref{eq:dN}, $u=\tilde c_3{\rm e}^{y-x}$,
belong to the family of simple solutions~\eqref{eq:dNs1.3Rho1InvarSolutions}.

For several specific values of $(c_1,c_2)$,
when the integral in the left-hand side of~\eqref{eq:ReducedEq2.2_0ImplicitExpression}
is reduced to cases of the Chebyshev theorem on the integration of binomial differentials,
it be expressed in terms of elementary functions.
This gives the following $G_{2.2}$-inequivalent
parametric solutions (without quadratures) of reduced equation~2.2$^0$:
\begin{alignat}{3}&\nonumber
\hspace*{-\mylength}\circ\quad
\varphi=\frac{s^3+2}{s^3-1}&&
\quad\mbox{with}\quad\frac12\ln\frac{s^2+s+1}{(s-1)^2}-\sqrt3\arctan\frac{2s+1}{\sqrt3}= y,
&\\[1ex]&\label{eq:ReducedEq2.2_0ImplicitExpression2}
\hspace*{-\mylength}\circ\quad
\varphi=|s^3-1|^{-1/2}&&
\quad\mbox{with}\quad\frac12\ln\frac{s^2+s+1}{(s-1)^2}-\sqrt3\arctan\frac{2s+1}{\sqrt3}=2y,&\\[1ex]&\nonumber
\hspace*{-\mylength}\circ\quad
\varphi=(s^3-1)^{-1/3}&&
\quad\mbox{with}\quad\frac12\ln\frac{s^2+s+1}{(s-1)^2}-\sqrt3\arctan\frac{2s+1}{\sqrt3}=3y&
\end{alignat}
if $4c_1^3=-27c_2^2$, ($c_1\ne0$, $c_2=0$) and ($c_1=0$, $c_2\ne0$), respectively;
cf.\ \cite[Section~4.1.1.2]{moro2021a},
where there are several typos and a needless involvement of complex numbers.
In the first case, the polynomial $\varphi^3+c_1\varphi+c_2$
has a root~$\lambda$ of multiplicity two and can thus be factorized to $(\varphi+2\lambda)(\varphi-\lambda)^2$,
i.e., $c_1=-3\lambda^2$ and $c_2=2\lambda^3$.
It is then obvious that we can set $\lambda=1$ up to $G_{2.2}$-equivalence, more precisely, by scaling of~$\varphi$.
In the second and the third cases, we analogously can set $c_1=\sgn(s^3-1)$ and $c_2=1$, respectively.
For $c_1\ne0$, the integral in the left-hand side of~\eqref{eq:ReducedEq2.2_0ImplicitExpression}
was reduced in \cite[Eqs.~(33)--(34)]{moro2021a} to an integral
that, as stated therein, can be expressed in terms of elliptic functions
but the corresponding representation of the solution~\eqref{eq:ReducedEq2.2_0ImplicitExpression} does not seem useful.

For general values of~$\nu$, the order of reduced equation~2.2$^\nu$
can be lowered by the differential substitution $z=\varphi_\omega/\varphi$, $p=\varphi_{\omega\omega}/\varphi$
inspired by the Lie invariance algebra~$\mathfrak a_{2.2}$.
This leads to a~first-order ordinary differential equation with respect to~$p=p(z)$,
\begin{gather}\label{eq:ReducedEq2.2Lowered}
\begin{split}&
\big(2\nu(\nu^3-1)p+3\nu^3z-z+\nu^2\big)(p-z^2)p_z\\&\qquad{}
+\big(2\nu(\nu^3-1)z+5\nu^3-1\big)p^2
+\big((3\nu^3-1)z^2+12\nu^2z+3\nu\big)p+5\nu z^2+2z=0.
\end{split}
\end{gather}

Looking for solutions of~\eqref{eq:ReducedEq2.2Lowered} that are at most quadratic with respect to~$z$,
we construct only the solutions
$p=-z/\nu$ if $\nu=0$,
$p=4$ if $\nu=1/2$ and
$p=z-1$ if $\nu=-1$.
For reduced equation~2.2$^\nu$, this gives
the above trivial solutions $\varphi=c_1{\rm e}^{-\omega/\nu}+c_2$ if $\nu=0$ as well as
$\varphi=c_1{\rm e}^{-2\omega}+c_2{\rm e}^{2\omega}$ if $\nu=1/2$ and
$\varphi=c_1{\rm e}^{\omega/2}\cos(\frac12\sqrt3\omega+c_2)$ if $\nu=-1$.
As a result, up to $G$-equivalence we construct the following new solutions of the original equation~\eqref{eq:dN}:
\[
\solution u={\rm e}^{y-x}\pm{\rm e}^{-y},\qquad\qquad
\solution u={\rm e}^{y-x}\cos\big(\sqrt3(x+y)\big).
\]

If $\nu=1$, the equation~\eqref{eq:ReducedEq2.2Lowered} becomes the Abel equation of the second kind,
\[(2z+1)(p-z^2)p_z+4p^2+(2z^2+12z+3)p+5z^2+2z=0,\]
which is reduced by the point transformation $s=z+\frac12$, $r=(z+\frac12)^2(p-z^2)$ 
to the simpler Abel equation
\[
16rr_s+4s(28s^2-1)r+s^3(4s^2-1)(12s^2+1)=0,
\]
whose general solution in implicit form is
\[
\frac{(4r+4s^4-s^2)^2\big((144r+144s^4-1)^2-(12s^2+1)^3\big)}
{\big(3\big(32r+(8s^2+1)(4s^2-1)\big)^2-(8s^2+1)(4s^2-1)^2\big)^2}=c_1.
\]

The latter equation has two polynomial solutions up to degree four,
$r=-\frac14 s^2(4s^2-1)$ and
$r=-\frac1{64}(4s^2-1)(12s^2+1)$,
which correspond to the values $c_1=0$ and $c_1=1/256$ and the solutions
\[
p=-z
\quad\mbox{and}\quad
p=\frac{z(z^3-2z^2-3z-1)}{(2z+1)^2}
\]
of the former Abel equation, respectively.
After the inverse differential substitution,
we respectively obtain two ordinary differential equations.
The integration of the first one only results in trivial solutions
of the original equation~\eqref{eq:dN},
whereas solving the second equation, we construct, up to the $G$-equivalence, the following parametric solution of~\eqref{eq:dN}:
\[
\solution
u={\rm e}^{-(x+y)/2}\frac{(3z^2+3z+1)^{-1/6}}{|z|^{1/2}|z+1|^{1/2}}
\quad\mbox{with}\quad
\ln\left|\frac{z+1}z\right|-\frac2{\sqrt3}\arctan\big(\sqrt3(2z+1)\big)=y-x.
\]

\par\noindent{\bf 2.3.}\   
$\tilde{\mathfrak s}_{2.3}^\nu=\big\langle D^t(1),\,2D^t(t)+\tfrac13D^{\rm s}+R^x(1)+R^y(\nu)\big\rangle$,
$|\nu|\!\leqslant\!1$ $(\!{}\bmod G)$
(we replace the subalgebra~$\mathfrak s_{2.3}^\nu$ by the $G$-equivalent subalgebra $\tilde{\mathfrak s}_{2.3}^\nu$
for convenience of the reduction procedure):
\begin{gather*}
u=x\varphi+(y+\nu x)\ln|x|,\quad
\omega=y/x;\\
\big(2\omega(\omega^3-1)\varphi_{\omega\omega}-2\omega^3+\nu\omega^2+1\big)\varphi_{\omega\omega\omega}
+(5\omega^3-1)\varphi_{\omega\omega}{}^{\!\!2}
-\omega(8\omega-3\nu)\varphi_{\omega\omega}+3\omega-2\nu=0.
\end{gather*}
For any values of $\nu$, reduced equation~2.3$^\nu$ can be represented in normal form
since the coefficient of~$\varphi_{\omega\omega\omega}$ in it does not vanish on its solutions.
Its maximal Lie invariance algebra is
the algebra $\mathfrak a_{2.3}=\langle\p_\varphi,\omega\p_\varphi\rangle$.
The corresponding Lie group~$G_{2.3}$ consists of the point transformations
$\tilde\omega=\omega$, $\tilde\varphi=\varphi+\tilde c_1\omega+\tilde c_2$,
where $\tilde c_1$ and~$\tilde c_2$ are arbitrary constants.
All the subalgebras~$\tilde{\mathfrak s}_{2.3}^\nu$ have the same normalizer
${\rm N}_{\mathfrak g}(\tilde{\mathfrak s}_{2.3}^\nu)=\langle D^t(1),2D^t(t)+\tfrac13D^{\rm s},R^x(1),R^y(1)\rangle$ in~$\mathfrak g$.
The vector fields
$D^t(1)$, $2D^t(t)+\tfrac13D^{\rm s}$, $R^x(1)$ and $R^y(1)$
from ${\rm N}_{\mathfrak g}(\tilde{\mathfrak s}_{2.3}^\nu)$
induce the Lie-symmetry vector fields
$0$, $-(\omega+\nu)\p_\varphi$, $\p_\varphi$ and $\omega\p_\varphi$
of reduced equation~2.3$^{\nu}$, respectively,
i.e., the algebra~$\mathfrak a_{2.3}$ is entirely induced,
and thus the entire group~$G_{2.3}$ is induced by the stabilizer of~$\tilde{\mathfrak s}_{2.3}^\nu$ in~$G$.

Using the algebraic method,
we compute the point symmetry group~$G_{2.3}^\nu$ of reduced equation~2.3$^\nu$
for any fixed value of~$\nu$.
Let $\Phi$: $\tilde\omega=\Omega(\omega,\varphi)$, $\tilde\varphi=F(\omega,\varphi)$
with $\Omega_\omega F_\varphi-\Omega_\varphi F_\omega\ne0$
be a point symmetry transformation of this equation.
The condition $\Phi_*\mathfrak a_{2.3}\subseteq\mathfrak a_{2.3}$,
implies the equations $\Omega_\varphi=0$, $F_\varphi=a_{11}+a_{12}\Omega$ and $\omega F_\varphi=a_{21}+a_{22}\Omega$,
where $a_{11}$, $a_{12}$, $a_{21}$ and~$a_{22}$ are constants with $a_{11}a_{22}-a_{12}a_{21}\ne0$.
Therefore,
\[
\Omega=\frac{-a_{11}\omega+a_{21}}{a_{12}\omega-a_{22}},\quad
F=-\frac{a_{11}a_{22}-a_{12}a_{21}}{a_{12}\omega-a_{22}}\varphi+f(\omega)
\]
with a function~$f$ of~$\omega$.
Taking into account the derived form for~$\Phi$,
we continue the computation with the direct method.
As a result, we obtain a cumbersome overdetermined system of determining equations
for the parameters~$a_{11}$, $a_{12}$, $a_{21}$, $a_{22}$ and~$f$,
whose solution depends on the value of~$\nu$.
For $\nu\ne\pm1$, we obtain that
$a_{12}=a_{21}=0$, $a_{11}=a_{22}=1$ and $f=\tilde c_1\omega+\tilde c_2$
with arbitrary constants~$\tilde c_1$ and~$\tilde c_2$.
In other words, the complete point symmetry group~$G_{2.3}^\nu$ of reduced equation~2.3$^\nu$
with $\nu\ne\pm1$ coincides with the common Lie symmetry group~$G_{2.3}$.
For $\nu=\pm1$, we have additional solutions,
$a_{11}=a_{22}=0$, $a_{12}=a_{21}=\nu$, $f=-(\nu+\omega^{-1})\ln|\omega|+\tilde c_1\omega^{-1}+\tilde c_2$,
where~$\tilde c_1$ and~$\tilde c_2$ are again arbitrary constants.
Hence the complete point symmetry group~$G_{2.3}^\nu$ of reduced equation~2.3$^\nu$
with $\nu=\pm1$ is generated by the elements of the common Lie symmetry group~$G_{2.3}$
and the discrete point symmetry transformation
$\tilde\omega=\omega^{-1}$, $\tilde\varphi=\nu\omega^{-1}\varphi-(\nu+\omega^{-1})\ln|\omega|$,
which is induced by the discrete point symmetries~$\mathscr J$ and $\mathscr J\circ\mathscr I^{\rm s}$
of the original equation~\eqref{eq:dN}
if $\nu=1$ and $\nu=-1$, respectively.
Therefore, for any value of~$\nu$
the group~$G_{2.3}^\nu$ is entirely induced by the stabilizer of~$\tilde{\mathfrak s}_{2.3}^\nu$ in~$G$.

For each value of $\nu$, the subalgebra~$\tilde{\mathfrak s}_{2.3}^\nu$
has, up to the $G_{\rm L}$-equivalence, a single counterpart among subalgebras
of the algebra~$\mathfrak g_{\rm L}$,
$\bar{\mathfrak s}_{2.3}^\nu=\big\langle\bar D^t(1),\,2\bar D^t(t)+\tfrac13\bar D^{\rm s}+\bar R^x(1)+\bar R^y(\nu)\big\rangle$.
This is why ansatz~2.3 is extended to~$v$ in a unique way up to the $G_{\rm L}$-equivalence
as \smash{$v=|x|^{1/2}\psi$},
and the nonlinear Lax representation~\eqref{eq:dNLaxPair} reduces to the system
\begin{gather}\label{eq:dNReduction2.3LaxPairGen}
\begin{split}&
16\varepsilon\omega(\omega^3-1)\psi_\omega^{\,\,3}
-12\varepsilon(\omega^3-1)\psi\psi_\omega^{\,\,2}
-24(\nu\omega^2-1)\psi\psi_\omega
+\varepsilon\omega\psi^3+12\nu\omega\psi=0,
\\[.5ex]&
\omega\varphi_{\omega\omega}
+\varepsilon\omega\psi_\omega^{\,\,2}
-\frac\varepsilon2\psi\psi_\omega-1=0,
\end{split}
\end{gather}
where $\varepsilon:=\sgn x$.

Reduced equation~2.3$^\nu$ is an Abel equation of the second kind with respect to~$\varphi_{\omega\omega}$.
Its particular solution $\varphi=\omega\ln|\omega|$ corresponds to
a solution of the system~\eqref{eq:dNReduction2.3LaxPairGen} with $\psi=0$ and
the trivial solution $u=y\ln|y|+\nu x\ln|x|$ of~\eqref{eq:dN} from the family~\eqref{eq:dNTrivialSolutions}.
The differential substitution $\varphi_{\omega\omega}=p+\omega^{-1}$
maps reduced equation~2.3$^\nu$ to the simpler Abel equation of the second kind
$\big(2\omega(\omega^3-1)p+\nu\omega^2-1\big)p_\omega+(5\omega^3-1)p^2+3\nu\omega p=0$.

\medskip\par\noindent{\bf 2.4.}\ 
$\tilde{\mathfrak s}_{2.4}=\big\langle D^t(1),\,3D^t(t)+Z(1)\big\rangle$
(we replace the subalgebra~$\mathfrak s_{2.4}$ by the $G$-equivalent subalgebra $\tilde{\mathfrak s}_{2.4}$
for convenience of the reduction procedure):
\begin{gather*}
u=\varphi+\ln|x|,\quad
\omega=y/x;\\[1ex]
\big(2\omega(\omega^3-1)\varphi_{\omega\omega}+(3\omega^3-1)\varphi_\omega-\omega^2\big)\varphi_{\omega\omega\omega}\\\qquad{}
+2(5\omega^3-1)\varphi_{\omega\omega}{}^{\!\!2}
+2\omega(11\omega\varphi_\omega-3)\varphi_{\omega\omega}
+2(5\omega\varphi_\omega-2)\varphi_\omega=0.
\end{gather*}
The normalizer of the subalgebra~$\tilde{\mathfrak s}_{2.4}$ in~$\mathfrak g$ is
${\rm N}_{\mathfrak g}(\tilde{\mathfrak s}_{2.4})=\langle D^t(1),D^t(t),Z(1)\rangle$.
The Lie-symmetry vector fields $D^t(1)$, $3D^t(t)+Z(1)$ and~$Z(1)$ of the equation~\eqref{eq:dN}
induce the Lie-symmetry vector fields 0, 0 and~$\p_\varphi$
of reduced equation 2.4, respectively.
This equation can be represented in normal form
since the coefficient of~$\varphi_{\omega\omega\omega}$ in it
does not vanish on its solutions.
The maximal Lie invariance algebra~$\mathfrak a_{2.4}$ of reduced equation~2.4
is one-dimensional, $\mathfrak a_{2.4}=\langle\p_\varphi\rangle$,
and thus it is entirely induced by ${\rm N}_{\mathfrak g}(\tilde{\mathfrak s}_{2.4})$.

We compute the point symmetry group~$G_{2.4}$ of reduced equation~2.4 by the algebraic method.
Let $\Phi$: $\tilde\omega=\Omega(\omega,\varphi)$, $\tilde\varphi=F(\omega,\varphi)$
with $\Omega_\omega F_\varphi-\Omega_\varphi F_\omega\ne0$
be a point symmetry transformation of this equation.
From the condition $\Phi_*\mathfrak a_{2.4}\subseteq\mathfrak a_{2.4}$,
we derive the equations $\Omega_\varphi=0$ and $F_{\varphi\varphi}=0$,
which mean that $\Omega=\Omega(\omega)$ and $F=a\varphi+f(\omega)$
with $\Omega_\omega\ne0$, a nonzero constant~$a$ and a~function~$f$ of~$\omega$.
Taking into account the derived form for~$\Phi$,
we continue the computation with the direct method
and obtain a cumbersome overdetermined system of determining equations for the parameters $\Omega$, $a$ and~$f$,
which can nevertheless be solved, giving $a=1$ and
either $\Omega=\omega$ and $f=c$
or $\Omega=\omega^{-1}$ and $f=\ln|\omega|+c$ with an arbitrary constant~$c$.
Therefore, the group~$G_{2.4}$ is generated by the one-parameter subgroup of the shifts with respect to~$\varphi$
and the discrete point symmetry transformation $\tilde\omega=\omega^{-1}$, $\tilde\varphi=\varphi+\ln|\omega|$.
The last transformation is induced by the permutation~$\mathscr J$ of the variables~$x$ and~$y$
in the original equation~\eqref{eq:dN}.
Therefore, the group~$G_{2.4}$ is entirely induced by the stabilizer of~$\tilde{\mathfrak s}_{2.4}^\nu$ in~$G$.

The subalgebra~$\tilde{\mathfrak s}_{2.4}$ has the family of $G_{\rm L}$-inequivalent counterparts
$\bar{\mathfrak s}_{2.4}^{\varepsilon\nu}=\big\langle\bar D^t(1),\,3\bar D^t(t)+\bar Z(\varepsilon)+\nu\bar P^v\big\rangle$
among subalgebras of the algebra~$\mathfrak g_{\rm L}$.
Here $\varepsilon=\pm1$ and $\nu\geqslant0$ ($\!{}\bmod G_{\rm L}$).%
\footnote{%
The discrete point symmetry $\mathscr I^{\rm s}:=\mathscr D^{\rm s}(-1)$ of the equation~\eqref{eq:dN},
which alternates the signs of $(x,y,u)$ and of the vector fields $Z(\sigma)$ from the algebra~$\mathfrak g$,
has no counterpart among point symmetries of the nonlinear Lax representation~\eqref{eq:dNLaxPair}.
As a result, in contrast to the subalgebra~$\tilde{\mathfrak s}_{2.4}$,
the parameter~$\varepsilon$ in the subalgebra family~$\{\bar{\mathfrak s}_{2.4}^{\varepsilon\nu}\}$
can be gauged, up to the $G_{\rm L}$-equivalence, merely to $\pm1$ but not to~1.
}
This results in a family of $G_{\rm L}$-inequivalent extensions of ansatz~2.4 to~$v$
that are parameterized by~$\varepsilon$ and~$\nu$,
$u=\varphi+\varepsilon\ln|x|$, $v=\psi+\nu\ln|x|$.
The corresponding reduced systems are
\begin{gather}\label{eq:dNReduction2.4LaxPair}
\begin{split}&
3\big((\omega^3\!+1)\psi_\omega-\nu\omega^2\big)\varphi_\omega
-2\omega(\omega^3\!-1)\psi_\omega^{\,\,3}
+3\nu(\omega^3\!-1)\psi_\omega^{\,\,2}-3\varepsilon\omega\psi_\omega-\nu(\nu^2\!-3\varepsilon)\omega=0,\!
\\&
\omega\varphi_{\omega\omega}+\varphi_\omega
+\omega\psi_\omega^{\,\,2}-\nu\psi_\omega=0.
\end{split}
\end{gather}

The condition of vanishing the coefficient of~$\varphi_\omega$ in the first equation of~\eqref{eq:dNReduction2.4LaxPair}
is consistent with~\eqref{eq:dNReduction2.4LaxPair} only if $\nu=0$ and thus $\psi_\omega=0$,
which implies in view of the second equation of~\eqref{eq:dNReduction2.4LaxPair} that
$\omega\varphi_{\omega\omega}+\varphi_\omega=0$.
The associated family of particular solutions $\varphi=c_1\ln|\omega|+c_2$ of reduced equation~2.4,
which are parameterized by the arbitrary constants~$c_1$ and~$c_2$,
corresponds to the subfamily of
the trivial solutions $u=c_1\ln|y|+(c_1+1)\ln|x|+c_2$ from the family~\eqref{eq:dNTrivialSolutions}.

Further $(\omega^3+1)\psi_\omega-\nu\omega^2\ne0$.
Solving the first equation of~\eqref{eq:dNReduction2.4LaxPair} with respect to~$\varphi_\omega$
and excluding~$\varphi$ from the second equation of~\eqref{eq:dNReduction2.4LaxPair},
we derive a first-order ordinary differential equation with respect to~$\zeta:=\psi_\omega$,
\begin{gather}\label{eq:dNReduction2.4EqForPsi}
\begin{split}&
\big(4\omega(\omega^6-1)\zeta^3-3\nu(\omega^3-1)(3\omega^3+1)\zeta^2+6\nu^2\omega^2(\omega^3-1)\zeta
-\nu\omega(\nu^2\omega^3+\nu^2-3\varepsilon)\big)\zeta_\omega\\&
\qquad{}+(7\omega^6+18\omega^3-1)\zeta^4-6\nu\omega^2(3\omega^3+5)\zeta^3
+3\omega(5\nu^2\omega^3+3\nu^2+3\varepsilon)\zeta^2\\&
\qquad{}-2\nu(2\nu^2\omega^3-\nu^2+3\varepsilon)\zeta=0.
\end{split}
\end{gather}
Let $\nu=0$ and thus $\zeta\ne0$.
Then the equation~\eqref{eq:dNReduction2.4EqForPsi} reduces to the simple Bernoulli equation
\begin{gather*}
\begin{split}&
4\omega(\omega^6-1)\zeta\zeta_\omega+(7\omega^6+18\omega^3-1)\zeta^2+9\varepsilon\omega=0,
\end{split}
\end{gather*}
which integrates to $\zeta=\pm(\omega^3-1)^{-1}\sqrt{\tilde c_1|\omega|^{-1/2}(\omega^3+1)+3\varepsilon\omega}$.
The first equation of~\eqref{eq:dNReduction2.4LaxPair} with this value of~$\zeta$ implies that
\begin{gather*}
\varphi=\frac{\varepsilon}3\ln|\omega^3-1|+c_1\ln\left|\frac{|\omega|^{3/2}+1}{|\omega|^{3/2}-1}\right|+c_2
\quad\mbox{for}\quad \omega\geqslant0,
\\
\varphi=\frac{\varepsilon}3\ln|\omega^3-1|+c_1\arctan|\omega|^{3/2}+c_2\quad\mbox{for}\quad \omega\leqslant0,
\end{gather*}
where $c_1$ and~$c_2$ are arbitrary constants,
and the constants~$c_2$ and $\varepsilon$ can be set to 0 and~1 up to the $G_{2.4}$-equivalence, respectively.
The corresponding solutions of the equation~\eqref{eq:dN} are
\begin{gather*}
\solution
u=\frac13\ln|y^3-x^3|+c_1\ln\left|\frac{|x|^{3/2}+|y|^{3/2}}{|x|^{3/2}-|y|^{3/2}}\right|
\quad\mbox{for}\quad xy\geqslant0,
\\
\solution
u=\frac13\ln|y^3-x^3|+c_1\arctan\left|\frac yx\right|^{3/2}\quad\mbox{for}\quad xy\leqslant0.
\end{gather*}

For $\nu=\sqrt3$, the equation~\eqref{eq:dNReduction2.4EqForPsi}
can be integrated implicitly,
\[
4\ln\big|(\omega^3-1)\zeta-\sqrt3\omega^2\big|
-2\ln\big|(\omega^3+1)\zeta-\sqrt3\omega^2\big|
+\ln\big|\omega\zeta-\sqrt3\big|+\ln\big|\zeta\big|=c_1.
\]
Then the function~$\varphi$ is defined by the first equation of~\eqref{eq:dNReduction2.4LaxPair}
with $\psi_\omega=\zeta$ and $\nu=\sqrt3$.

\medskip\par\noindent{\bf 2.13.}\
$\mathfrak s_{2.13}=\big\langle D^t(1),\,D^{\rm s}\big\rangle$:
\quad $u=x^3\varphi,\quad \omega=y/x;$
\begin{gather*}
\big(2\omega(\omega^3-1)\varphi_{\omega\omega}
-2(3\omega^3-1)\varphi_\omega+6\omega^2\varphi\big)
\varphi_{\omega\omega\omega}\\\qquad{}
-(5\omega^3-1)\varphi_{\omega\omega}^{\,\,2}
+2\omega(11\omega\varphi_\omega-9\varphi)\varphi_{\omega\omega}
-4(5\omega\varphi_\omega-6\varphi)\varphi_\omega
=0.\!\!
\end{gather*}
The normalizer of the subalgebra~$\mathfrak s_{2.13}$ in~$\mathfrak g$ is
${\rm N}_{\mathfrak g}(\mathfrak s_{2.13})=\langle D^t(1),D^t(t),D^{\rm s}\rangle$,
and the entire maximal Lie invariance algebra $\mathfrak a_{2.13}=\langle \varphi\p_\varphi\rangle$
of reduced equation~2.13 is induced by this normalizer.

Reduced equation~2.13 can be included in the family of reduced equations~2.1$^{\kappa}$
as the element with $\kappa=3$, which corresponds to the limit values $\lambda=\pm\infty$.

Simple solutions of reduced equation~2.13,
$\varphi=|\omega|^{3/2}$ and the family of solutions that are cubic polynomials in~$\omega$,
were found in~\cite{moro2021a}, see the equations~(24) and~(25) therein, respectively.
The solution $u=|xy|^{3/2}$ of the original equation~\eqref{eq:dN},
which corresponds to the solution $\varphi=|\omega|^{3/2}$,
is essentially generalized in Section~\ref{sec:dNMultiplicativeSeparationOfVars}
using multiplicative separation of the variables~$x$ and~$y$.
The above family of polynomial solutions in~$\omega$ is associated
with the family of solutions of~\eqref{eq:dN}
that are homogeneous cubic polynomials in~$(x,y)$ with constant coefficients.

\section{Lie reductions to algebraic equations}\label{sec:dNLieReductionsOfCodim3}

For any three-dimensional subalgebra~$\mathfrak s_3$ of the algebra~$\mathfrak g$,
either its rank~$r$ is less than three
and thus it cannot be used for Lie reduction of the equation~\eqref{eq:dN}
to an algebraic equation
or all the corresponding invariant solutions are, up to the $G$-equivalence,
just particular elements of parameterized families of solutions that have been constructed
in Sections~\ref{sec:dNLieReductionsOfCodim1} and~\ref{sec:dNLieReductionsOfCodim2}.
To show this, we present an outline of the classification
of three-dimensional subalgebras of~$\mathfrak g$.

Consider a three-dimensional subalgebra $\mathfrak s_3=\langle Q^i,\,i=1,2,3\rangle$ of~$\mathfrak g$
spanned by three (linearly independent) vector fields
\begin{gather*}
Q^i=D^t(\tau^i)+\lambda^iD^{\rm s}+P^x(\chi^i)+P^y(\rho^i)+R^x(\alpha^i)+R^y(\beta^i)+Z(\sigma^i)
\end{gather*}
from~$\mathfrak g$
with arbitrary smooth functions $\tau^i$, $\chi^i$, $\rho^i$, $\alpha^i$, $\beta^i$ and $\sigma^i$ of $t$
and arbitrary constants $\lambda^i$ such that
the tuples $(\tau^i,\lambda^i,\chi^i,\rho^i,\alpha^i,\beta^i,\sigma^i)$ are linearly independent.
Here and in what follows the index~$i$ runs from~1 to~3.
The consideration splits into cases mainly depending on two values,
$k_1=k_1(\mathfrak s_3):=\dim\langle\tau^i\rangle$ and $k_2=k_2(\mathfrak s_3):=\dim\langle(\tau^i,\lambda^i)\rangle$.
For brevity, we use transitions to $G$-equivalent subalgebras, basis changes
and hints from the proofs of Lemmas~\ref{lem:dN1DInequivSubalgs} and~\ref{lem:dN2DInequivSubalgs}
without referring to this.
Below, $\kappa_1$, $\kappa_2$, $\kappa_3$ and~$\nu$ denote constants.

\medskip\par\noindent$k_1=3$.\
In view of the classical Lie theorem on Lie algebras of vector fields on the real line,%
\footnote{%
See \cite{bihlo2017a,boyk2021a,boyk2024b,boyk2015a,kuru2018a,kuru2020a,opan2017a} and references therein
for applications of this theorem to classifying subalgebras of various algebras of vector fields.
}
we can set $\tau^1=1$, $\tau^2=t$ and $\tau^3=t^2$,
which implies $\mathfrak s_3\subset\mathfrak g'$, and thus $\lambda^i=0$.
The vector fields~$Q^1$ and~$Q^2$ reduce to $D^t(1)$ and $D^t(t)+Z(\delta)$ with $\delta\in\{0,1\}$,
respectively.
We successively derive from the commutation relations $[Q^1,Q^3]=2Q^2$ and $[Q^2,Q^3]=Q^3$
that $\chi^3,\rho^3,\alpha^3,\beta^3=\const$, $\sigma^3_t=\delta$
and hence $\chi^3,\rho^3,\alpha^3,\beta^3=0$.
Therefore, $r=2$.

\medskip\par\noindent$k_1=2$, $k_2=3$.\
We can make $\tau^1=1$, $\tau^2=t$, $\tau^3=0$, $\lambda^1=\lambda^2=0$, $\lambda^3=1$,
and then $Q^3=D^{\rm s}$.
The commutation relations $[Q^1,Q^3]=[Q^2,Q^3]=0$ imply $Q^1=D^t(1)$ and $Q^2=D^t(t)$,
i.e., $r=2$.

\medskip\par\noindent$k_1=k_2=2$.\
Setting $Q^1=D^t(1)$, $\tau^2=t$, $\tau^3=0$, $\lambda^3=0$,
we derive from the commutation relations
$[Q^j,Q^3]=\kappa_jQ^3$, $j=1,2$, $[Q^1,Q^2]=Q^1+\kappa_3Q^3$ that
$\chi^3_t=\kappa_1\chi^3$, $t\chi^3_t=\kappa_2\chi^3$,
$\rho^3_t=\kappa_1\rho^3$, $t\rho^3_t=\kappa_2\rho^3$,
and thus, if $r=3$, $(\chi^3,\rho^3)\ne(0,0)$, $\kappa_1=\kappa_2=0$,
which further implies that $\chi^3_t=\rho^3_t=\alpha_t=\beta_t=\sigma_t=0$.
In other words, the subalgebra~$\mathfrak s_3$ contains, up to $G$-equivalence,
a subalgebra from the family~$\{\mathfrak s_{2.14}^{0\nu\delta'}\}$.

\medskip\par\noindent$k_1=1$, $k_2=2$.\
We make $\tau^1=1$, $\tau^2=\tau^3=0$, $\lambda^1=\lambda^3=0$, $\lambda^2=1$
and then $Q^2=D^{\rm s}$.
The commutation relations $[Q^1,Q^2]=0$, $[Q^j,Q^3]=\kappa_jQ^3$, $j=1,2$,
imply $Q^1=D^t(1)$ and, if $r=3$, then $\kappa_2=-1$ and $Q^3=P^x({\rm e}^{\kappa_1t})+\nu P^y({\rm e}^{\kappa_1t})$,
i.e., the subalgebra~$\mathfrak s_3$ contains a subalgebra that is $G$-equivalent to one
from the family~$\{\mathfrak s_{2.9}^{\tilde\rho}\}$.

\medskip\par\noindent$k_2\leqslant1$.\
If $r=3$, then up to $G$-equivalence, the subalgebra~$\mathfrak s_3$ contains
a subalgebra from the family~$\{\mathfrak s_{2.17}^{\rho\alpha\beta}\}$
and, therefore, a subalgebra from the family~$\{\mathfrak s_{1.4}^\beta\}$.

\medskip\par
As a result, we conclude that Lie reductions of the equation~\eqref{eq:dN} to algebraic equations
give no new $G$-equivalent solutions in comparison with those
that have been constructed in a closed explicit form
in Sections~\ref{sec:dNLieReductionsOfCodim1} and~\ref{sec:dNLieReductionsOfCodim2}.


\section{Multiplicative separation of variables}\label{sec:dNMultiplicativeSeparationOfVars}

The equation~\eqref{eq:dN} is identically satisfied
under the additive separation of the variables~$x$ and~$y$,
and the solutions from the corresponding family~\eqref{eq:dNTrivialSolutions} are trivial.

Consider solutions of the equation~\eqref{eq:dN}
with nontrivial multiplicative separation of the variables~$x$ and~$y$.
They are represented in the form $u=\varphi(t,x)\psi(t,y)$ with $\varphi_x\ne0$ and~$\psi_y\ne0$.%
\footnote{%
The functions $\varphi$ and~$\psi$ are defined up to the transformations
$\tilde\varphi=\varphi/f$, $\tilde\psi=f\psi$ with an arbitrary nonzero function of~$t$.
If $\varphi_x=0$ or~$\psi_y=0$, then one can set $\varphi=1$ or~$\psi=1$, respectively,
and thus the separation of the variables~$x$ and~$y$ is trivial;
moreover, then the corresponding solutions belong to the family of trivial solutions~\eqref{eq:dNTrivialSolutions}.
}

Substituting the multiplicative ansatz $u=\varphi(t,x)\psi(t,y)$ into the equation~\eqref{eq:dN}
and separating the variables~$x$ and~$y$, we obtain the equation
\[
\frac{\varphi_{tx}}{\varphi_x}+\frac{\psi_{ty}}{\psi_y}
=\frac{(\varphi_{xx}\varphi_x)_x}{\varphi_x}\psi
+\frac{(\psi_{yy}\psi_y)_y}{\psi_y}\varphi,
\]
which we further simultaneously differentiate with respect~$x$ and~$y$ and derive
\[
\frac1{\varphi_x}\left(\frac{(\varphi_{xx}\varphi_x)_x}{\varphi_x}\right)_x
+\frac1{\psi_y}\left(\frac{(\psi_{yy}\psi_y)_y}{\psi_y}\right)_y=0.
\]
These two equations imply that
\[
\frac{(\varphi_{xx}\varphi_x)_x}{\varphi_x}=\alpha\varphi+\beta, \quad
\frac{\varphi_{tx}}{\varphi_x}=\gamma\varphi+\delta
\quad\mbox{and}\quad
\frac{(\psi_{yy}\psi_y)_y}{\psi_y}=-\alpha\psi+\gamma,\quad
\frac{\psi_{ty}}{\psi_y}=\beta\psi-\delta
\]
for some sufficiently smooth functions~$\alpha$, $\beta$, $\gamma$ and~$\delta$ of~$t$.
These systems with respect to~$\varphi$ and~$\psi$ integrate~to
\begin{gather}\label{eq:dNMultiplicativeSeparationOfVarsSys1Int}
\varphi_x^{\,\,3}=\frac\alpha2\varphi^3+\frac32\beta\varphi^2+\zeta^1\varphi+\zeta^0,\quad
\varphi_t=\frac\gamma2\varphi^2+\delta\varphi+\zeta^2,
\\\label{eq:dNMultiplicativeSeparationOfVarsSys2Int}
\psi_y^{\,\,3}=-\frac\alpha2\psi^3+\frac32\gamma\psi^2+\theta^1\psi+\theta^0,\quad
\psi_t=\frac\beta2\psi^2-\delta\psi+\theta^2,
\end{gather}
where $\zeta^0$, $\zeta^1$, $\zeta^2$, $\theta^0$, $\theta^1$ and $\theta^2$ are also sufficiently smooth functions of~$t$,
and for solutions to be nontrivial, we should impose the conditions that the tuples
$(\alpha,\beta,\zeta^1,\zeta^0)$ and $(\alpha,\gamma,\theta^1,\theta^0)$ are nonzero.
Due to the indeterminacy of $(\varphi,\psi)$, we set $\delta=0$ without loss of generality.

We exclude the derivatives of~$\varphi$ and~$\psi$
in view of the systems~\eqref{eq:dNMultiplicativeSeparationOfVarsSys1Int} and~\eqref{eq:dNMultiplicativeSeparationOfVarsSys2Int}
from their compatibility conditions $(\varphi_x)_t=(\varphi_t)_x$ and $(\psi_y)_t=(\psi_t)_y$
and split the obtained equalities with respect to~$\varphi$ and~$\psi$,
which gives the following systems for the parameter functions depending on~$t$:
\begin{gather*}
\alpha\beta=\alpha\gamma=\beta\gamma=0, \quad
\alpha_t=0,\quad
\beta_t=\frac53\gamma\zeta^1-\alpha\zeta^2,\quad 
\gamma_t=\frac53\beta\theta^1+\alpha\theta^2,\\ 
\zeta^1_t=3\gamma\zeta^0-3\beta\zeta^2,\quad 
\theta^1_t=3\beta\theta^0-3\gamma\theta^2,\quad 
\zeta^0_t=-\zeta^1\zeta^2,\quad 
\theta^0_t=-\theta^1\theta^2. 
\end{gather*}
Consider possible cases separately.

1. $\alpha\ne0$. Then $\alpha=\const$, $\beta=\gamma=\zeta^2=\theta^2=0$,
and thus $\zeta^0$, $\zeta^1$, $\theta^0$ and $\theta^1$ are constants.
Integrating the systems~\eqref{eq:dNMultiplicativeSeparationOfVarsSys1Int}
and~\eqref{eq:dNMultiplicativeSeparationOfVarsSys2Int} with these parameters' values
and simplifying the result by transformations from~$G$,
we derive a family of $G$-inequivalent solutions of the equation~\eqref{eq:dN}
that generalizes the solutions~\eqref{eq:s2.2InvSolutions},
\[
\solution
u=\varphi(x)\psi(y),\quad
\int\frac{{\rm d}\varphi}{(\varphi^3+c_1\varphi+c_2)^{1/3}}=x,\quad
\int\frac{{\rm d}\psi}{(\psi^3+c_3\psi+c_4)^{1/3}}=-y.
\]
Up to the $G$-equivalence,
one of the constants~$c_1$ and~$c_3$, if it is nonzero, can be set to be equal $\pm1$
or one of the constants~$c_2$ and~$c_4$, if it is nonzero, can be set to be equal one.
Both quadratures here are the same as in~\eqref{eq:ReducedEq2.2_0ImplicitExpression}.
Hence they can be computed explicitly for certain values of the tuples $(c_1,c_2)$ and $(c_3,c_4)$,
see~\eqref{eq:ReducedEq2.2_0ImplicitExpression2}.

2. $\beta\ne0$. Then $\alpha=\gamma=0$, and thus $\theta^0=\theta^1=0$,
which contradicts the nontriviality condition~$\psi_y\ne0$.
The case $\gamma\ne0$ reduces to the case $\beta\ne0$
by permutation of~$x$ and~$y$.

3. $\alpha=\beta=\gamma=0$,
and thus $\zeta^1$ and $\theta^1$ are constants,
$\zeta^0=-\int\zeta^1\zeta^2{\rm d}t$ and $\theta^0=-\int\theta^1\theta^2{\rm d}t$.
Rearranging the solution sets of the systems~\eqref{eq:dNMultiplicativeSeparationOfVarsSys1Int}
and~\eqref{eq:dNMultiplicativeSeparationOfVarsSys2Int}
with these parameters' values up to the $G$-equivalence and
in view of the indeterminacy of $(\varphi,\psi)$,
we construct the solutions of the equation~\eqref{eq:dN} of the form
\begin{gather*}
\solution
u=\big(|x|^{3/2}+\zeta(t)\big)\big(|y|^{3/2}+\theta(t)\big),\qquad\qquad
\solution
u=\big(x+\zeta(t)\big)|y|^{3/2}
\end{gather*}
and the solution $u=xy$, which belongs to the family~\eqref{eq:InvSolutions1.4}.
Here $\zeta$ and~$\theta$ are arbitrary sufficiently smooth functions of~$t$.
The first and the second families of solutions generalize
the $\mathfrak s_{2.13}$-invariant solution $u=|xy|^{3/2}$ and
\smash{$\mathfrak s_{2.1}^{5/3}$}-invariant solution $u=|x|^{3/2}y$,
see~\cite[Eq.~(26)]{moro2021a} and
the last paragraph related to reduction~2.1 in Section~\ref{sec:dNLieReductionsOfCodim2Collection2},
respectively.

\begin{remark}
For any $\nu$, the $\mathfrak s_{2.2}^\nu$-invariant solutions can be interpreted
as those with multiplicative separation of variables after their linear change.
Following the consideration in this section,
one can try to carry out a comprehensive study of such separation of variables.
Maybe, the most interesting is the multiplicative separation of the variables
$\tilde x=x+y$ and $\tilde y=x-y$, cf.\ the last (parametric) solution
obtained by reduction~2.2$^1$.
\end{remark}


\section{Conclusion}\label{sec:Conclusion}

In the present paper, we have constructed wide families of new exact invariant solutions
of the dispersionless Nizhnik equation~\eqref{eq:dN}
in closed form in terms of elementary, Lambert and hypergeometric functions
as well as in parametric or implicit form.
The main tool for this purpose was the optimized procedure of Lie reduction.
A rigorous description and a proper substantiation of this procedure
is in fact the main attainment of the paper.

Using the results of~\cite{boyk2024a} on
the maximal Lie invariance algebras~$\mathfrak g$ and~$\mathfrak g_{\rm L}$
of the equation~\eqref{eq:dN} and of its nonlinear Lax representation~\eqref{eq:dNLaxPair}
and their point-symmetry pseudogroups~$G$ and~$G_{\rm L}$,
we have classified
one- and two-dimensional subalgebras of the algebra~$\mathfrak g$
and one-dimensional subalgebras of the algebra~$\mathfrak g_{\rm L}$
up to the $G$- and $G_{\rm L}$-equivalences, respectively.
We could only classify subalgebras that are appropriate for Lie reduction
but this would not result in an essential simplification
in comparison with the classification of all one- and two-dimensional subalgebras
and the further selection of the appropriate ones among the listed inequivalent subalgebras.
Instead of the standard equivalences
within the algebras~$\mathfrak g$ and~$\mathfrak g_{\rm L}$ up to their inner automorphisms,
which coincide with the $G_{\rm id}$- and $G_{\rm L, id}$-equivalences,
where $G_{\rm id}$- and $G_{\rm L, id}$ are the identity components of~$G$ and $G_{\rm L}$, respectively,
we have used the stronger $G$- and $G_{\rm L}$-equivalences.
In this way, we have also taken into account the discrete point symmetry transformations of~\eqref{eq:dN},
which has allowed us to reduce the optimal lists of subalgebras.
Moreover, as explained in Section~\ref{sec:LieReductionProcedure},
it has also made the Lie reduction procedure consistent
with the natural $G$-equivalence on the solution set of the equation~\eqref{eq:dN}.
The above arguments clearly confirm
that the correct computation of~$G$ and $G_{\rm L}$ in~\cite{boyk2024a} was important.
Note that in fact the algebras~$\mathfrak g$ and~$\mathfrak g_{\rm L}$ are
infinite-dimensional Lie pseudoalgebras of vector fields.
In general, the classification of (low-dimensional) subalgebras of such an algebra
is complicated, in particular, by the necessity of considering
differential~\cite{fush1994b} or even functional~\cite{davi1986a,mart1989a} equations
in the course of this classification.

The algebra~$\mathfrak g$ is injectively mapped into the algebra~$\mathfrak g_{\rm L}$
via extending the vector fields from~$\mathfrak g$ to the variable~$v$.
The vector fields~\eqref{eq:dNMIA}, which span~$\mathfrak g$,
are extended trivially and formally coincide with their counterparts in~$\mathfrak g_{\rm L}$.
The only exception is the vector field~$D^{\rm s}$,
which extends to $\bar D^{\rm s}=x\p_x+y\p_y+3u\p_u+\tfrac32v\p_v$.
Moreover, $\mathfrak g_{\rm L}=\bar{\mathfrak g}\lsemioplus\langle\bar P^v\rangle$,
where $\bar{\mathfrak g}$ is the image of~$\mathfrak g$ under the above mapping,
and $\bar P^v=\p_v$.
Although the corresponding homomorphism%
\footnote{%
For this homomorphism and the isomorphism below,
we should replace~$G$ by its trivial prolongation to~$v$,
considering the restriction of elements of the prolongation
on open subsets of the space with the coordinates $(t,x,u,v)$.
}
of the pseudogroup~$G$ into the pseudogroup~$G_{\rm L}$ is not injective,
its kernel is generated by the discrete involution $\mathscr I^{\rm s}:=\mathscr D^{\rm s}(-1)$
from~$G$, which of course involves the restrictions of~$\mathscr I^{\rm s}$ as well,
and the quotient pseudogroups
$G/\{{\rm id},\mathscr I^{\rm s}\}$ and
$G_{\rm L}/\{\bar{\mathscr P}^v(B),\bar{\mathscr I}^v\circ\bar{\mathscr P}^v(B)\mid B\in\mathbb R\}$
are isomorphic, see the paragraph after Theorem~\ref{thm:dNCompletePointSymGroupOfLaxRepresentation}.
As a result, the classifications of one- and two-dimensional subalgebras of the algebra~$\mathfrak g_{\rm L}$
up to the $G_{\rm L}$-equivalence can be easily derived
from the respective classifications for the algebra~$\mathfrak g$ up to the $G$-equivalence,
cf.\ Lemmas~\ref{lem:dN1DInequivSubalgs} and~\ref{lem:dNLaxPair1DInequivSubalgs}
for the case of dimension one.
Nevertheless, we have not presented the classification of two-dimensional subalgebras of the algebra~$\mathfrak g_{\rm L}$
since we needed only a few of these subalgebras,
which are given directly when using them
for Lie reductions of the nonlinear Lax representation~\eqref{eq:dNLaxPair}
in Section~\ref{sec:dNLieReductionsOfCodim2Collection2}.
The correspondence between the equivalence classes of
one-dimensional (resp.\ two-dimensional) subalgebras of the algebras~$\mathfrak g$ and~$\mathfrak g_{\rm L}$
is injective but not one-to-one.
The list of inequivalent subalgebras of~$\mathfrak g$ of any fixed dimension can be trivially embedded in
the corresponding list for the algebra~$\mathfrak g_{\rm L}$
via the above extension of elements of~$\mathfrak g$ to the variable~$v$.
The bijection breaking is related to the disappearance of~$\mathscr I^{\rm s}$
and the appearance of~$\bar P^v$ in the course of the transition
from $(G,\mathfrak g)$ to $(G_{\rm L},\mathfrak g_{\rm L})$,
see the subalgebras $\bar{\mathfrak s}_{2.1'}^0$, \smash{$\bar{\mathfrak s}_{2.1'}^{2/3}$}
and $\bar{\mathfrak s}_{2.4}^{\varepsilon\nu}$
in Section~\ref{sec:dNLieReductionsOfCodim2Collection2}.
The last family of subalgebras is the most interesting
since, in contrast to the corresponding coefficient in the second basis vector field
of~$\tilde{\mathfrak s}_{2.4}$ and the $G$-equivalence,
the parameter~$\varepsilon$ in $\bar{\mathfrak s}_{2.4}^{\varepsilon\nu}$
cannot be set to~1 up to the $G_{\rm L}$-equivalence.\looseness=-1

The described relation between the lists of inequivalent subalgebras of~$\mathfrak g$ and of~$\mathfrak g_{\rm L}$
can be reformulated in terms of the relation between the corresponding collections of
inequivalent Lie reductions of the equation~\eqref{eq:dN} and of the nonlinear Lax representation~\eqref{eq:dNLaxPair}.

If subalgebras of~$\mathfrak g$ are $G$-equivalent,
then the corresponding reduced equations are necessarily similar with respect to point transformations of the invariant variables.
Consider a class~$\mathcal C$ of reduced equations for the equation~\eqref{eq:dN}
that is associated with a parameterized family~$\mathcal F$ of subalgebras of~$\mathfrak g$,
and thus the arbitrary elements of~$\mathcal C$ are expressed in terms of the subalgebra parameters.
Then the stabilizer of~$\mathcal F$ in~$G$ induces
a (pseudo)subgroup~$G^{\sim}_{\mathcal C,{\rm ind}}$
of the equivalence (pseudo)group~$G^{\sim}_{\mathcal C}$ of the class~$\mathcal C$.
The proper inclusion $G^{\sim}_{\mathcal C,{\rm ind}}\lneqq G^{\sim}_{\mathcal C}$,
which happens quite commonly, means that some elements of~$G^{\sim}_{\mathcal C}$
are not induced by transformations from~$G$,
and hence we call them \emph{hidden equivalence transformations} of the class~$\mathcal C$.
If the subalgebras from the family~$\mathcal F$ are $G$-inequivalent to each other,
then transformations from the group~$G$ can induce only point symmetries of equations from the class~$\mathcal C$
but not point transformations between different elements of this class.
At the same time, in the case of the presence of hidden equivalence transformations,
a wide subset of the action groupoid of~$G^{\sim}_{\mathcal C}$ can still be used
for mapping the class~$\mathcal C$ to its proper subclass~$\mathcal C'$,
which formally has less number of (significant) arbitrary elements.%
\footnote{%
See \cite{opan2022a} and references therein for mappings
between classes of differential equations that are generated by families of point transformations.
}
Then the correspondence between the parameters of the family~$\mathcal F$ and
the arbitrary elements of the subclass~$\mathcal C'$ is definitely not injective.
We can select ansatzes associated with subalgebras of~$\mathcal F$
such that the corresponding class of reduced equations is minimal
up to the described mappings by hidden equivalence transformations.
Nevertheless, this is not always convenient as shown by reductions~1.3 and~1.4.
The classes of reduced equations~2.5, 2.9 and~2.14 are also not minimal in the above sense.
In this context, the family~$\mathcal F$ of subalgebras $\{\mathfrak s_{1.3}^\rho\}$ is especially demonstrative.
After excluding the singular subalgebra $\mathfrak s_{1.3}^1$ from~$\mathcal F$
and properly modifying the corresponding ansatzes from Table~\ref{tab:LieReductionsOfCodim1},
we have derived the single simple reduced equation~\eqref{eq:dNs1.3RhoNe1ModifiedRedEq}
instead of a class of reduced equations with the functional parameter $\rho=\rho(t)$ of~$\mathcal F$
as its arbitrary element.\looseness=-1

We have paid considerable attention to the selection of optimal ansatzes
and thus simplified the further consideration
but the simplification is not as significant as, e.g.,
that achieved for the Navier--Stokes equations in~\cite{fush1994a,fush1994b,popo1995b}.
Most of the reduced equations for the equation~\eqref{eq:dN} are quite cumbersome,
and this is not the only feature of them that complicates
the computation of their Lie and discrete point symmetries.
Thus, each of reduced equations~1.1$^0$, 2.1$^\kappa$ (including~2.13), 2.2$^\nu$ and 2.14$^{0\nu\delta'}$
is not of maximal rank on the entire associated manifold in the corresponding jet space.%
\footnote{%
Each of the reductions 2.1$^\kappa$ (including~2.13), 2.2$^\nu$ and 2.14$^{0\nu\delta'}$
can be considered as a two-step Lie reduction with reduction~1.1$^0$ as its first step.}
Even if a reduced equation is of maximal rank,
it is not necessarily can be represented in the normal form,
see reduced equation~2.5$^{\lambda\mu}$ with $\lambda=2/3$.
As far as we know, Lie and general point symmetries of such unusual differential equations
have not been considered in the literature.
For some reductions of codimension two, even under the optimal choice of ansatzes,
the permutation~$\mathscr J$ of~$x$ and~$y$,
which is a simple and obvious discrete point symmetry of the equation~\eqref{eq:dN},
induces more complicated and nontrivial discrete point symmetries of the corresponding reduced equations,
and this leads to the complexity of general elements of the point symmetry groups
of some reduced ordinary differential equations.
There is no similar phenomenon for Lie symmetries of reduced equations obtained from the equation~\eqref{eq:dN}.
Nevertheless, we have comprehensively studied point symmetries and their induction
for all the reduced equations
selected in the course of applying the optimized Lie reduction procedure to the equation~\eqref{eq:dN}.
This study itself is a~necessary ingredient of the reduction procedure.
It has helped us
to cut down the number of Lie reductions to be considered and
to integrate or at least to lower the order of reduced ordinary differential equations.
Note that discrete symmetries of reduced equations have been computed for the first time in the present paper,
and in view of the above reasons such as
the complexity of reduced equations and their point symmetries and
the simplicity of their Lie-symmetry vector fields,
the algebraic method by Hydon and its various modifications
are especially efficient and convenient for this computation.

For finding exact solutions of reduced ordinary differential equations,
we have also used the associated reduced systems for the nonlinear Lax representation~\eqref{eq:dNLaxPair}.
Due to properly arranging the hierarchy of Lie reductions of the equation~\eqref{eq:dN}
and accurately selecting a low number of reduced ordinary differential equations to be integrated,
we were able to deeply analyze them and construct
wider families of exact solutions of the equation~\eqref{eq:dN}
than those presented in the literature.
Of course, there are a number of possibilities for extending and generalizing the results
of the present paper.
In particular, since most of Lie symmetries of the reduced equation~\eqref{eq:dNs1.3InvarSolutions}
are hidden for the original equation~\eqref{eq:dN},
one can actually represent more solutions from the family~\eqref{eq:dNs1.3RhoNe1ModifiedRedEq}
in an explicit form by means of Lie reductions of~\eqref{eq:dNs1.3InvarSolutions}
than those found in Section~\ref{sec:dNLieReductionsOfCodim2Collection1}.
In addition, the results of Section~\ref{sec:dNMultiplicativeSeparationOfVars}
on multiplicative separation of variables for the equation~\eqref{eq:dN}
and of~\cite{moro2021a} on solutions of~\eqref{eq:dN} that are polynomial in~$(x,y)$
show that more closed-form solutions of~\eqref{eq:dN} can be constructed
using other tools of symmetry analysis of differential equations.\looseness=-1

We have checked all the obtained solutions and have selected $G$-inequivalent ones among them.
Checking constructed solutions and their inequivalence to known solutions
is the last but most important step of any procedure of finding exact solutions
of differential equations.
Unfortunately, this step is commonly disregarded,
which led to many papers containing only incorrect or known solutions,
see the discussion of such papers in~\cite{kudr2009a,popo2010c}.

\section*{Acknowledgments}

The authors are grateful to Dmytro Popovych, Galyna Popovych and Artur Sergyeyev for helpful discussions and suggestions.
R.O.P. also expresses his gratitude for the hospitality shown by the University of Vienna during his long staying at the university.
This work was supported by a grant from the Simons Foundation (1290607, O.O.V., V.M.B.).
The work of R.O.P. was supported in part by the Ministry of Education, Youth and Sports of the Czech Republic (M\v SMT \v CR)
under RVO funding for I\v C47813059.
The authors express their deepest thanks to the Armed Forces of Ukraine and the civil Ukrainian people
for their bravery and courage in defense of peace and freedom in Europe and in the entire world from russism.

\noprint{
\section*{Declarations}

\subsection*{Ethical Approval}
Not applicable.

\subsection*{Conflict of interest}

The authors declare that they have no competing interests.

\subsection*{Authors' contributions}

The contribution of all three authors to the article is equal.

\subsection*{Funding}
This work was supported by a grant from the Simons Foundation (1030291, O.O.V., V.M.B.).
The work of R.O.P. was supported in part by the Ministry of Education, Youth and Sports of the Czech Republic (M\v SMT \v CR)
under RVO funding for I\v C47813059.

\subsection*{Data availability statement}
Data sharing not applicable to this article as no datasets were generated or analyzed during the current study.

}

\end{document}